\documentclass[11pt,fullpage,letterpaper]{article} 

\usepackage{natbib}

\setcitestyle{authoryear,round,citesep={;},aysep={,},yysep={;}}
\usepackage[english]{babel}

\usepackage[margin=1in]{geometry}
\usepackage{fancyhdr}
\usepackage[utf8]{inputenc} 
\usepackage[T1]{fontenc}    

\usepackage{titletoc}
\usepackage[toc, page, header]{appendix} 

\usepackage{url}            
\usepackage{booktabs}       
\usepackage{amsfonts}       
\usepackage{nicefrac}       
\usepackage{microtype}      
\usepackage{xcolor}         

\usepackage{parskip}

\title{Non-Monetary Mechanism Design without Priors:\\Achieving Efficiency via Adaptive Costly Audits}
\author{%
    Yan Dai~\thanks{Massachusetts Institute of Technology, Operations Research Center. Email: \texttt{yandai20@mit.edu}.}\and
    Mo\"ise Blanchard~\thanks{Georgia Tech, School of Industrial and Systems Engineering. Email: \texttt{mblanchard41@gatech.edu}.}\and
    Patrick Jaillet~\thanks{Massachusetts Institute of Technology, Department of EECS. Email: \texttt{jaillet@mit.edu}.}
}
\date{First version: Feb 2025. This version: May 2026.\footnote{A preliminary version of this work was accepted to the 38th Annual Conference on Learning Theory (COLT 2025). Compared to that version: We extended our \mechname and its analysis to adversarial and stochastic audit models. The main body and technical sections have been thoroughly revised for enhanced clarity and expository depth. The literature overview has been enriched to include the research on mechanism design with costly state verification (CSV).}}

\usepackage{multirow}
\usepackage{arydshln}
\usepackage{footnote}
\usepackage{enumitem}
\setitemize{leftmargin=*, nosep, itemsep=2pt, parsep=3pt}
\setenumerate{leftmargin=*, nosep, itemsep=2pt, parsep=3pt}

\usepackage{amsmath}
\usepackage{amsthm}
\usepackage{amssymb}
\usepackage{mathtools}
\usepackage{thmtools}
\usepackage{comment}
\usepackage{bbm}
\usepackage{bm}
\allowdisplaybreaks

\let\originalmiddle=\middle
\def\middle#1{\mathrel{}\originalmiddle#1\mathrel{}}

\usepackage{algorithm}
\usepackage[noEnd,commentColor=blue!70!white]{algpseudocodex}

\makeatletter
\newcommand{\algeq}[2][]{%
  \par%
  \noindent%
  \makebox[\columnwidth]{%
    \hfill #2%
    \if\relax\detokenize{#1}\relax
      \hfill
    \else
      \hfill (\refstepcounter{equation}\theequation)\label{#1}%
    \fi
  }%
  \par%
}
\makeatother

\usepackage[colorlinks=true, linkcolor=blue!60!black, citecolor=blue!60!black,
urlcolor=black]{hyperref}
\usepackage{cleveref}
\AddToHook{cmd/appendix/before}{%
    \crefalias{section}{appendix}%
    \crefalias{subsection}{appendix}
}

\Crefname{ALG@line}{Line}{Lines}
\Crefname{assumption}{Assumption}{Assumptions}
\Crefname{protocol}{Protocol}{Protocols}
\Crefformat{equation}{Eq. #2(#1)#3}
\Crefrangeformat{equation}{Eqs. #3(#1)#4 to #5(#2)#6}
\Crefmultiformat{equation}{Eqs. #2(#1)#3}{ and #2(#1)#3}{, #2(#1)#3}{ and #2(#1)#3}
\Crefrangemultiformat{equation}{Eqs. #3(#1)#4 to #5(#2)#6}{ and #3(#1)#4 to #5(#2)#6}{, #3(#1)#4 to #5(#2)#6}{ and #3(#1)#4 to #5(#2)#6}

\newtheorem{theorem}{Theorem}
\newtheorem{lemma}[theorem]{Lemma}

\newtheorem{claim}[theorem]{Claim}

\theoremstyle{definition}
\newtheorem{definition}{Definition}
\newtheorem{assumption}[definition]{Assumption}

\usepackage{tikz}
\usetikzlibrary{positioning}
\usetikzlibrary{calc}
\usetikzlibrary{decorations.pathreplacing}
\usetikzlibrary{patterns}

\newcommand{\E}{\operatornamewithlimits{\mathbb{E}}}

\renewcommand{\O}{\operatorname{\mathcal O}}
\newcommand{\Otil}{\operatorname{\tilde{\mathcal O}}}
\newcommand{\argmax}[1]{\operatornamewithlimits{argmax}}
\newcommand{\argmin}[1]{\operatornamewithlimits{argmin}}

\renewcommand{\tilde}{\widetilde}
\renewcommand{\hat}{\widehat}

\usepackage{comment}
\renewcommand{\paragraph}[1]{\vspace{5pt}\noindent\textbf{#1}}

\usepackage{xspace}
\newcommand{\mechname}{\texttt{AdaAudit}\xspace}

\Crefname{enumi}{Restriction}{Restrictions}

\newcommand{\paren}[1]{\left( #1 \right)}
\newcommand{\sqb}[1]{\left[ #1 \right]}
\newcommand{\set}[1]{\left\{ #1 \right\}}

\newcommand{\abs}[1]{\left|#1\right|}

\newcommand{\mR}{\mathfrak{R}}
\newcommand{\mB}{\mathfrak{B}}

\newcommand{\poly}{\mathrm{poly}}

\newcommand{\truth}{\textbf{truth}}

\newcommand{\mA}{\mathcal{A}}

\newcommand{\mH}{\mathcal{H}}

\newcommand{\mE}{\mathcal{E}}
\newcommand{\mD}{\mathcal{D}}

\newcommand{\mV}{\mathcal{U}}
\newcommand{\1}{\mathbbm{1}}

\begin{document}

\maketitle

\begin{abstract}
We study repeated resource allocation with strategic agents, where monetary transfers are disallowed and the planner has no prior information on agents' utility distributions. Inspired by the costly state verification literature, we assume the planner can request costly \emph{audits} on the winning agent after allocation, revealing their true utility but without the ability to revoke the allocation. We design a mechanism achieving $T$-independent $\mathcal O(K^2)$ regret in social welfare while requesting $\mathcal O(K^3 \log T)$ audits in expectation, where $K$ is the number of agents and $T$ is the number of rounds. We further show an $\Omega(K)$ lower bound on the regret and an $\Omega(1)$ lower bound on the number of audits required for low regret. We also generalize our mechanism and analysis to imperfect audit models.
Algorithmically, we show that incentivizing truthful behavior relies on accurately estimating agents' truthful winning probability online. To achieve this, we impose future punishments via \emph{adaptive audits}; we also introduce an \emph{incentive-aligned flagging} component allowing agents to flag biased estimates, which we prove is in their best interest. Analytically, without distributional information, the revelation principle cannot dictate a truth-telling equilibrium. Instead, we characterize a Perfect Bayesian Equilibrium via a reduction to an \emph{auxiliary game} with only benign strategies. The technical tools developed herein can be of independent interest for other robust mechanism design problems where the revelation principle is inapplicable.
\end{abstract}

\section{Introduction}

Consider a central planner tasked with allocating scarce resources among a group of self-interested agents. To maximize social welfare, the planner must overcome a fundamental information asymmetry: Agents' true utilities for the resource are private information, and the planner can only observe their strategic reports.
A classical solution to this challenge is the seminal Vickrey-Clarke-Groves (VCG) class of mechanisms \citep{vickrey1961counterspeculation,clarke1971multipart,groves1973incentives}, which employs monetary transfers as the main tool to ensure {incentive-compatibility}. Under such allocation and transfer rules, it is in each agent's best interest to truthfully report their utilities, thereby allowing the planner to maximize welfare-maximizing allocations.

However, in many realistic scenarios---such as the allocation of computational resources within a non-profit organization \citep{ng2011online}, the distribution of foods in a food bank \citep{prendergast2017food,prendergast2022allocation}, or the pairing of volunteers with opportunities \citep{manshadi2023redesigning}---monetary transfers are either impractical or undesirable. These settings are often inherently socially-motivated, where the introduction of monetary transfers could undermine the organization's mission, exacerbate inequality, or be deemed ethically inappropriate. This motivates the design of non-monetary resource allocation mechanisms; we direct the readers to \citep{schummer2007mechanism} for a comprehensive overview on the study of non-monetary mechanisms.

Per Arrow's theorem and related impossibility results \citep{arrow1950difficulty,gi73,s75}, in single-round allocations without monetary transfers, it is in general impossible to achieve incentive-compatibility while maximizing social welfare (except in some specific settings that we discuss in \Cref{sec:related work}).
Meanwhile, a stream of papers showed positive results in repeated allocation settings, where a fixed set of forward-looking agents interact with the planner for multiple rounds.
Intuitively, while the absence of monetary transfers prohibits intra-round transfers, the repeated setup allows the planner to leverage inter-round incentives---essentially, favoring or penalizing agents in future rounds based on their current reports.

The most common design of non-monetary mechanisms is via reduction to monetary ones. Specifically, the planner can allocate a budget of \emph{artificial currencies} to each agent to simulate a market mechanism. This paradigm provides approximate Bayesian Incentive-Compatibility and social welfare guarantees \citep{gorokh2021monetary}. Due to its simplicity, it also exhibits practical success in allocating courses \citep{budish2016bringing} or food-bank donations \citep{prendergast2017food,prendergast2022allocation}.
Another design of non-monetary mechanisms uses the \emph{promised utility} framework, which ensures exact Bayesian Incentive-Compatibility and typically yield faster convergence towards the optimal social welfare \citep{balseiro2019multiagent,blanchard2024near}.

Notably, both design paradigms rely on the assumption that all agents' utility distributions are known to the mechanism \textit{a-priori}: the artificial currency approaches require distributions to set initial budgets, while the promised utility framework calculates a fixed point of a Bellman-like operator induced by utility distributions.
We argue, however, that obtaining precise information about \emph{strategic} agents' utility distributions is a significant practical challenge. The classical rationale for assuming distributional priors is that it can be derived from historical data. Unfortunately, relying on historical data when facing strategic agents is fragile, since such data is fundamentally \emph{endogenous} and subject to manipulation if prior mechanisms failed to be incentive-compatible.
This fragility highlights the necessity of a \emph{robust mechanism design} approach that does not rely on the precise knowledge of agents \citep{wilson1987game,bergemann2005robust}.
Therefore, rather than assuming that this difficult estimation task has been done already, we aim to understand what can be achieved without such prior knowledge. This leads to our central research question:

\begin{center}
\textbf{\textit{(Q): Without prior distributional information on agent's utilities, is it still possible to design a non-monetary resource allocation mechanism that maximizes social welfare?}}
\end{center}

Without further instruments for incentives, the answer is likely negative. Agents could arbitrarily manipulate their reports, and the hardness results arising from Arrow's impossibility theorem \citep{arrow1950difficulty} persist.
Inspired by the costly state verification framework in contract theory \citep{townsend1979optimal,ben2014optimal}, we consider a setting where the planner can occasionally request \emph{costly ex-post} audits to reveal the true utility of the allocated agent (in the majority of this paper, we consider the case where audits exactly reveal true utilities; in \Cref{sec:noisy case}, we discuss two \emph{imperfect audits models} where our mechanism also works well).
This capability mirrors many real-world non-profit operations where planners can eventually evaluate allocation efficiency---albeit at a significant cost---through surveys, utilization reports, or operational audits.
While audits have long been studied to enforce incentive-compatibility in mechanism design, existing works heavily relies on \textit{i)} the comprehensive prior knowledge on agents' private utility distributions, and \textit{ii)} the ability of directly punishing untruthful agents discovered by audits. Therefore, as we discuss in more details in \Cref{sec:related work}, existing methodologies do not apply to our setup. Here, the planner has to \emph{learn} agents' utility distributions online when facing strategic behavior, to which the planner \emph{cannot} give direct punishments.

\subsection{Setup: Non-Monetary Resource Allocation without Distributional Priors}
We consider a $T$-round repeated game between one central planner and $K$ self-interested agents. In each round $t=1,2,\ldots,T$, a same indivisible resource is presented to the central planner. Each agent $i=1,2,\ldots,K$ privately observes their utility $u_{t,i}\in [0,1]$ for the resource, such that $u_{t,i}$ is an independent sample from their own utility distribution $\mV_i$. Importantly, all the distributions $\mV_i$'s are \emph{fixed but unknown} to the planner.

All agents simultaneously submit their reports $v_{t,i}\in [0,1]$ according to a Perfect Bayesian Equilibrium (PBE), based on which the planner irrevocably allocates the resource to an agent $i_t$. The planner is unable to collect any payment from any agent; however, they can decide whether to audit the winner $i_t$---which we encode as a binary decision $o_t\in \{0,1\}$. When setting $o_t=1$, the planner observes the true utility $u_{t,i_t}$ (or a imperfect signal of it, as we will extend in \Cref{sec:noisy case}); the planner does not obtain any information if $o_t=0$.

By designing a collection of allocation and audit rules for each round (which constitute a \emph{mechanism}, as we will formulate in \Cref{sec:setup}), the planner aims to minimize two objectives at the same time: \textit{i)} the social welfare regret $\mR_T:=\E[\sum_t (\max_i u_{t,i}-u_{t,i_t})]$ and \textit{ii)} the expected number of audits $\mB_T:=\E[\sum_t o_t]$.

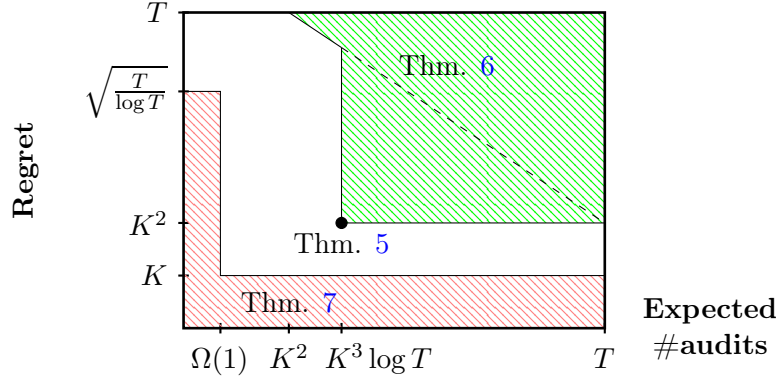
\begin{figure}[tb]
\centering
\begin{tikzpicture}[scale=0.7]

\draw (4,-1.5) node {$\E[\textbf{\#audits}]$};
\draw (-1.5,4) node[rotate=90]{\textbf{Regret}};

\draw[thick,draw=none,pattern=north west lines,pattern color=green!40!white] (2,8) -- (8,2) -- (8,8);
\draw[thick,draw=none,pattern=north east lines,pattern color=green] (3,7) -- (3,2) -- (8,2);
\draw[thick,draw=none,pattern=north west lines,pattern color=red!40!white] (0,0) -- (8,0) -- (8,1) -- (0.5,1) -- (0.5,6) -- (0,6);

\draw[line width=1pt,black] (0,0) rectangle (8,8);
\draw[thick] (8,0.1) -- (8,-0.1) node[below]{$T^{\vphantom{2}}$};
\draw[thick] (0.1,8) -- (-0.1,8) node[left]{$T$};

\draw[thick] (2,0.1) -- (2,-0.1) node[below]{$K^2$};
\draw[thick] (0.1,2) -- (-0.1,2) node[left]{$K^2$};
\draw (6,6) node{\textbf{\Cref{thm:simple mech main theorem}}};

\draw[thick] (3,0.1) -- (3,-0.1);
\draw (3.7,-0.1)  node[below]{$K^3\log T$};
\draw[thick] (0.1,1) -- (-0.1,1) node[left]{$K$};
\draw (2,8) -- (3,7);
\draw[dashed] (3,7) -- (8,2);
\draw (3,7) -- (3,2) -- (8,2);
\filldraw (3,2) circle (3pt) node[below]{\textbf{\Cref{thm:deterministic main thm}}};

\draw[thick] (0.5,0.1) -- (0.5,-0.1) node[below]{$\Omega(1)^{\vphantom{2}}$};
\draw[thick] (0.1,6) -- (-0.1,6) node[left]{$\sqrt {\frac{T}{\log T}}$};
\draw (0,6) -- (0.5,6) -- (0.5,1) -- (8,1);
\draw (3,0.5) node {\textbf{\Cref{thm:lower_bounds}}};

\end{tikzpicture}

\caption{Trade-off between social welfare regret and expected number of audits. The dashed green (resp. red) region depicts known achievable (resp. unachievable) regions.}
\label{fig:main_plot}
\end{figure}

\subsection{Our Contributions}
There is a natural tradeoff between the regret $\mR_T$ and audits $\mB_T$: tackling strategic manipulations requires frequent audits, whereas minimizing audits risks leaving misreports unchecked. Therefore, how effectively can a central planner resolve this conflict? Our primary contribution is demonstrating that the planner does \emph{not} need to heavily compromise. Instead, a small (logarithmic in $T$) number of audits is sufficient to substitute for the missing distributional priors and achieve the first-best social welfare.
Towards characterizing the frontier between $\mR_T$ and $\mB_T$, we summarize our core theoretical and algorithmic contributions as follows:
\begin{itemize}
\item \textbf{Achievable Regions.}
As a baseline, we first achieve a linear trade-off between $\mR_T$ and $\mB_T$ via a mechanism that \emph{uniformly} audits the winning agent with a fixed probability $p\in (0,1]$ (\Cref{alg:simple_mechanism} and \Cref{thm:simple mech main theorem}). Surpassing this greatly, we propose \mechname, a mechanism that adaptively audits the winner, which attains $T$-independent regret $\mR_T=\O(K^2)$ using only a logarithmic number of audits $\mB_T=\O(K^3\log T)$ (\Cref{alg:mechanism} and \Cref{thm:deterministic main thm}). This demonstrates that a small number of ex-post audits suffices to substitute for prior distributional information---which is hard to obtain in practice due to strategic agents.
\item \textbf{Online Estimation in Strategic Environments.}
Algorithmically, the central challenge in designing an effective adaptive audit scheme is estimating agents' \emph{ex-ante first-best utilities} (expected per-round utilities under the welfare-maximizing allocation) in an \emph{online} manner. The historical data used for estimation is generated by \emph{strategic} agents, who may have manipulated reports to bias future estimates; hence standard concentration inequalities fail (we demonstrate this in \Cref{sec:description of est and flag}). We resolve this failure by introducing an algorithmic component called \emph{incentive-aligned flagging}. By requesting and incorporating information from all agents, we align individual incentives with the planner's objective of accurate online estimation.
\item \textbf{Analysis via Reduction to Auxiliary Games.}
In our prior-free setting, a major analytical barrier is the \emph{inapplicability} of the revelation principle: While a truthful direct-revelation mechanism exists for any \emph{fixed} utility distribution, such a mechanism cannot be constructed without distributional information. In other words, the revelation principle cannot yield a \emph{single} mechanism that enforces truthfulness across \emph{all} possible distributions.
Thus we cannot simply \emph{assume} truth-telling as a PBE; instead, to characterize a PBE, we develop a reduction framework mapping the complex dynamic game to a simpler \emph{auxiliary game} restricted to tractable strategies. As we justify in \Cref{sec:sketch varying-p}, any PBE in this auxiliary game correspond to a PBE in the original game (despite having a much larger strategy class), thereby establishing \Cref{thm:deterministic main thm}.
\item \textbf{Hardness Results and Robustness Guarantees.}
In \Cref{sec:lower bounds}, we complement our mechanisms by hardness results. Specifically, as we color with red in \Cref{fig:main_plot}, we prove that $\mR_T=\Omega(K)$ is unavoidable, and that any mechanism attaining $\mR_T=o(\sqrt{T/\log T})$ must require $\mB_T=\Omega(1)$ audits in expectation.
In \Cref{sec:noisy case}, we further demonstrate the robustness of \mechname by extending our results to \emph{imperfect} audit models, where either \textit{i)} agents can strategically distort the audit outcome, or \textit{ii)} audits are subject to stochastic noise. Notably, our auxiliary game analysis remains largely the same, suggesting it can be of independent interest to other mechanism design problems where the revelation principle is inapplicable.
\end{itemize}

\section{Related Literature}\label{sec:related work}
\subsection{Mechanism Design for Non-Monetary Resource Allocation}\label{sec:related work non-monetary}

When allocating scarce resources to strategic agents, classical mechanisms typically rely on monetary transfers \citep[e.g., the VCG mechanisms by][]{vickrey1961counterspeculation,clarke1971multipart,groves1973incentives} to ensure both incentive-compatibility (incentivizing truthful reports from the agents) and efficiency (maximizing social welfare). When monetary transfers are prohibited, achieving these two goals simultaneously is notoriously difficult.

In one-shot settings (i.e., $T=1$), Arrow's impossibility theorem dictates that non-monetary mechanisms face inherent limits in maximizing social welfare \citep{arrow1950difficulty,gi73,s75}. Positive results are largely restricted to narrow environments, e.g., two-agent symmetric settings \citep{miralles2012cardinal} or specific additive utility classes \citep{guo2010strategy,han2011strategy,cole2013positive}. A variant model called \emph{money burning}---where the planner can destroy a portion of the agents' utilities---also incurs an efficiency loss \citep{hartline2008optimal,hoppe2009theory,condorelli2012money}.

To overcome these one-shot impossibilities, a rich literature studies repeated allocations over $T$ rounds, leveraging inter-round incentives. When agents' utility distributions are exactly known \textit{a-priori}, folk theorems guarantee asymptotic efficiency \citep{fudenberg2009folk,guo2009competitive}. Specifically, two dominant paradigms have emerged to achieve finite-time convergence. The \emph{artificial currency} approach allocates initial budgets based on the known priors to simulate a monetary market \citep{friedman2006efficiency,jackson2007overcoming,kash2007optimizing,kash2015equilibrium,budish2011combinatorial,johnson2014analyzing,gorokh2021monetary}. Alternatively, the \emph{future promises} framework, building upon the d'Aspremont--Gérard-Varet--Arrow mechanism \citep{d1979incentives,arrow1979property}, targets Perfect Bayesian Equilibria (PBE) by dynamically adjusting agents' future utilities \citep{balseiro2019multiagent,blanchard2024near}. These powerful paradigms, however, are tailored to environments where the planner has access to precise utility distributions.

Dropping the assumption of known priors motivates a robust mechanism design approach \citep{wilson1987game,bergemann2005robust}.
Should monetary transfers be allowed, the planner can leverage payments to elicit truthful reports, seamlessly learning the distribution from historical data for auction design \citep{babaioff2012dynamic,amin2013learning,kanoria2021incentive} or resource allocation \citep{besbes2019static,kandasamy2023vcg,dai2025incentive}.
In non-monetary cases, current research has mostly focused on establishing robust utility guarantees for individual agents, where \citet{gorokh2019remarkable} pioneered the \emph{$\beta$-Utopia} metric that has been adopted by subsequent works on prior-free non-monetary resource allocation \citep{banerjee2023robust,fikioris2025beyond,onyeze2025allocating,lin2026robust}. Generalizing the competitive ratio metric in single-agent cases, this metric ensures that for any given target allocation probability vector $\bm \alpha$, in equilibrium, every agent attains at least a $\beta$ fraction of their ex-post optimal utility.

While the $\beta$-Utopia metric serves as a powerful benchmark for individual worst-case guarantees, our work pursues the complementary objective of maximizing social welfare. As no single allocation probability vector $\bm \alpha$ guarantees sub-linear social welfare regret across \emph{all} utility distributions, the performance guarantees of these prior works are not directly comparable to ours.
By empowering the planner with adaptive audits, our work bypasses the impossibility to simultaneously elicit near-truthful behavior and achieve highly efficient allocations. Consequently, our analysis centers on the social welfare regret, a metric frequently studied in non-strategic online resource allocation \citep{devanur2009adwords,balseiro2023best,besbes2025dynamic}.

Most relevant to our work is that of \citet{yin2022online}, who elegantly design non-monetary mechanisms without distributional priors to maximize social welfare for \emph{homogeneous} agents. In their setting, as all agents share the same utility distribution, the planner can leverage cross-agent reports to detect if a specific agent is misreporting; this consequently elicits truthful reports from agents. On the other hand, our setting features \emph{heterogeneous} agents. Without the shared distribution, we cannot infer one agent's truthfulness by exogenous reports from others. Instead, the planner must assess each agent solely based on their own endogenous history, thus requiring alternative instruments to verify agents' private information occasionally. This leads to our incorporation of ex-post audits and resonates with the literature on costly state verification.

\subsection{Mechanism Design with Costly State Verification}\label{sec:related work CSV}
Audits, as a costly tool to verify the private information of agents, have been widely studied in the mechanism design literature. The seminal work by \citet{townsend1979optimal} considers a one-shot \emph{monetary} resource allocation problem, where the planner can request costly audits \emph{before} making allocation decisions (which we refer to as \emph{ex-ante audits}; we discuss the distinction between ex-ante and ex-post audits shortly). Subsequently, other monetary mechanisms utilizing costly state verification \citep{gale1985incentive,border1987samurai,mookherjee1989optimal}, as well as money-burning mechanisms \citep{patel2017costly,lundy2019allocation}, have been proposed.
More recently, post-allocation inspections and exclusion mechanisms have continued to drive significant advances in monetary settings \citep{belloni2025approximately,bayrak2025distributionallycontract,li2025revenue,alaei2026optimal}.
Given our emphasis on environments where transfers are either infeasible or undesirable, we depart from this monetary stream to focus purely on {non-monetary} resource allocation.

Non-monetary resource allocation with audits was first studied by \citet{ben2014optimal}. They consider a one-shot allocation problem with homogeneous agents and \emph{known} utility distributions, where the planner can request ex-ante audits prior to allocation. This setting is later generalized to multi-item scenarios with identical items \citep{ben2019mechanisms,erlanson2024optimal}, as well as to secretary problems with sequential agent arrivals \citep{epitropou2019optimal,li2024dynamics,li2026dynamic}.

Our work differs from this line of research in two critical dimensions. In our setting: \textit{i)} agents are \emph{heterogeneous} with \emph{unknown} utility distributions; and \textit{ii)} costly audits occur only \emph{after} allocations, which we refer to as \emph{ex-post audits}. Regarding Difference \textit{ii)}, the key distinction between ex-ante and ex-post audits lies not merely in timing, but in the severity of available punishments. With ex-ante audits, the planner can withhold the resource from an agent found misreporting, imposing an \emph{unlimited punishment} (i.e., the agent receives zero utility). Conversely, with ex-post audits, the winning agent has already received the allocation and retains utility, rendering the punishment \emph{limited} (see \citealt{epitropou2019optimal} for a related discussion).

Focusing on ex-post audits with limited punishment, \citet{mylovanov2017optimal} design a one-shot non-monetary mechanism utilizing \emph{costless} ex-post audits, while \citet{li2020mechanism} study the more general setting of \emph{costly} ex-post audits with \emph{heterogeneous} agents. Both works assume perfect knowledge of agents' utility distributions, thus diverging from our setting regarding Difference \textit{i)}. Furthermore, they equip the planner with an \emph{exogenous} ability to destroy a fixed portion of the audited agent's utility (e.g., a university revoking part of the scholarship). In contrast, our repeated allocation setting requires the planner to \emph{endogenously} construct any punishment to agents through a dynamic process, as meticulously designed in \Cref{sec:main mechanism}.

We now turn to Difference \textit{i)}: the prior knowledge of utility distributions. As we are aware of, the only previous work on costly state verification addressing this informational friction is by \citet{bayrak2025distributionally}. Taking a distributionally robust optimization perspective, they consider a one-shot non-monetary allocation with costly \emph{ex-ante} audits, assuming the agents' joint distribution belongs to a \emph{known} ambiguity set. While also maximizing the worst-case social welfare, we allow agents' distributions to be completely \emph{arbitrary} (subject to independence and the mild \Cref{assumption:min report}), and the planner lacks \emph{any} prior knowledge of them. Moreover, our planner is only able to conduct \emph{ex-post} audits, which inherently induce limited punishment.

Additional but less directly related literature includes the application of audits to settings beyond resource allocation \citep{estornell2021incentivizing,estornell2023incentivizing,jalota2024catch}; mechanisms utilizing alternative instruments like subsidies or robust classifier design \citep{estornell2021incentivizing,estornell2023incentivizing}; and probabilistic verification models where audits are subject to noise, analogous to our \Cref{sec:noisy case} \citep{caragiannis2012mechanism,ball2019probabilistic}.

To conclude, the existing literature on mechanism design with costly state verification does not readily apply to our setting, which features heterogeneous agents, unknown distributions, and costly ex-post audits with limited punishment. Collectively, these two differences forge our core technical challenge: an effective audit scheme typically requires either extensive prior knowledge of agents or the ability to impose substantial punishment. In our setup, we must estimate these properties \emph{online} in the face of strategic behavior, all while relying exclusively on costly ex-post audits and lacking the ability to impose direct unlimited punishments.

\section{Setup: Distribution-Free Non-Monetary Resource Allocation}\label{sec:setup}

\paragraph{Notation.}
For $n\in\mathbb N$, we use $[n]$ to denote the set $\{1,2,\ldots,n\}$.
We use bold letters like $\bm u$ and $\bm v$ for vectors and unbolded letters like $u_i$ and $v_i$ for elements therein.
We use standard Landau notations, where $\O$ and $\Omega$ only hide universal constants and $\Otil$ and $\tilde{\Omega}$ additionally hide logarithmic factors.

\subsection{Repeated Resource Allocation with Audits}

We consider a $T$-round sequential game in which a central planner repeatedly allocates an indivisible but reusable resource to $K\geq 2$ strategic agents. In each round, private utilities of each agent $i\in[K]$ are sampled i.i.d. from their own fixed utility distribution---which is \emph{unknown} to the planner. Without loss of generality, we assume these distributions have support on $[0,1]$. In each round, the planner receives reports $v_{t,i}$ from all agents $i\in[K]$ and decides whom to allocate the resource to (potentially no one, i.e., the ``reject'' option is available). The planner can then potentially request an \emph{audit} of the winner, revealing their true utility at a cost. Formally, each round $t\in[T]$ proceeds as follows.
\begin{enumerate}
    \item Each agent $i\in[K]$ observes their private utility $u_{t,i}\overset{i.i.d.}{\sim} \mV_i$, which is sampled independently from the history and between agents. We abbreviate the generation of all agents' private utilities by $\bm u_t\sim \mV$.
    \item Each agent $i\in[K]$ submits a report $v_{t,i}\in [0,1]$ based on their \emph{observable history} $\mH_t^i$ (formulated later) and their own utility $u_{t,i}$. They cannot access others' utilities $(u_{t,j})_{j\ne i}$, nor their future utilities $(u_{\tau,i})_{\tau>t}$.
    \item The planner allocates the resource to an agent $i_t\in[K]\cup\{0\}$ based on the \emph{public history} $\mH_t^0$ (again, formulated later) and reports $\bm v_t=\{v_{t,i}\}_{i\in[K]}$. Choosing $i_t=0$ corresponds to not allocating the resource.
    \item After the allocation, the planner may decide to perform an {audit} or not, denoted by a binary variable $o_t\in\{0,1\}$. If $o_t=1$, the planner observes feedback $w_t\in [0,1]$ related to the winning agent's true utility, namely $u_{t,i_t}$; otherwise, they observe nothing, in which case we write $w_t=0$. \label{item:possible_audit}
    \item All public information---specifically reports $\bm v_t$, winning agent $i_t$, audit decision $o_t$, and outcome $w_t$---is revealed to all agents. The planner is free to request further information from the agents (thus allowing the planner to further interact with agents beyond their reports), but agents are also free to answer strategically. We denote by $f_{t,i}$ the answer of agent $i\in[K]$ in round $t\in [T]$. Notationally, we assume answers must lie in a fixed measurable space $F$, where $F=\{0\}$ corresponds to not asking further questions.
    \label{item:additional_questions}
\end{enumerate}
In the vast majority of this paper, we assume a \emph{perfect audit model} where the planner's audit $w_t$ exactly reveals the winner's true utility $u_{t,i_t}$, as formalized in \Cref{assump:noiseless}. (In \Cref{sec:noisy case}, we discuss two \emph{imperfect} audit models, namely \Cref{assump:adv noise model,assump:stoc noise model}, to which our proposed mechanism can be extended.)
\begin{assumption}[Perfect Audit Model]\label{assump:noiseless}
In a round $t\in [T]$ where the planner decides to audit the winner (i.e., $o_t=1$), the audit outcome $w_t$ equals the winner's true utility $u_{t,i_t}$; otherwise, the planner observes nothing and we write $w_t=0$. Equivalently, we write $w_t=o_tu_{t,i_t}$.
\end{assumption}

\paragraph{No Distributional Priors.}
We assume the planner has neither access to, nor any prior knowledge on the distributions $\{\mV_i\}_{i\in[K]}$. As a result, any distributional information must be learned by requesting audits, which is the critical difference between our setup and standard (non-monetary) resource allocation (see \Cref{sec:related work}).
However, the type distributions $\{\mV_i\}_{i\in [K]}$ are public knowledge among agents. This assumption, a standard one in robust mechanism design \citep{bergemann2005robust}, not only is technically necessary to define Bayesian equilibria, but also captures the prevalent information asymmetry in practice: participants often posses knowledge about their peers, whereas the planner acts as an outsider facing uncertainty.

\subsection{Game-Theoretic Definitions}
We start by formally defining histories at the beginning of a round $t\in[T]$. The \emph{public history} $\mH_t^0$ contains all agents' previous reports, allocations, audit decisions, audit outcomes, and agents' answers, {i.e.}, $\mH_t^0:=\set{(\bm v_\tau,i_\tau,o_\tau,w_\tau, \bm f_\tau)}_{\tau<t}$. For any agent $i\in[K]$, their \emph{observable history} $\mH_t^i$ additionally contains their private utilities, {i.e.}, $\mH_t^i:=\set{(u_{\tau,i},\bm v_\tau,i_\tau,o_\tau,w_\tau,\bm f_\tau)}_{\tau<t}$. Last, we define the {complete history} at round $t$ via $\mH_t:=\set{(\bm u_{\tau},\bm v_\tau,i_\tau,o_\tau,w_\tau, \bm f_\tau)}_{\tau<t}$. Note that the complete history is not available to any player and will only be used for proof purposes. For convenience, we define the space of public histories via $H_t^0:=([0,1]^K \times([K]\cup\{0\})\times \{0,1\}\times [0,1]\times F^K)^{t-1}$, and similarly the space of agent observable histories $H_t^i$ for $i\in[K]$ and the space of full histories $H_t$. We now define the planner's \emph{mechanism} and agents' \emph{strategies}. To allow for randomization, we further fix a measurable space $\Xi$ and a distribution $\mD_\xi$ on $\Xi$.

\begin{definition}[planner Mechanism]
A {planner mechanism} is a sequence of measurable functions $\bm M:=(M_t)_{t\in[T]}$ where $M_t:[0,1]^K\times H_t^0\times \Xi \to ([K]\cup\{0\}) \times \{0,1\}$. In each round $t\in [T]$, this mechanism maps reports $\bm v_t$, public history $\mH_t^0$, and internal randomness $\xi_{t,0}\sim \mD_\xi$ (sampled independently from all other random variables) to allocation $i_t$ and audit decision $o_t$ as $(i_t,o_t):=M_t(\bm v_t,\mH_t^0,\xi_{t,0})$.
\end{definition}

We assume the planner has full commitment power, i.e., they are free to choose \emph{any} mechanism $\bm M$. Classical economics theory reveals that it is in the planner's interest to announce this $\bm M$ as public information.

\begin{definition}[Agent Strategies]
\label{def:agent_strat_mechanism}
    A {report strategy} for agent $i\in[K]$ is a sequence of measurable functions $\pi_i^r:=(\pi^r_{t,i})_{t\in[T]}$ where $\pi^r_{t,i}:[0,1]\times H_t^i\times \Xi \to [0,1]$. Their report in round $t\in[T]$ is given by $v_{t,i}:=\pi^r_{t,i}(u_{t,i},\mH_t^i,\xi^r_{t,i})$, where $\xi^r_{t,i}\sim\mD_\xi$ is sampled independently from other random variables.

    An {answer strategy} (or flagging strategy, which we use interchangeably) for agent $i$ is a sequence of measurable functions $\pi^f_i:=(\pi^f_{t,i})_{t\in[T]}$ where $\pi^f_{t,i}:[0,1]\times [0,1]^K\times ([K]\cup\{0\})\times \{0,1\}\times [0,1]\times H_t^i\times \Xi \to F$. Their answer in round $t\in[T]$ (recall Step~\ref{item:additional_questions}) is given by $f_{t,i}:=\pi^f_{t,i}(u_{t,i},\bm v_t,i_t,o_t,w_t, \mH_t^i,\xi^f_{t,i})$, where $\xi^f_{t,i}\sim\mD_\xi$ is sampled independently from all other random variables.

    A {strategy} for agent $i$ is composed of a report and answer strategy, namely $\pi_i:=(\pi^r_i,\pi_i^f)$.

    A {joint strategy} for agents is a collection of all agent strategies $\pi:=(\pi_{i})_{i\in[K]}$. We also use the joint report strategy and flagging strategy, denoted by $\pi^r:=(\pi^r_i)_{i\in [K]}$ and $\pi^f:=(\pi^f_i)_{i\in [K]}$, respectively.
\end{definition}

For a given planner mechanism $\bm M$ and joint strategy for agents $\pi$, we define the value function (also known as V-function) of any agent $i\in[K]$ from any complete history $\mH_t\in H_t$ as 
\begin{equation}
V_i^{\pi}(\mH_t;\bm M):=\E_{\substack{\bm u\sim \mV,\bm v\sim \pi^r,\\(\bm i,\bm o)\sim \bm M, \bm f\sim\pi^f}}\left [\sum_{\tau=t}^T u_{\tau,i} \1[i_{\tau}=i]\middle \vert \mH_t\right ].\label{eq:V-func general}\end{equation}
When it is clear from the context, we often drop the parameter $\bm M$ for the ease of presentation.

The agents' objectives are to strategically maximize their own value functions. Formally, fixing a mechanism $\bm M$, we assume the agents' joint strategy forms a Perfect Bayesian Equilibrium (PBE).

\begin{definition}[Perfect Bayesian Equilibrium]\label{def:PBE}
Under a planner mechanism $\bm M$, a joint strategy $\pi=(\pi_i)_{i\in [K]}$ is a Perfect Bayesian Equilibrium (PBE) if, for any possible history, any unilateral deviation of any agent cannot increase their value function. That is, for any alternative strategy $\pi_i'$ of agent $i\in [K]$, let $\pi'$ be the joint strategy where agent $i$ follows $\pi_i'$ and any other agent $j\ne i$ follows $\pi_j$, then we must have
\begin{equation}\label{eq:PBE}
V_i^{\pi'}(\mH_t;\bm M)\le 
V_i^{\pi}(\mH_t;\bm M),\quad \forall t\in[T],\mH_t\in H_t.
\end{equation}
\end{definition}

While part of the complete history $\mH_t$ is inaccessible to the agent $i$, we remark that \Cref{eq:PBE} naturally implies for any observable $\mH_t^i\in H_t^i$ that $\E[V_i^{\pi'}(\mH_t;\bm M)\mid \mH_t^i]\le \E[V_i^{\pi}(\mH_t;\bm M)\mid \mH_t^i]$.

\subsection{Objectives of the planner}\label{sec:setup planner objective}
The planner aims to design a mechanism that both (in expectation) maximizes the social welfare $\sum_{t=1}^T u_{t,i_t}$ and minimizes the number of audits requested, at a corresponding PBE. For notational simplicity, we define $u_{t,0}=0$ for those rounds where $i_t=0$, i.e., no welfare arises if the allocation is rejected. The welfare maximization objective is equivalent to minimizing the regret w.r.t. optimal allocations, as defined below.

\begin{definition}[Regret and Expected Number of Audits]
For a planner mechanism $\bm M$ and an agents' joint strategy $\pi$, the regret and the expected number of audits are
\begin{align*}
\mR_T(\pi,\bm M)&= \E_{\substack{\bm u\sim \mV,\bm v\sim \pi^r,\\(\bm i,\bm o)\sim \bm M, \bm f\sim\pi^f}}\sqb{\sum_{t=1}^T \left (\max_{i\in [K]} u_{t,i}-u_{t,i_t}\right )},\\
\mB_T(\pi,\bm M)&= \E_{\substack{\bm u\sim \mV,\bm v\sim \pi^r,\\(\bm i,\bm o)\sim \bm M, \bm f\sim\pi^f}}\sqb{\sum_{t=1}^T o_t}.
\end{align*}
\end{definition}

We remark that the social welfare regret $\mR_T(\pi,\bm M)$ compares established utility $u_{t,i_t}$ to the first-best utility $\max_{i\in [K]} u_{t,i}$ in every round, which is much stronger than the online learning (static) regret that compares to a fixed optimal agent in hindsight. This regret notion is tractable because reports from the agents---as extra pieces of information---are made available before decision-making.

\paragraph{Deterministic Tie-Breaking Rule.}
When comparing utilities or reports across agents, the planner always uses the lexicographical order of agent indices to break ties. That is, for any two agents $i\ne j$, whenever we write $u_i<u_j$, it means either \textit{i)} $u_i<u_j$, or \textit{ii)} $u_i=u_j$ but $i<j$. Employing a deterministic tie-breaking rule is critical since otherwise the probability of allocation (the ``ex-ante first-best utility'' we define later) can change drastically when ties are not broken properly, as observed by \citet{dutting2024online}.

\section{Non-Monetary Mechanisms without Distributional Priors}\label{sec:main mechanism}
While dominant approaches for non-monetary online resource allocation, such as {artificial currency} \citep{gorokh2021monetary} and {promised utilities} \citep{balseiro2019multiagent}, allow sharp convergence to (approximate) Bayesian equilibria when utility distributions are known, it is unclear how to adapt them to prior-free settings. 
Specifically, the artificial currency framework assigns an initial budget to each agent based on their expected total payment throughout the game (which depend on, for example, the distribution of the second-highest utilities), while the promised utility framework relies on computing fixed points of Bellman-like functionals.
Our audit-based mechanism instead relies on a different one-dimensional statistics called \emph{ex-ante first-best utility} (i.e., agents' expected per-round utility under truthful reports and welfare-maximizing allocation), which is easier to estimate in the strategic and online environment.
Before introducing this statistics, however, we begin by introducing a baseline mechanism to demonstrates how ex-post audits elicit truthful behavior.

\subsection{Baseline Mechanism: Fixed-Probability Auditing Yields a Linear Trade-Off}\label{sec:simple mechanism}

\begin{algorithm}[!t]
\caption{Baseline Mechanism $\bm M^0(p)$}
\label{alg:simple_mechanism}
\begin{algorithmic}[1]
\Require{Number of rounds $T$, number of agents $K$, audit probability $p$}
\Ensure{Allocations $i_1,i_2,\ldots,i_T\in \{0\}\cup [K]$}
\State Initialize the set of alive agents $\mA_1=[K]$.
\For{$t=1,2,\ldots,T$}
\State Collect reports $v_{t,i}\in [0,1]$ from $i\in \mA_t$ and allocate to agent $i_t=\argmax_{i\in \mA} v_{t,i}$.
\State Decide whether to audit the winner via $o_t\sim \text{Ber}(p)$. Observe audit outcome $w_t$.
\If{$o_t=1$ \textbf{and} $w_t \neq v_{t,i_t}$} 
\State Eliminate agent $i_t$ permanently: $\mA_{t+1}\gets \mA_t\setminus \{i_t\}$.\label{line:simple mech elimination}
\Else
\State All agents stay alive: $\mA_{t+1}\gets \mA_t$.
\EndIf
\EndFor
\end{algorithmic}
\end{algorithm}

To motivate our adaptive approach, we first analyze a baseline mechanism, $\bm M^0(p)$. This mechanism utilizes the concept of \emph{future punishment}: since the planner lacks the ability to revoke allocations \citep[unlike previous mechanisms utilizing ex-post audits, e.g.,][]{mylovanov2017optimal,li2020mechanism}, they can only impose threats about future allocations. Specifically, in each round $t\in [T]$, this mechanism \textit{i)} allocates to the agent with highest report, namely $i_t=\argmax_i v_{t,i}$ (with ties broken lexicographically), \textit{ii)} audits the winner with a fixed and public probability $p\in (0,1]$, i.e., $o_t\sim \text{Ber}(p)$, and \textit{iii)} eliminates agent $i_t$---by ignoring their reports forever---if the audit reveals a misreport. The pseudo-code of $\bm M^0(p)$ is provided in \Cref{alg:simple_mechanism}.

Under the perfect audit model defined in \Cref{assump:noiseless} (where the audit outcome $w_t$ equals $u_{t,i_t}$ when audited, and $0$ otherwise), $\bm M^0(p)$ establishes a linear trade-off between the social welfare regret $\mR_T$ and the expected number of audits $\mB_T$. This baseline performance is visualized in the upper-right region of \Cref{fig:main_plot}.

\begin{theorem}[Performance of Fixed-Probability Auditing]\label{thm:simple mech main theorem}
Fix parameter $p\in(0,1]$ and assume \Cref{assump:noiseless} holds. Under mechanism $\bm M^0(p)$, for any utility distributions $\{\mV_i\}_{i\in[K]}$, any PBE $\pi^\ast$ satisfies:
\begin{align*}
\mR_T(\pi^\ast,\bm M^0(p))&\leq \frac{K(K-1)}{p},\\ \mB_T(\pi^\ast,\bm M^0(p))&=pT.
\end{align*}
\end{theorem}

The formal proof of \Cref{thm:simple mech main theorem} is in \Cref{sec:appendix fixed-p}.
Intuitively, this mechanism leverages the threat of permanent elimination (\Cref{line:simple mech elimination}) to deter misreporting. To see this, we consider an ``alive'' (i.e., not yet eliminated) agent $i \in \mA_t$ in round $t\in [T]$ who is hesitating between the following two options:
\begin{enumerate}
\item \textbf{Truthful} ($v_{t,i}=u_{t,i}$). This incurs no risk of elimination, preserving their access to allocations.
\item \textbf{Misreporting} ($v_{t,i}=1$). This secures an immediate gain of $u_{t,i}$, but risks losing all future utility w.p. $p$.
\end{enumerate}

A rational agent only prefers misreporting if the expected penalty (due to elimination) is outweighed by the immediate gain (due to allocation). Since $u_{t,i}\le 1$, truth-telling must be better whenever $p \cdot \E[\text{future utility}] \ge 1\ge u_{t,i}$.
Therefore, by imposing future punishments, $\bm M^0(p)$ successfully elicits truthful behavior from those agents with high expected future utilities (exceeding $p^{-1}$). Agents with low expected future utility may still misreport, but their impact on the total social welfare---and thus to the regret---is relatively limited.

\subsection{The \mechname Mechanism and Main Result}

\begin{algorithm}[!t]
\caption{Main Mechanism \mechname}
\label{alg:mechanism}
\begin{algorithmic}[1]
\Require{Number of rounds $T$, number of agents $K$, minimum winning utility $c$}
\Ensure{Allocations $i_1,i_2,\ldots,i_T\in \{0\}\cup [K]$, where $i_t=0$ means no allocation to anyone}
\State Initialize the set of alive agents $\mA_1=[K]$, epoch counter $\ell\gets 1$ and start round of the epoch $t_1=0$. For any agent $i\in [K]$, initialize their empirical winning probability $\hat q_{1,i}=0$. \label{line:initialization}
\For{$t=1,2,\ldots,T$}
\State Collect reports $v_{t,i}$ from alive agents $i\in \mA_t$. Set $i_t=\argmax_{i\in \mA_t} v_{t,i}$. \label{line:report}
\State By default, keep $\mA_{t+1}\gets \mA_t$ and $\hat{\bm q}_{t+1}\gets \hat{\bm q}_t$ unchanged.
\If{$v_{t,i_t}\geq c$} 
\State Allocate to the agent $i_t$ and compute an audit probability $\hat p_{t,i_t}:=\min\left (\frac{4(1+K^2)}{(T-t)\hat q_{t,i_t} c},1\right )$. \label{line:check prob}
\State Decide whether to audit the winner by $o_t\sim \text{Ber}(\hat p_{t,i_t})$. Observe audit outcome $w_t$. \label{line:randomly_check}
\If{$o_t=1$ \textbf{and} $v_{t,i_t}>w_t$} \Comment{Winner marked up (under \Cref{assump:noiseless}, $w_t=o_tu_{t,i_t}$)} \label{line:elimination}
\State Eliminate agent $i_t$ permanently by updating $\mA_{t+1}\gets \mA_t\setminus \{i_t\}$. Start a new epoch $\ell\gets \ell+1$ and $t_{\ell}=t+1$. Reset $\hat q_{t+1,i}=0$ for all alive agents $i\in\mA_{t+1}$.
\EndIf

\If{$\mA_{t+1}=\mA_t$ \textbf{and} $\hat q_{t,i_t}=0$} \Comment{Winner does not have winning probability estimate} 
\State Estimate the winning probability of agent $i_t$ as $\hat q_{t+1,i_t}=\frac{\abs{\set{\tau\in [t_\ell,t]:i_t=i}}}{t-t_\ell+1}$. \label{line:estimate winning prob}
\State Ask every agent $i\in [K]$ whether to flag $\hat q_{t+1,i_t}$ is biased, denoted by $f_{t,i}\in \{0,1\}$. \label{line:flagging}
\If{$\bm f_t\neq \bm 0$} \Comment{Someone flags the winner for having a biased $\hat q_{i_t}$}
\State Continue estimation for $i_t$ by resetting $\hat q_{t+1,i_t}=0$. \label{line:restart estimation}
\EndIf
\EndIf
\Else \Comment{No allocation if $v_{t,i_t}<c$}
\State Do not allocate to anyone, that is, setting $i_t=0$.
\EndIf
\EndFor
\end{algorithmic}
\end{algorithm}

While the baseline mechanism $\bm M^0(p)$ well incentivizes those agents with expected future utilities exceeding $p^{-1}$ (from now on, \emph{high-future-utility} agents), the remaining \emph{low-future-utility} agents---whose expected future utilities are bounded by $p^{-1}$---may still misreport, resulting in the $\O(p^{-1})$ regret.
A natural idea is making the audit probability \emph{adaptive}, such that in every round almost all agents are high-future-utility and thus hesitate to misreport.
However, this idea faces significant challenge in our prior-free setting: to make this adaptive audit scheme effective, we need to estimate agents' expected future utility online. Each agent, on the other hand, may strategize to distort this process if an accurate estimation harms their utility.

Therefore, in addition to designing an adaptive audit scheme, we also need to \emph{align} agents' individual incentives with the planner's objective of estimating future utilities. This results in our main mechanism \mechname (\textbf{Ada}ptive \textbf{Audit}ing; \Cref{alg:mechanism}).
As we will introduce in \Cref{sec:description of check prob,sec:description of est and flag}, it incorporates two algorithmic ideas: adaptive audits leveraging expected future utilities, and incentive-aligned flagging for online estimation.
\mechname requires the following regularity assumption on agents' utility distributions.

\begin{assumption}[One Agent's Utility Away from Zero]\label{assumption:min report}
There is an agent $i_0\in[K]$ with utility at least $c>0$ almost surely, i.e., $\mV_{i_0}$ is supported on $[c,1]$, where $c>0$ is a public constant.
\end{assumption}

We remark that the identity of such $i_0$ needs not be known to the planner, and that we only require \emph{one} agent to have their utilities bounded away from $0$.
\Cref{assumption:min report} is a mild regularity assumption in the sense that it is automatically satisfied if, for example, the central planner has a reserve price of $c>0$ for the resource or has a fixed cost $c>0$ for each round of allocation and audit, which is equivalent to an additional agent whose utility is deterministically $c$ \citep[see, e.g.,][]{myerson1981optimal,riley1981optimal}.

Under \Cref{assumption:min report}, \mechname enjoys the following performance guarantee, which means a constant regret and logarithmic expected number of audits when focusing on the dependency on $T$, the length of the game. The proof of \Cref{thm:deterministic main thm} is sketched in \Cref{sec:sketch varying-p} and formally included in \Cref{sec:appendix varying-p,sec:appendix regret and audits,sec:appendix deterministic main theorem}.

\begin{theorem}[Main Theorem of \mechname under Perfect Audit Model]\label{thm:deterministic main thm}
Consider the perfect audit model in \Cref{assump:noiseless}.
Under the mechanism \mechname defined in \Cref{alg:mechanism}, for any utility distributions $\{\mV_i\}_{i\in[K]}$ satisfying \Cref{assumption:min report}, there exists a PBE of agents' strategies $\pi^\ast$ such that
\begin{align*}
\mR_T(\pi^\ast,\mechname)&\leq K^2,\\\mB_T(\pi^\ast,\mechname)&= \O\left (\frac{K^3\log T}{c}\right ).
\end{align*}
\end{theorem}

Compared to the simple mechanism $\bm M^0(p)$, while \mechname additionally requires \Cref{assumption:min report} and only ensures the \emph{existence} of a PBE $\pi^\ast$ inducing low $\mR_T$ and $\mB_T$, the trade-off achieved in terms of expected regret and number of audits is much more favorable, as we demonstrated in the center of \Cref{fig:main_plot}.

\subsection{Adaptive Audits Leveraging Expected Future Utilities}\label{sec:description of check prob}
We now detail the rationale behind the adaptive audit probabilities employed in \mechname ($\hat p_{t,i_t}$ defined in \Cref{line:check prob} of \Cref{alg:mechanism}). Recall our analysis of the baseline mechanism $\bm M^0(p)$ in \Cref{sec:simple mechanism}: the threat of elimination essentially divides agents into two categories, \emph{high-future-utility} agents with $\E[\text{future utility}]\ge p^{-1}$ who hesitate to misreport, and \emph{low-future-utility} agents willing to secure a rare win by misreporting.
The fixed audit probability in the baseline mechanism $\bm M^0(p)$ is inherently inefficient for two reasons:
\begin{enumerate}
\item A fixed $p$ imposes a uniform threshold $p^{-1}$ between these two categories. Without distributional information, it is impossible to calibrate this $p$ optimally: a conservative $p$ wastes audits on high-future-utility agents who are naturally deterred, while a lenient $p$ fails to restrain agents with low future utilities.
\item An agent's future utility is dynamic: it not only depends on their specific utility distribution, but also \emph{decays} as the game progresses (i.e., the ``future'' becomes shorter since our game lasts for a finite number of rounds). An agent who is high-future-utility at the beginning can be low-future-utility near the end.
\end{enumerate}

We therefore \emph{adapt} the audit probability for each agent $i\in[K]$ and round $t\in[T]$, discouraging misreports while minimizing audits.
The governing parameter for this adaptation is the agent's \emph{ex-ante first-best utility}, i.e., the expected per-round utility under the welfare-maximizing allocation, which we formally define as:
\begin{equation}\label{eq:fair share informal}
\mu_i=\E_{\bm u\sim \mV}\big [u_{i}\1[u_i>u_{j},\forall j\ne i]\big ],\quad \forall i\in [K],
\end{equation}
where ties are broken lexicographically by agent index (recall \Cref{sec:setup planner objective}). Note that due to potential eliminations, this $\mu_i$ is in fact time-varying: it depends on the set of alive agents $\mA_t$. For the ease of presentation, we omit this dependency in the current overview; the formal definition is in \Cref{eq:fair share}.

As a first step, we introduce an idealized \emph{full-information mechanism}, denoted by $\bm M^\ast$. In each round $t\in[T]$, $\bm M^\ast$ allocates to the alive agent with the highest report (namely, $i_t:=\argmax_{i\in \mA_t} v_{t,i}$) and audits this winner with a dynamic probability $p_{t,i_t^\ast}$ that is inversely proportional to their expected future utility:
\begin{equation}\label{eq:ideal mech check prob}
o_t\sim \text{Ber}(p_{t,i_t}^\ast),\quad p_{t,i_t}^\ast := \min\left\{\frac{1}{(T-t)\mu_{i_t}}, 1\right\}.
\end{equation}
If the audit reveals a misreport (that is, $o_t=1$ and $u_{t,i_t}\neq v_{t,i_t}$), the agent is permanently eliminated as \Cref{line:simple mech elimination} in \Cref{alg:simple_mechanism}.
We assert that this adaptive scheme effectively aligns agents' incentives with the planner's welfare objective, achieving $\text{poly}(K)$ regret using a logarithmic number of audits in expectation.

\paragraph{$\bm M^\ast$ Ensures $\text{poly}(K)$ Regret.}
Consider an agent $i\in [K]$ deciding their report $v_{t,i}$ for round $t\in [T]$ after observing their private utility $u_{t,i}$ and observable history $\mH_t^i$. Similar to the analysis of $\bm M_0(p)$ in \Cref{sec:simple mechanism}, we analyze the trade-off between \textit{i)} the immediate gain from misreporting, and \textit{ii)} the potential future loss due to elimination.
Notationally, let $V_t$ denote agent $i$'s expected future utility from round $t+1$ onward should they remained alive. For the purpose of this heuristic derivation, we treat $V_t$ as a constant independent to current-round reports.
Agent $i$ faces two primary options:
\begin{enumerate}
\item \textbf{Truthful} ($v_{t,i}=u_{t,i}$). The agent may lose the current round (if their $u_{t,i}$ is not high enough) and thus do not score any immediate gain, but they always stay alive. The expected total payoff is therefore at least:
\begin{equation*}
\text{Payoff}_{\text{truthful}} \approx \E[u_{t,i}\cdot\1[\text{win}]] + V_t \ge V_t.
\end{equation*}
\item \textbf{Misreporting} ($v_{t,i}=1$). This ensures an immediate gain of $u_{t,i}$, but risks being audited and eliminated with probability $p_{t,i}^\ast$. The expected total payoff is therefore upper bounded by:
\begin{equation*}
\text{Payoff}_{\text{misreport}} \le 1 + (1-p_{t,i}^\ast)V_t.
\end{equation*}
\end{enumerate}

Agent $i$ prefers truthful reporting when $\text{Payoff}_{\text{truthful}} \ge \text{Payoff}_{\text{misreport}}$, which simplifies to the condition $p_{t,i}^\ast \ge 1/V_t$.
Crucially, if near-truthful reports are also elicited for all other agents---as in our equilibrium strategy $\pi^\ast$ in \Cref{thm:deterministic main thm}---the expected future utility $V_t$ is simply the sum of expected per-round utilities:
\begin{equation*}
V_t \approx \sum_{\tau=t+1}^T \mu_i = (T-t)\mu_i,
\end{equation*}
where $\mu_i$ is defined in \Cref{eq:fair share informal}.
Thus $p_{t,i}^\ast \approx \frac{1}{(T-t)\mu_i}$ is sufficient to deter misreporting for agents with high future utilities ($V_t \ge 1/p_{t,i}^\ast$). For the remaining agents, while the condition $p_{t,i}^\ast\ge 1/V_t$ cannot hold and thus they are free to misreport, \Cref{eq:ideal mech check prob} gives $p_{t,i}^\ast=1$---so misreporting results in an audit and elimination almost surely.
Thus under $\bm M^\ast$, there is an equilibrium $\pi^\ast$ where agents either \textit{i)} report truthfully, or \textit{ii)} misreport, win, and eliminated. Since each agent can at most be eliminated once, the total regret is bounded by $\O(K)$.

\paragraph{$\bm M^\ast$ Ensures Logarithmic Number of Audits.}
Moreover, after eliciting (mostly) truthful behavior from agents, this specific choice of $p_{t,i}^\ast$ in \Cref{eq:ideal mech check prob} naturally gives a logarithmic number of audits in expectation.
Analogous to $\mu_i$, let $q_i$ denote the \emph{ex-ante winning probability} that agent $i$ wins under the first-best allocation:
\begin{equation}\label{eq:fair winning prob informal}
    q_{i}:=\Pr_{\bm u\sim \mV}\{u_i>u_{j},\forall j\ne i\},\quad \forall i\in [K].
\end{equation}
Since agents are (mostly) truthful, agent $i$ wins roughly $q_i$ fraction of the rounds. Thus,
\begin{align*}
    \mB_T = \E\left [\sum_{t=1}^T o_t\right ] &= \sum_{t=1}^T \sum_{i=1}^K \Pr\{\text{agent } i \text{ wins}\} \cdot p_{t,i}^\ast \\
    &\approx \sum_{t=1}^{T} \sum_{i=1}^K q_{i} \cdot \frac{1}{(T-t) \mu_i}.
\end{align*}
Recall from \Cref{assumption:min report} that at least one agent has utility bounded away from zero by $c$. Thus, in order for any agent $i$ to win under the first-best allocation, we must have $u_i\ge c$. By comparing \Cref{eq:fair share informal,eq:fair winning prob informal}, we immediately have $q_i \le \mu_i/c$. Plugging this into the above equality yields:
\begin{equation*}
    \mB_T \lesssim \sum_{i=1}^K \frac{\mu_i}{c \mu_i} \sum_{t=1}^{T} \frac{1}{T-t} = \frac{K}{c} \sum_{k=1}^{T} \frac{1}{k} = \O\left (\frac{K\log T}{c}\right ).
\end{equation*}

While $\bm M^\ast$ offers desirable theoretical guarantees, it relies on the ex-ante first-best utilities $\bm\mu=(\mu_i)_{i\in [K]}$, which are unknown in advance. In practice, $\bm \mu$ must be learned online via agents' historical reports. However, simply substituting empirical estimates $\hat{\bm \mu}$ into the mechanism fails, because strategic agents can manipulate their reports to bias these estimates. We now detail our solution to this challenge of online estimation.

\subsection{Incentive-Aligned Flagging for Online Estimation}\label{sec:description of est and flag}
To practically implement $\bm M^\ast$ without distributional priors, we must estimate the $p_{t,i}^\ast \approx \frac{1}{(T-t)\mu_i}$ in \Cref{eq:ideal mech check prob} via historical data. Directly estimating the ex-ante first-best utility $\mu_i=\E_{\bm u\sim \mV}[u_i \1[u_i>u_j,\forall j\ne i]]$ is fragile, as it requires observing exact utility values that are easily manipulated by agents. In \Cref{alg:mechanism}, we instead opt for the ex-ante winning probability $q_i$---defined in \Cref{eq:fair winning prob informal} as $q_i=\Pr_{\bm u\sim \mV}\{u_i>u_j,\forall j\ne i]\}$---whose plug-in estimation relies only on the \emph{identity} of the winner $i_t$ but not their \emph{report}.
These two quantities are closely related under \Cref{assumption:min report}: since $\1[u_i>u_j,\forall j\ne i]$ implies $u_i>u_{i_0}\ge c$, we must have $\mu_i \ge c \cdot q_i$. Establishing a reliable estimate $\hat{q}_i \approx q_i$ therefore enables us to set the audit probability as $\hat{p}_{t,i} \approx \frac{1}{(T-t) c \hat{q}_i}$.

\paragraph{Challenge of Strategic Environments.}
In a hypothetical scenario where all agents are non-strategic and truthful, online estimating $q_i$ is straightforward: for any $i\in [K]$, the intervals between two wins of agent $i$ follow a geometric distribution with parameter $q_i$. A sample size of $\O(\log T)$ suffices to w.h.p. ensure $\hat{q}_i$ is within a constant multiplicative factor of $q_i$, thus giving an audit probability $\hat p_{t,i}$ close enough to $p_{t,i}^\ast$.

However, in the strategic case, historical data---for example, the past winners $\{i_\tau\}_{\tau<t}$---is generated by self-interested agents who \emph{anticipate} how current outcomes affect future estimates. This introduces a correlation between agents' strategic behavior and the planner's online estimation, and standard concentration bounds cannot apply.
In fact, even if agents do not directly manipulate their reports (which directly affects $\hat q_i$), they can still selectively deviate based on the \emph{conditional failure probability} of the estimator.

To see this, consider a heuristic model where $\hat q_i$ is the empirical mean of a stochastic process $X_1, X_2, \ldots, X_N$ (e.g., interval between consecutive wins), such that in non-strategic settings this estimation is w.h.p. close:
\begin{equation}\label{eq:trivial iid concentration}
\Pr_{\text{non-strategic}}\left \{\hat q_i\not \in \left [\frac 12 q_i,2q_i\right ]\right \}\le \frac{1}{T}.
\end{equation}

Suppose an agent can influence some $X_n$'s at the risk of potential elimination. That is, right before each $X_n$ is generated, based on $X_1,X_2,\ldots,X_{n-1}$, this agent decides whether to use this ability and make $X_n$ arbitrary (if not, $X_n$ is identically distributed as the one in \Cref{eq:trivial iid concentration}). This agent can adopt the following strategy:
\begin{enumerate}
\item Stay truthful for the first $N/2$ rounds, i.e., $X_1,X_2,\ldots,X_{N/2}$ are identically distributed as in \Cref{eq:trivial iid concentration}.
\item Compute the conditional failure probability, namely $Y:=\Pr\{\hat q_i\not\in [\frac 12 q_i,2q_i]\mid X_1,X_2,\ldots,X_{N/2}\}$.
\item If $Y$ is beyond a threshold, say $Y\ge 1/\sqrt T$, this agent deviates to a malicious strategy which gives a much higher individual utility whenever $\hat q_i\not \in [\frac 12 q_i,q_i]$. Otherwise, this agent remains truthful.
\end{enumerate}

Such a strategy greatly amplifies the probability that $\hat q_i\not \in [\frac 12 q_i,2q_i]$: \textit{i)} according to \Cref{eq:trivial iid concentration}, with probability $1/T$ (which is negligible), the plug-in estimator $\hat q_i$ fails on its own; and \textit{ii)} once this agent deviates to the malicious strategy---happening whenever $Y\ge 1/\sqrt T$---no guarantee on the relationship between $\hat q_i$ and $q_i$ can be derived.
We cannot say much about \textit{ii)}: the only thing we know from \Cref{eq:trivial iid concentration} is $\E[Y]\le 1/T$, which only gives $\Pr\{Y\ge 1/\sqrt T\}\le 1/\sqrt T$ according to Markov inequality.
Therefore, even in this heuristic example, due to one strategic agent, the failure probability of such $\hat q_i$ is already amplified from $1/T$ to $1/\sqrt T$.

\paragraph{Solution: Incentive-Aligned Flagging Component.}
To resolve the correlation between agents' strategic behavior and planner's online estimation, we introduce an \emph{incentive-aligned flagging component} that allows agents to trigger a reset if they find estimates biased. This component (\Cref{line:flagging} in \Cref{alg:mechanism}) leverages the agents' own self-interest to verify the estimation process. The incentive alignment works in two directions:
\begin{itemize}
\item \textbf{Case 1: Over-estimation ($\hat{q}_i \gg q_i$).} If agent $i$ is estimated to win more often than they should, their calculated audit probability $\hat{p}_{t,i}\approx \frac{1}{(T-t)c \hat q_i}$ becomes much smaller than the $p_{t,i}^\ast\approx \frac{1}{(T-t)\mu_i}$ used in $\bm M^\ast$. This motivates agent $i$ to misreport frequently, which reduces the winning chances of all other agents $j \ne i$. In this scenario, the ``victim'' agents $j$ have a strong incentive to flag and invalidate this biased estimate.
\item \textbf{Case 2: Under-estimation ($\hat{q}_i \ll q_i$).} On the other hand, if agent $i$ is estimated to win rarely, their audit probability $\hat{p}_{t,i}\approx \frac{1}{(T-t) c \hat q_i}$ becomes aggressively high, which unfairly penalizes agent $i$. In this scenario, agent $i$ themself has the incentive to flag this biased estimate and request a fair re-estimation.
\end{itemize}

By incorporating this flagging capability, \mechname ensures that any significant error in $\hat q_i$ is identified and corrected by the affected parties. In \Cref{thm:equivalence of PBE}, we prove that this design admits a PBE where agents adopt this flagging strategy (formalized in \Cref{def:well-behaved flagging}). This consequently secures the accuracy of $\hat{q}_i$ throughout the game, and recovers the regret guarantee of $\bm M^\ast$ without any distributional information.

\section{Analysis: Reduction to Auxiliary Games}\label{sec:sketch varying-p}

In this section, we outline the theoretical analysis of \mechname and prove the performance guarantees stated in \Cref{thm:deterministic main thm}.
Due to the absence of distributional information, standard techniques relying on the revelation principle are inapplicable. To overcome this, as another main technical contribution, we develop a \emph{reduction-based framework} that maps the complex extensive-form game to a tractable ``auxiliary game.''

We first outline the analytical challenges and our reduction-based framework in \Cref{sec:technical overview auxiliary game}.
We then formally define the auxiliary game in \Cref{sec:sketch varying-p Part I}, establish the correspondence between games in \Cref{thm:equivalence of PBE} (main steps proved in \Cref{sec:sketch varying-p Part I 1,sec:sketch varying-p Part I 2,sec:sketch varying-p Part I 3}), and finally derive the bounds on regret and audits in \Cref{sec:sketch varying-p Part II}.

\subsection{Challenge: Inapplicability of Revelation Principle}\label{sec:technical overview auxiliary game}
Standard approaches in non-monetary mechanism design typically leverage the revelation principle to restrict attention to direct-revelation mechanisms that admit truth-telling equilibria \citep{balseiro2019multiagent,gorokh2021monetary,blanchard2024near}.
The practical applicability of this reduction, however, hinges on the planner's knowledge of the underlying distribution. Formally, the revelation principle asserts that for any dynamic mechanism $\bm M$ and any \emph{fixed} prior distribution $\mV$, there exists a ``direct revelation'' mechanism $\bm M_{\mV}'$ that is Bayesian incentive-compatible (BIC) and implements the same equilibrium outcome.

Crucially, the constructed mechanism $\bm M_{\mV}'$ structurally depends on this specific distribution $\mV$.
In our setting, the planner has no distributional information and is thus unable to formulate a \emph{single} direct mechanism $\bm M$ that preserves truthfulness and efficiency across a \emph{broad class} of unknown distributions.
In fact, achieving welfare maximization under such strict BIC requirements is fundamentally impossible: as established by \citet{bergemann2005robust}, if a mechanism ensures BIC across \emph{all} possible distributions, it must be ex-post incentive-compatible (EPIC). In our private value setup, EPIC further reduces to dominant-strategy incentive-compatibility (DSIC).
To see why welfare maximization must fail under DSIC, consider a one-shot example with two agents having deterministic utilities $u_1 \equiv \frac 13$ and $u_2 \equiv \frac 23$, but the planner does not know the identities of the agents.
For a mechanism to be both perfectly incentive-compatible and welfare-maximizing, it must always allocate the resource to the $\frac 23$-agent, leaving the $\frac 13$-agent with zero utility.
However, since the planner is unaware of the priors, they cannot distinguish the genuine high type from a fabricated one.
Thus the $\frac 13$-agent can always profitably mimic the $\frac 23$-agent, secure a winning probability of at least $\frac 12$, and yield a strictly positive expected utility of $\frac 16$. This, therefore, violates dominant-strategy incentive-compatibility.

We therefore must give up the strict BIC requirement, but instead characterize a non-trivial PBE where agents are allowed to occasionally deviate from truthful reporting.
To tame the complexity of the full strategy space in the extensive-form game induced by \mechname, we develop a reduction-based framework.
We define a tractable \emph{auxiliary game}, where agents are restricted to a smaller class strategies (see \Cref{def:auxiliary game}). Letting $\pi^\ast$ be a PBE in this auxiliary game, our analysis framework proceeds in three steps:
\begin{enumerate}
\item \textbf{Equivalence of Utility Profiles.} Every joint strategy in this auxiliary game gives an equivalent strategy in the original game, in the sense that they induce an identical utility profile for every agent (\Cref{lem:equiv between actual and no-flagging}).
\item \textbf{Welfare Guarantee in the Auxiliary Game.} The auxiliary-game PBE $\pi^\ast$---despite containing a small amount of strategic deviations---yields high utility for every agent and high social welfare (\Cref{lem:V lower and upper bound}).
\item \textbf{Stability in the Original Game.} The original-game strategy induced by the auxiliary-game PBE $\pi^\ast$ is a PBE: since agents already secure high utility under $\pi^\ast$, even if they are allowed to adopt those strategies prohibited in the auxiliary game, they have no incentives to unilaterally deviate (\Cref{lem:min report with u is good}).
\end{enumerate}

This reduction framework allows us to bound the regret by analyzing the simpler auxiliary game while rigorously establishing equilibrium existence in the original game (\Cref{thm:equivalence of PBE}).
We believe this general analysis framework can be of independent interest to other mechanism design problems where the revelation principle is not directly applicable (see also \Cref{sec:noisy case} for extensions of our framework). Nevertheless, we remark that while the strategic behavior in $\pi^\ast$ does not hurt the social welfare regret $\mR_T$, it does affect the expected number of audits $\mB_T$; we will control this in \Cref{sec:sketch varying-p Part II} via agents' incentives in the original game.

\subsection{Defining the Auxiliary Game}\label{sec:sketch varying-p Part I}
We start by formalizing the $\mu_i$ and $q_i$ in \Cref{eq:fair share informal,eq:fair winning prob informal}, which ignored \textit{i)} the possible change in the alive agents set $\mA_t$ throughout the game, and \textit{ii)} the potential no-allocation of \mechname due to $\max_{i\in \mA_t} v_{t,i}<c$.
\begin{definition}[Ex-Ante First-Best]\label{eq:fair share}
For any $\mA\subseteq [K]$ and $i\in \mA$, the \emph{ex-ante first-best utility} is
\begin{equation*}
\mu_i(\mA)=\E_{\bm u\sim \mV} \Big [u_i\1[(u_{i}\ge c)\wedge (u_i>u_j,\forall j\in \mA\setminus \{i\})]\Big ].
\end{equation*}
Similarly, we define the \emph{ex-ante winning probability} for any $\mA\subseteq [K]$ and $i\in \mA$ as
\begin{equation*}
q_i(\mA) = \Pr_{\bm u\sim\mV}\Big \{(u_{i}\ge c)\wedge (u_i>u_j,\forall j\in \mA\setminus \{i\})\Big \}.
\end{equation*}
In each round $t$ of \mechname, we write $\mu_{t,i}:=\mu_i(\mA_t)\1[i\in \mA_t]$ and $q_{t,i}:=q_i(\mA_t)\1[i\in \mA_t]$ for any $i\in [K]$, where we recall from \Cref{line:elimination} of \Cref{alg:mechanism} that $\mA_t$ is the set of alive agents in round $t$.
\end{definition}

Before defining the auxiliary game, we provide some intuition.
Recall from \Cref{sec:description of check prob} that our adaptive audits deter misreporting by leveraging expected future utilities: those agents with $\E[\text{future utility}]\ge \hat p_{t,i}^{-1}$ hesitate to lie. However, as the game progresses, any agent's future utility diminishes, which eventually becomes too low such that $\E[\text{future utility}]\approx (T-t)\mu_i<1$. In this case, even the maximum threat of elimination---auditing with probability 1---is insufficient to incentivize truthfulness, which means the planner must give up on enforcing truthfulness for this agent.
\Cref{def:auxiliary game} formalizes this by explicitly \emph{prohibiting} over-reports when $\E[\text{future utility}]\gtrsim 1$ (the $p_{t,i}\le 1$ therein), and allowing arbitrary behavior afterwards.

\begin{definition}[Auxiliary Game]\label{def:auxiliary game}
Initially, alive agent set $\mA_1=[K]$. In round $t\in [T]$, each alive $i\in \mA_t$ reports a $v_{t,i}\in [0,1]$ that can only depend on their own utility $u_{t,i}$, round number $t$, and set of alive agents $\mA_t$---but not any other historical information.
Further, if $p_{t,i}\le 1$ where
\begin{equation*}
p_{t,i}:= \frac{1+K^2}{(T-t)q_{t,i} c},\quad q_{t,i}:=q_i(\mA_t),
\end{equation*}
agent $i$ is not allowed to ``over-report'', i.e., $v_{t,i}\le u_{t,i}$ must hold with probability 1.
After observing all reports $\bm v_t$, the mechanism allocates the resource to the alive agent with highest report, i.e., $i_t=\argmax_{i\in \mA_t} v_{t,i}$ (ties broken lexicographically as discussed in \Cref{sec:setup planner objective}). Then, if $v_{t,i_t}>u_{t,i_t}$, agent $i_t$ gets eliminated with probability 1 (instead of being probabilistic) so that $\mA_{t+1}=\mA_t\setminus \{i_t\}$; otherwise, $\mA_{t+1}=\mA_t$ is unchanged.
\end{definition}

We defer the formal definition of auxiliary-game strategies to \Cref{def:agents strategy auxiliary} in \Cref{sec:appendix auxiliary game}.
Compared to the original game induced by \mechname, the auxiliary game poses three key distinctions:
\begin{enumerate}
\item The report $v_{t,i}$ from agent $i$ can only depend on true utility $u_{t,i}$ and the \emph{simplified history} $\tilde \mH_t:=(t,\mA_t)$, whereas the original game allows dependency on their \emph{full observable history} $\mH_t^i$. \label{item:only simplified history}
\item Agents cannot over-report $v_{t,i}>u_{t,i}$ unless $p_{t,i}>1$, whereas in the original game they are free to do so. \label{item:no mark up unless p>=1}
\item Winners who over-report, i.e., $v_{t,i_t}>u_{t,i_t}$, are eliminated deterministically w.p. 1. In contrast, \mechname eliminates such agents with an adaptive probability $\hat p_{t,i}=\min \big (\frac{4(1+K^2)}{(T-t) \hat q_{t,i} c},1\big )\le 1$ (\Cref{line:check prob} of \Cref{alg:mechanism}).\label{item:eliminate immediately}
\end{enumerate}

We emphasize that these restrictions are only analytical tools for our auxiliary game and reduction-based analysis. In the actual game induced by \mechname, agents are free to violate them.
The core of our reduction is that if agents adopt a specific \emph{well-behaved flagging strategy}---formalizing the discussions in \Cref{sec:description of est and flag}---any PBE in the restricted auxiliary game translates to a PBE in the original game.
\begin{definition}[Well-Behaved Flagging Strategy]\label{def:well-behaved flagging}
The well-behaved flagging strategy $\pi^{\ast,f}=(\pi_{t,i}^{\ast,f})_{t\in [T],i\in [K]}$ is defined such that for any agent $i\in[K]$ and round $t\in[T]$, $\pi_{t,i}^{\ast,f}$ answers $f_{t,i}=1$ if and only if either \textit{i)} $i_t\ne i$ and $\hat q_{t+1,i_t}>4q_{t+1,i_t}$, or \textit{ii)} $i_t=i$ and $\hat q_{t+1,i_t}<q_{t+1,i_t}/4$.
\end{definition}

When agents follow $\pi^{\ast,f}$, they collectively ensure that the estimates $\hat q_{t,i}\in [\frac 14 q_{t,i},4q_{t,i}]$ remain accurate throughout the game, thus the \mechname audit probability $\hat p_{t,i}=\min\big (\frac{4(1+K^2)}{(T-t)\hat q_{t,i}c},1\big )$ is close to the auxiliary-game $p_{t,i}=\frac{1+K^2}{(T-t) q_{t,i} c}$ (and also consequently the $p_{t,i}^\ast=\frac{1}{(T-t)\mu_{t,i}}$ we used in $\bm M^\ast$; recall \Cref{sec:description of check prob}).

This makes the auxiliary game a good proxy for studying the complex original game.
In the next three subsections, we establish the following key result of our reduction-based analysis framework:
\begin{theorem}[Correspondence of Equilibrium]\label{thm:equivalence of PBE}
Let ${\pi}^{\ast,r}$ be a PBE in the auxiliary game from \Cref{def:auxiliary game} and $\pi^{\ast,f}$ be the well-behaved flagging strategy defined in \Cref{def:well-behaved flagging}. Then, the joint strategy $\pi^\ast=({\pi}^{\ast,r},{\pi}^{\ast,f})$ constitutes a PBE in the original game under \mechname.
\end{theorem}

Essentially, \Cref{thm:equivalence of PBE} says that no agent can gain by unilaterally deviating from $\pi^\ast=(\pi^{\ast,r},\pi^{\ast,f})$, even if they adopt a reporting strategy violating \Cref{item:only simplified history,item:no mark up unless p>=1,item:eliminate immediately} and/or deviate from the well-behaved $\pi^{\ast,f}$.
The deviation from the flagging strategy $\pi^{\ast,f}$ is relatively easily deterred: Since $\pi^{\ast,f}$ is designed to trigger resets when estimations are biased against the agent's interest, adhering to it is naturally incentive-compatible (see \Cref{sec:equivalence of PBE formal} for the formal proof).
Consequently, our main analytical task is to verify that any unilateral deviation in the \emph{reporting strategy} $\pi^{\ast,r}$ is unprofitable.
This is established in the following three subsections.

\subsection{Equivalence of Utilities between Auxiliary- and Original-Game Strategies}\label{sec:sketch varying-p Part I 1}
We first establish that any strategy in the auxiliary game maps to a strategy in the original game---via the well-behaved flagging strategy---yielding identical expected utilities for every agent.

\begin{lemma}[Equivalence of Utility Profiles; Informal \Cref{lem:equiv between actual and no-flagging formal}]\label{lem:equiv between actual and no-flagging}
For any auxiliary-game strategy $\pi^r$, let $\pi:=(\pi^r,\pi^{\ast,f})$ be the corresponding original-game strategy. For any round $t\in[T]$ and complete history $\mH_t\in H_t$, let $\mA_t$ be the set of alive agents at that time. For any alive agent $i\in \mA_t$, the following equal:
\begin{enumerate}
\item the expected utility of $i$ under $\pi$ in the original game starting from $\mH_t$, namely $V_i^{\pi}(\mH_t)$ in \Cref{eq:V-func general}; and
\item the expected utility of $i$ under $\pi^r$ in the auxiliary game starting from the simplified history $\tilde \mH_t:=(t,\mA_t)$.
\end{enumerate}
\end{lemma}

\begin{proof}[Proof Sketch]
Since the auxiliary game imposes strictly tighter restrictions than the original game, the strategy $\pi_{t,i}^r$ is always feasible in the original game provided the set of alive agents $\mA_t$ evolved identically.
Since the allocation rule $i_t=\argmax_{i\in \mA_t}v_{t,i}$ is shared between \mechname and the auxiliary game, to verify the equivalence of utilities, it suffices to show both games induce the same sequence of eliminations.

Firstly, truthful agents (or those who do not over-report) are never eliminated in either game.
Conversely, if the winner $i_t$ over-reports, the auxiliary game eliminates them with probability $1$ (\Cref{item:eliminate immediately}) while \mechname eliminates with probability $\hat p_{t,i_t}$. Thus it only remains to prove $\hat p_{t,i_t}=1$. From \Cref{item:no mark up unless p>=1}, over-reporting implies $p_{t,i_t} > 1$; from $\pi^{\ast,f}$ defined in \Cref{def:well-behaved flagging}, we always have $\hat q_{t,i} \le 4 q_{t,i}$. Therefore
\begin{align*}
\hat p_{t,i}&=\min\left (\frac{4(1+K^2)}{(T-t)\hat q_{t,i_t} c},1\right )\\&\ge \min\left (\frac{(1+K^2)}{(T-t)q_{t,i_t} c},1\right )=\min(p_{t,i},1)=1.
\end{align*}
Consequently, both the auxiliary game and \mechname eliminate agent $i$ with probability 1.
\end{proof}

The formal statement and proof is presented as \Cref{lem:equiv between actual and no-flagging formal}.
This correspondence result is only \emph{one-way} from the auxiliary game to the original game: for every auxiliary-game strategy $\pi^r$, the original-game strategy $\pi=(\pi^r,\pi^{\ast,f})$ generates identical utility profiles. However, the converse is not true: A valid strategy in the original game may violate \Cref{item:only simplified history,item:no mark up unless p>=1,item:eliminate immediately} and does not correspond to any auxiliary-game strategy.

\subsection{V-Function and Regret Bounds under Auxiliary-Game PBE}\label{sec:V lower and upper bound}\label{sec:sketch varying-p Part I 2}
With the correspondence established, we focus on analyzing the auxiliary game. Since \Cref{def:auxiliary game} only permits strategies satisfying \Cref{item:only simplified history,item:no mark up unless p>=1}, \emph{any} PBE $\pi^{\ast,r}$ therein guarantees a high expected utility for every agent.
Specifically, for any agent $i\in[K]$, due to the strict elimination rule in \Cref{item:eliminate immediately}, opponents' ability to ``win by over-reporting'' is limited. Thus, if agent $i$ unilaterally deviated to the truthful strategy, they would secure a high utility; by definition of PBE, this utility lower bounds agent $i$'s payoff under any PBE. We formalize this result as a sandwich-style bound in \Cref{lem:V lower and upper bound}.

\begin{lemma}[Bounds on V-Functions; Informal \Cref{lem:V lower and upper bound formal}]\label{lem:V lower and upper bound}
Let $\pi^{\ast,r}$ be a PBE in the auxiliary game (\Cref{def:auxiliary game}).
Under $\pi^\ast=(\pi^{\ast,r},\pi^{\ast,f})$ in the original game, for any round $t\in [T]$ and history $\mH_t$, the V-function of an alive agent $i\in \mA_t$ approximates their future cumulative first-best utility:
\begin{equation}
V_i^{{\pi}^\ast}(\mH_t)-(T-t+1) \mu_{t,i}\in [-K,K(K-1)],\label{eq:for alive}
\end{equation}
where $\mu_{t,i}$ is defined in \Cref{eq:fair share}.
For any eliminated agent $i\notin \mA_t$, we naturally have $V_i^{{\pi}^\ast}(\mH_t)=0$.
\end{lemma}
\begin{proof}[Proof Sketch]
In the auxiliary game, since $\pi^{\ast,r}$ is a PBE, the utility obtained by agent $i$ under $\pi^{\ast,r}$ is lower-bounded by that under the truthful reporting strategy $\truth_i$ (i.e., always reporting $v_{t,i}=u_{t,i}$ regardless of opponents' strategies). Fix any round $t\in [T]$ and history $\mH_t\in H_t$, and suppose that agent $i$ is alive (i.e., $i\in \mA_t$).
Consider a round $\tau\ge t$ where \textit{i)} $u_{\tau,i}>u_{\tau,j}$ for all $j\in \mA_t\setminus \{i\}$ (i.e., agent $i$ is the first-best winner among all alive agents), but \textit{ii)} $i_\tau\ne i$. This is the only case agent $i$'s expected utility can drop below their ex-ante first-best utility, namely $(T-t+1)\mu_{t,i}$.

In this case, the winner $i_\tau$ must have over-reported $v_{\tau,i_\tau}>u_{\tau,\tau}>u_{\tau,i_\tau}$. Due to \Cref{item:eliminate immediately} of \Cref{def:auxiliary game}, this winner $i_\tau$ is immediately eliminated from the auxiliary game. Therefore, such events can only happen for no more than $K$ times. In other words, in those rounds where $u_{\tau,i}>u_{\tau,j}$ for all $j\in \mA_t\setminus \{i\}$, agent $i$ always wins except for no more than $K$ rounds. This gives the auxiliary-game lower bound: under $\pi^{\ast,r}$, which is a PBE, agent $i$'s expected utility is at least $(T-t+1)\mu_{t,i}-K$. Using the equivalence from \Cref{lem:equiv between actual and no-flagging}, starting from $\mH_t$, agent $i$'s original-game expected utility under $\pi^\ast$ is also at least $(T-t+1)\mu_{t,i}-K$, i.e.,
\begin{equation*}
V_i^{\pi^\ast}(\mH_t)-(T-t+1)\mu_{t,i}\ge -K.
\end{equation*}

The upper bound part in \Cref{eq:for alive} follows from the fact that---by definition of ex-ante first-best utilities---the social welfare in any future round $\tau\ge t$ is upper-bounded by $\sum_{i\in \mA_\tau} \mu_{\tau,i}$. Therefore we have
\begin{equation*}
\sum_{i\in \mA_t} V_i^{\pi^\ast}(\mH_t)\le (T-t+1)\sum_{i\in \mA_t} \mu_{t,i}.
\end{equation*}
As every other agent $j\in \mA_\tau\setminus \{i\}$ also has an expected utility of at least $(T-t+1)\mu_{t,j}-K$, agent $i$ cannot claim more than $(T-t+1)\mu_{t,i}+K(K-1)$. This finishes the proof.
\end{proof}

\Cref{lem:V lower and upper bound} serves \emph{two} main purposes.
First, applying \Cref{eq:for alive} to the initial state of the game (namely $t=1$ and $\mA_1=[K]$) implies $\sum_{i=1}^K V_i^{\pi^\ast}(\mathcal H_1)\ge T \sum_{i=1}^K \mu_{1,i}-K(K-1)$. Since $\sum_{i=1}^K \mu_{1,i}$ equals the optimal per-round social welfare, this implies an $\O(K^2)$ regret bound for \mechname when agents all use $\pi^\ast$.
However, we are still one step away from proving the regret guarantee in \Cref{thm:deterministic main thm}: We need to justify that $\pi^\ast$ constitutes a PBE in the original game. While $\pi^{\ast,r}$ is an auxiliary-game PBE, the corresponding $\pi^\ast$ (despite yielding identical utility profiles for agents) is not necessarily an original-game PBE, since there are no \Cref{item:only simplified history,item:no mark up unless p>=1,item:eliminate immediately}. This brings us to the second use of \Cref{lem:V lower and upper bound}: By upper-bounding the gain from misreports and lower-bounding the utility when being truthful, misreporting is unprofitable.

\subsection{Establishing the PBE Result (Proving \Cref{thm:equivalence of PBE})}\label{sec:sketch varying-p Part I 3}
To prove that $\pi^\ast$ remains a PBE in the original game---where \Cref{item:only simplified history,item:no mark up unless p>=1,item:eliminate immediately} are absent---we must demonstrate that agents gain no advantage by violating them.
We first show that the restriction on over-reporting (i.e., \Cref{item:no mark up unless p>=1}) is consistent with agents' own incentives.
Specifically, we fix an auxiliary-game PBE $\pi^{\ast,r}$ and corresponding original-game strategy $\pi^\ast=(\pi^{\ast,r},\pi^{\ast,f})$. We now consider any agent $i\in [K]$'s unilateral deviation strategy $\pi_i=(\pi^r_{i},\pi_i^{\ast,f})$ \emph{in the original game} that \textit{i)} adheres to \Cref{item:only simplified history} (i.e., the report $v_{t,i}$ only depend on utility $u_{t,i}$, round number $t$, and set of alive agents $\mA_t$); but \textit{ii)} violates \Cref{item:no mark up unless p>=1} (i.e., agent $i$ may over-report in some round $t$ even if $p_{t,i}=\frac{1+K^2}{(T-t)q_{t,i}c}\le 1$).

Since this $\pi_i^r$ is \emph{invalid} in the auxiliary game by violating \Cref{item:no mark up unless p>=1}, it is \emph{not} automatically inferior than $\pi_i^{\ast,r}$ although $\pi^{\ast,r}$ is a PBE therein. To justify that in the original game $\pi_i$ is no better than $\pi_i^\ast$, we construct a strategy $\pi_{i}^{\prime}=(\pi_i^{\prime,r},\pi_i^{\ast,f})$ such that \textit{i)} in the original game, $\pi_{i}^{\prime}$ is better than $\pi_i$ for agent $i$ when opponents use $\pi_{-i}^{\ast}$; and \textit{ii)} $\pi_i^{\prime,r}$ is valid in the auxiliary game, and is thus dominated by $\pi_i^{\ast,r}$. This gives:
\begin{lemma}[Over-Reporting is Unprofitable; Informal \Cref{lem:min report with u is good formal}]\label{lem:min report with u is good}
For agent $i\in [K]$, let $\pi_i^{r}$ violate \Cref{item:no mark up unless p>=1} but not \Cref{item:only simplified history}.
Construct an alternative strategy $\pi_{i}^{\prime,r}$ via a coupling with $\pi_{i}^r$: For private utility $u_{t,i}$ and corresponding report $v_{t,i}\sim \pi^r_{t,i}(u_{t,i};\mH_t^i)$, let $\pi_{t,i}^{\prime,r}(u_{t,i};\mH_t^i)$ report
\begin{equation*}
v_{t,i}':=\begin{cases}
\min(v_{t,i},u_{t,i}),&p_{t,i}\le 1;\\
v_{t,i},&p_{t,i}>1.
\end{cases}
\end{equation*}

In the original game, when opponents are using $\pi_{-i}^\ast$, agent $i$'s strategy of $\pi_i:=(\pi_{t,i}^r,\pi_{t,i}^{\ast,f})$ is dominated by $\pi_i':=(\pi_{t,i}^{\prime,r},\pi_{t,i}^{\ast,f})$. Consequently, agent $i$'s unilateral deviation from $\pi_i^\ast$ to $\pi_i$ is unprofitable.
\end{lemma}
\begin{proof}[Proof Sketch]
Using the performance difference lemma \citep[see \Cref{sec:appendix varying-p} for the detailed arguments]{kakade2002approximately}, it suffices to verify that any \emph{single-round deviation} is unprofitable: Namely, we argue that for any fixed round $t\in [T]$ and ``valid'' history $\mH_t$ (such that \Cref{lem:V lower and upper bound} is applicable; formalized in the appendix), if agent $i$ only unilaterally deviates \emph{in this round} from $\pi_{t,i}^\ast$ to $\pi_{t,i}$, it is unprofitable.

Indeed, in this case, since in all subsequent rounds every agent is using $\pi^\ast$, the correspondence result in \Cref{lem:equiv between actual and no-flagging} is applicable. Hence the only case where agent $i$'s expected future utility differs is when the simplified history $\tilde \mH_{t+1}=(t+1,\mA_{t+1})$ differs. We therefore only need to consider the case where the two reporting strategies of $\pi_{t,i}^r$ and $\pi_{t,i}^{\prime,r}$ give a different set of alive agents $\mA_{t+1}$. This implies two things: \textit{i)} agent $i$'s reports must be different (i.e., $v_{t,i}\ne v_{t,i}'$), which only happens when $p_{t,i}\le 1$ and $v_{t,i}>u_{t,i}$; and \textit{ii)} agent $i$ is the winner under (at least) one reporting strategy, which means $v_{t,i}>\max_{j\in \mA_t\setminus \{i\}} v_{t,j}$.

Therefore, let $p_{t,i}\le 1$. We fix a realization of utilities and reports such that
\begin{equation*}
v_{t,i}>u_{t,i}=v_{t,i}',\quad v_{t,i}>\max_{j\in \mA_t\setminus \{i\}} v_{t,j}.
\end{equation*}

We upper bound the expected utility under $\pi_{t,i}$. Given $v_{t,i}>\max_{j\in \mA_t\setminus \{i\}} v_{t,j}$, agent $i$ would win (i.e., $i_t=i$). Thus for this round, agent $i$'s utility is $u_{t,i}$, which is bounded by $[0,1]$. Agent $i$ is then audited with probability $\hat p_{t,i}$, in which case $\mA_{t+1}=\mA_t\setminus \{i\}$ (because agent $i$ over-reported a $v_{t,i}>u_{t,i}$). If no audits happen, we will have $\mA_{t+1}=\mA_t$. Hence the expected utility of agent $i$ under $\pi_{t,i}$ is upper bounded by
\begin{align}\label{eq:upper bound deviation informal}
&\quad u_{t,i}+\hat p_{t,i}\cdot 0+(1-p_{t,i})\cdot V_i^{\pi^\ast}(\mH_{t+1})\\
&\le 1+(1-\hat p_{t,i}) \left ((T-t)\mu_{t,i}+(K^2-K)\right ),
\end{align}
where we used the upper bound part of \Cref{lem:V lower and upper bound} and the fact that, if $\mA_{t+1}=\mA_t$, then $\mu_{t+1,i}=\mu_{t,i}$.

Using the lower bound part of \Cref{lem:V lower and upper bound}, we further lower bound the expected utility under $\pi_{t,i}'$. While agent $i$ may not win the current round, since $\pi_{t,i}'$ reported $v_{t,i}'=\min(v_{t,i},u_{t,i})=u_{t,i}$, they are never eliminated. This implies $\mu_{t+1,i}\ge \mu_{t,i}$. Hence agent $i$'s expected utility under $\pi_{t,i}'$ is at least
\begin{equation}\label{eq:lower bound deviation informal}
0+V_i^{\pi^\ast}(\mH_{t+1})\ge (T-t)\mu_{t,i}-K.
\end{equation}

Comparing the RHS of \Cref{eq:upper bound deviation informal,eq:lower bound deviation informal}, it only remains to prove $1+(1-\hat p_{t,i}) ((T-t)\mu_{t,i}+K(K-1))\le (T-t)\mu_{t,i}-K$. It thus suffices to prove $\hat p_{t,i} \ge \frac{1+K^2}{(T-t)\mu_{t,i}}$.
Comparing the RHS to the definitions of $p_{t,i}=\frac{1+K^2}{(T-t)q_{t,i} c}$ in \Cref{def:auxiliary game} and $\hat p_{t,i}=\min(\frac{4(1+K^2)}{(T-t)\hat q_{t,i} c},1)$ in \Cref{line:check prob} of \Cref{alg:mechanism}, this inequality indeed holds:
\begin{align*}
\hat p_{t,i}&=\min\left (\frac{4(1+K^2)}{(T-t)\hat q_{t,i} c},1\right )\\
&\overset{(a)}{\ge} \min\left (\frac{1+K^2}{(T-t) q_{t,i} c},1\right )\\
&\overset{(b)}{=}p_{t,i}\overset{(c)}{\ge} \frac{1+K+K(K-1)}{(T-t) \mu_{t,i}},
\end{align*}
when (a) uses $\hat q_{t,i}\le 4q_{t,i}$ (due to the well-behaved flagging strategy $\pi^{\ast,f}$ in \Cref{def:well-behaved flagging}); (b) uses $p_{t,i}\le 1$ (from the condition that $v_{t,i}'\ne v_{t,i}$); and (c) uses $\mu_{t,i}\ge q_{t,i} c$ (implied by \Cref{assumption:min report}). This gives the first claim: For agent $i$, when opponents are using $\pi_{-i}^\ast$, the original-game strategy $\pi_i=(\pi_i^r,\pi_i^{\ast,f})$ is dominated by $\pi_i'=(\pi_i^{\prime,r},\pi_i^{\ast,f})$, which caps the reports at true utilities when $p_{t,i}\le 1$.

Since $\pi_{i}^{\prime,r}$ is a valid auxiliary-game strategy and $\pi^{\ast,r}$ is a PBE therein, the correspondence result in \Cref{lem:equiv between actual and no-flagging} further suggests that $\pi_i^\prime$ is dominated by $\pi_i^\ast$ in the original game. This finishes the proof.
\end{proof}

Having established that over-reporting is unprofitable, we subsequently address the history dependency (\Cref{item:only simplified history}).
Notice that: \textit{i)} all other agents follow $\pi_{-i}^{\ast}$, whose reporting strategy only depends on their own utility and the simplified history; and \textit{ii)} when agent $i$ refrains from over-reporting (as per \Cref{lem:min report with u is good}), the original-game mechanism \mechname behaves identically to the auxiliary-game mechanism (i.e., audits and eliminations occur if and only if agent $i$ wins by over-reporting).
Consequently, the system dynamics are entirely determined by $\tilde \mH_t=(t,\mA_t)$, making any strategic dependence on the full history $\mH_t^i$ redundant.
Thus no unilateral deviation is profitable, and $\pi^\ast=(\pi^{\ast,r},\pi^{\ast,f})$ constitutes a PBE in the original game. This completes the proof of \Cref{thm:equivalence of PBE}. The formal version of the above arguments resides in \Cref{sec:equivalence of PBE formal}.

\subsection{Bounding Expected Number of Audits (Proving \Cref{thm:deterministic main thm})}\label{sec:sketch varying-p Part II}
We now derive the performance guarantees of \mechname stated in \Cref{thm:deterministic main thm}.
As discussed in the end of \Cref{sec:V lower and upper bound}, \Cref{lem:V lower and upper bound} and \Cref{thm:equivalence of PBE} already provides the desired bound on regret $\mR_T$. Indeed, from \Cref{lem:V lower and upper bound} the social welfare---the sum of each agent's utility---achieved by \mechname under $\pi^\ast$ satisfies
\begin{equation*}
\sum_{i=1}^K V_i^{\pi^\ast} (\mH_1) \geq T\sum_{i=1}^K\mu_{1,i} - K^2.
\end{equation*}
Since $\pi^\ast$ is an original-game PBE according to \Cref{thm:equivalence of PBE}, and by definition $T\sum_{i=1}^K\mu_{1,i}$ is the maximum expected social welfare that any mechanism cannot exceed, \mechname induces a PBE such that $\mR_T\le K^2$.

Controlling the expected number of audits $\mB_T$ requires additional care: The auxiliary game in \Cref{def:auxiliary game} only prohibits \emph{over-reporting} $v_{t,i}>u_{t,i}$. While \textit{under-reporting} $v_{t,i}<u_{t,i}$ does not affect the regret---thanks to the sandwich result presented in \Cref{lem:V lower and upper bound}---it can make the estimates $\hat q_{t,i}$ biased since agents lose some rounds where they are actually first-best. Should this result in a biased $\hat q_{t,i}\not\in [\frac 14 q_{t,i},4q_{t,i}]$, the flagging component (\Cref{line:flagging} in \Cref{alg:mechanism}) would be triggered, thus requiring extra audits on re-estimating $\hat q_{t,i}$.
Fortunately, under-reports only happen rarely: It results in an immediate utility loss while only gives very limited gain (if any), as we characterize in the following lemma:
\begin{lemma}[Number of Under-Reports; Informal \Cref{lem:number of under-reporting bound formal}]\label{lem:number of under-reporting bound}
For round $t\in [T]$ and agent $i\in [K]$, let $D_{t,i}$ be the indicator of the event that they are alive, do not win, but would have won had they reported honestly, i.e.,
$D_{t,i}=\1[i\in \mA_t]\1[i_t\ne i]\1[u_{t,i}\geq c]\1[u_{t,i}>v_{t,j},\forall j\in \mA_t\setminus\{i\}].$ Then we claim that
\begin{equation*}
\E_{\pi^\ast,\mechname}\sqb{\sum_{t=1}^T \sum_{i=1}^K D_{t,i}} \leq \frac{2K^3}{c}.
\end{equation*}
\end{lemma}
\begin{proof}[Proof Sketch]
For a rational agent $i$ to miss a win in a round $t$---which immediately loses utility $u_{t,i}$---the only reason is that the new history $\mH_{t+1}$ may be more favorable. Due to the correspondence in \Cref{lem:equiv between actual and no-flagging}, the expected future utility of agent $i$ only depends on $(t+1,\mA_{t+1})$, i.e., the round number and the set of alive agents. Hence, there must exist some $j\in \mA_t\setminus \{i\}$ who \textit{i)} would probably win-by-lying in this round; and \textit{ii)} the resulting $\mA_{t+1}=\mA_t\setminus \{j\}$ is more favorable for $i$. However, due to the sandwich bounds on $V_i$ in \Cref{lem:V lower and upper bound} and because any such $j$ can only be eliminated once, the utility gain due to eliminations is limited. Very informally, $\E[\sum_{t=1}^T \text{gain} \1[i_t\text{ eliminated}]]\le K^3$.

On the other hand, every time agent $i$ misses a win they suffer an immediate loss of $u_{t,i}\ge c$ (this is because of \Cref{assumption:min report}, which ensures the winning utility is at least $c$). Comparing the immediate loss to the potential future gain, we can upper bound the expected number of under-reports throughout the game, which gives the claim. Note that, however, all above events are in fact probabilistic and require extensive care. The formal statement and proof can be found in \Cref{sec:appendix regret and audits} as \Cref{lem:number of under-reporting bound formal}.
\end{proof}

We now control $\mB_T$.
For each agent $i\in[K]$, \mechname proceeds alternatively between an \emph{estimation phase} where $\hat q_{t,i}=0$ (where the mechanism is estimating the ex-ante winning probability $q_{t,i}$; during this phase, agent $i$ is audited with probability $1$ whenever they win), and an \emph{incentive-compatible} phase where the mechanism sticks to an estimated $\hat q_{t,i}>0$ after no agent answered $f_{t,i}=1$ (during this phase, agent $i$ is audited with probability $\hat p_{t,i}\leq 1$ as defined in \Cref{line:check prob} of \Cref{alg:mechanism}). We control them differently.

After proving \Cref{lem:number of under-reporting bound}, we can conclude that the estimation phase of each agent cannot last too long---roughly speaking, it should be the number of under-reports times $\O(\log T)$. On the other hand, the incentive-compatible phases are relatively straightforward, because agents under the PBE $\pi^\ast$ always use the well-behaved flagging strategy (recall \Cref{def:well-behaved flagging}) and thus the $\hat q_{t,i}$ used to construct $\hat p_{t,i}$ are multiplicatively close to the true $q_{t,i}$. Hence, $\hat p_{t,i}\lesssim \frac{\poly(K)}{(T-t)\mu_{t,i}}$ during this phase. Similar arguments as for the full-information mechanism $\bm M^\ast$ in \Cref{sec:description of check prob} then control the expected number of audits during incentive-compatible phases by $\O(\poly(K) \log T)$.
The formal arguments for $\mB_T(\pi^\ast,\mechname)$ is postponed to \Cref{sec:appendix regret and audits}.

\section{Hardness Results and Robustness Guarantees}\label{sec:lower bounds} \label{sec:noisy case}
In this section, we first introduce hardness results forming the red region in \Cref{fig:main_plot}, justifying the necessity of considering a trade-off between the regret and the expected number of audits. We then analyze our \mechname mechanism under two \emph{imperfect audit models}---where audit outcomes may be either adversarially manipulated or stochastically noisy---different from the perfect one in \Cref{assump:noiseless}. By keeping the \mechname unchanged and generalizing the analyses only, we demonstrate the robustness of our mechanism and our reduction-based auxiliary game analysis framework.

\subsection{Hardness Results: Unavoidability of Both Regret and Audits}
To complement the trade-off between regret and number of audits, in \Cref{thm:lower_bounds}, we give hardness results that: \textit{i)} one cannot make the regret completely independent to $K$, and \textit{ii)} to avoid $\text{poly}(T)$ regret, the expected number of audits must be $\Omega(1)$. The proof of \Cref{thm:lower_bounds} can be found in \Cref{sec:appendix lower bounds}.

\begin{theorem}[Hardness Results]
\label{thm:lower_bounds}
Fix $K\geq 2$. There exists utility distributions $\mV=\{\mV_i\}_{i\in [K]}$ satisfying \Cref{assumption:min report} with $c=\frac{1}{3}$, such that for any planner mechanism $\bm M$ and any corresponding PBE of agents $\pi$, we have $\mR_T(\pi,\bm M) =\Omega(K)$. Furthermore, even if $K=2$, there exists some $\{\mV_i\}_{i\in [K]}$ satisfying \Cref{assumption:min report} with $c=\frac 13$ and a constant $B>0$ such that if $\mB_T(\pi,\bm M) \leq B$ then $\mR_T(\pi,\bm M) =\Omega(\sqrt {T/\log T})$.
\end{theorem}

Remarkably, \Cref{thm:lower_bounds} holds even if the planner has access to agents' utility distributions (which is easier than our distribution-free setup). Therefore, the hardness results in \Cref{thm:lower_bounds} are due to the lack of \emph{direct punishments} of untruthful behavior. Indeed, should distributional information be available and should the audits be \emph{ex-ante} (i.e., the planner audits a potential agent \emph{before} deciding the allocation, thus misreports are punished in full by non-allocation), the first-best allocation would be implementable in one-shot \citep{ben2014optimal}. (While a few works also consider ex-post audits with limited punishments, they assume the planner can \emph{revoke} part of the allocation and are easier than our setting; see \Cref{sec:related work CSV}.) It is an important research question to characterize the hardness due to unknown utility distributions.

\subsection{Robustness of \mechname: Adversarial Audit Model}\label{sec:adv noise model}
While \mechname is designed under the perfect audit model (\Cref{assump:noiseless}) where audits perfectly reveal true utilities of winners, our mechanism as well as our reduction-based analysis framework extend to \emph{imperfect audit models}, where the audit outcome may differ from the true utility. While ``noisy verification'' resonates with the probabilistic verification model studied by \citet{caragiannis2012mechanism} and \citet{ball2019probabilistic}, the two main challenges distinguishing our work from existing ones on costly state verification---namely the absence of distributional information and the lack of direct punishments to untruthful behavior---still apply.

The first model, stated as \Cref{assump:adv noise model}, captures the possible strategic behaviors when facing audits:
knowing they are being audited, agents in reality may distort the audit process to hide untruthful behavior. Nevertheless, their manipulation abilities are usually limited to a specific extent, which we capture by $\sigma$:
\begin{assumption}[Adversarial Audit Model]\label{assump:adv noise model}
In a round $t\in [T]$ where the planner decides to audit the winner (i.e., $o_t=1$), the winner $i_t\in [K]$ strategically---possibly with any randomization---decides an additive noise $\eta_t\in [-\sigma,\sigma]$ where $\sigma\ge 0$ is a fixed constant, making the audit outcome $w_t=o_t(u_{t,i_t}+\eta_t)$.
\end{assumption}

We remark that in \Cref{assump:adv noise model}, $\sigma\ge 0$ is a constant fixed before the start of the game; specifically, when $\sigma=0$, it recovers the noiseless case (\Cref{assump:noiseless}) that we studied in previous sections. The constant $\sigma$ is public information among all agents, but does \emph{not} need to be known by the planner.
The mechanism deployed by the planner is exactly the same \mechname, and the parameter $\sigma$ only appears in the analysis.

Perhaps surprisingly, adversarial noises are easier to confront compared to stochastic ones---another imperfect audit model we study in \Cref{sec:stoc noise model}---because, informally, every single agent will try their best to utilize the adversarial manipulation ability by reporting $v_{t,i}=u_{t,i}+\sigma$. Consequently, the relative order between reports is preserved and our mechanism remains effective. We formalize this claim in \Cref{thm:adv noise model}.
\begin{theorem}[\mechname under Adversarial Audit Model]\label{thm:adv noise model}
Consider the adversarial audit model in \Cref{assump:adv noise model}.
Under the mechanism \mechname defined in \Cref{alg:mechanism}, for any utility distributions $\{\mV_i\}_{i\in[K]}$ satisfying \Cref{assumption:min report}, there exists a PBE of agents' strategies $\pi^\ast$ such that
\begin{align*}
\mR_T(\pi^\ast,\mechname)&\leq K^2,\\\mB_T(\pi^\ast,\mechname)&= \O\left (\frac{K^3}{c} \log T\right ).
\end{align*}
\end{theorem}

To prove \Cref{thm:adv noise model}, instead of starting from scratch, we generalize the auxiliary-game analysis developed in \Cref{sec:sketch varying-p} by allowing a \emph{mark-up tolerance} in the auxiliary game (see \Cref{sec:appendix auxiliary game} for the formal definition). Formalizing the intuitive claims above, in \Cref{sec:appendix varying-p,sec:appendix regret and audits}, we show that not only our \mechname mechanism is robust to adversarial noises (and also stochastic ones discussed in the next section), our auxiliary-game analysis framework \emph{also} seamlessly extend to these imperfect audit models.

\subsection{Robustness of \mechname to Stochastic Noises}\label{sec:stoc noise model}
We now consider another realistic issue that audits---due to objective limitations---cannot be always precise. For example, with a low but non-negligible probability (captured by parameter $\epsilon$), expert accountants may miss some discrepancies in the account book. Such randomly appearing errors are captured in \Cref{assump:stoc noise model}:

\begin{assumption}[Stochastic Audit Model]\label{assump:stoc noise model}
In a round $t\in [T]$ where the planner decides to audit the winner (i.e., $o_t=1$), with probability $\epsilon\in [0,1)$ the outcome $w_t$ may lie anywhere between $[0,1]$; otherwise, the outcome $w_t$ equals the true utility of the winner $u_{t,i_t}$. We allow the outcome distribution to depend on any history (even unobservable) up to that moment, namely $\{(\bm u_\tau,\bm v_\tau,i_\tau,o_\tau,w_\tau)\}_{\tau<t}\cup \{(\bm u_t,\bm v_t,i_t,o_t)\}$.
\end{assumption}

In \Cref{assump:stoc noise model}, $\epsilon\in [0,1)$ is a constant fixed before the start of the game; when $\epsilon=0$, it recovers \Cref{assump:noiseless} which we studied in previous sections. Again, this $\epsilon$---as well as its existence---does \emph{not} need to be known by the planner. That is, the planner still deploys the same \mechname mechanism, which is crafted for the perfect audit model. The performance guarantees in \Cref{thm:adv noise model,thm:stoc noise model} demonstrate the robustness of our \mechname mechanism: Even if the environment is misspecified (in the sense that audits can be noisy), the social welfare is still high, and the expected number of audits remains logarithmic in $T$.

Different from the adversarial case which essentially shifts all utility distributions by $+\sigma$ and thus only has minimal impact on agents' incentives (and consequently regret and expected number of audits), under \Cref{assump:stoc noise model}, we aim for the weaker equilibrium concept of Approximate Bayesian Equilibrium (ABE):
\begin{definition}[Approximate Bayesian Equilibrium]\label{def:approx PBE}
Under mechanism $\bm M$, a joint strategy $\pi=(\pi_i)_{i\in [K]}$ is a $\Delta$-Approximate Bayesian Equilibrium ($\Delta$-ABE in short) if the unilateral deviation of any agent $i\in [K]$ to any alternative strategy $\pi_i'$ does not increase their value function $V_i$ by more than $\Delta$. That is, let $\pi'$ be the joint strategy where agent $i$ follows $\pi'_i$ and any other agent $j\ne i$ follows $\pi_j$, then we must have
\begin{equation*}
V_i^{\pi'}(\mH_t;\bm M)\le V_i^{\pi}(\mH_t;\bm M)+\Delta,~~ \forall t\in [T],\mH_t\in H_t.
\end{equation*}
\end{definition}
Perfect Bayesian Equilibrium (PBE) in \Cref{def:PBE} can be viewed as a special case where $\Delta=0$. 
In non-monetary resource allocation setups, many works---especially those based on the artificial currency framework---have targeted for the ABE guarantee. Specifically, with perfect knowledge on agents' utility distributions, \citet{gorokh2021monetary} designed a non-monetary mechanism such that truthful reporting is an $\O_T(\sqrt{T \log T})$-ABE ensuring $\O_T(1)$ regret.
For our \mechname, we have the following guarantee:
\begin{theorem}[\mechname under Stochastic Audit Model]\label{thm:stoc noise model}
Suppose that \Cref{assump:stoc noise model} holds.
Under the mechanism \mechname defined in \Cref{alg:mechanism}, for any utility distributions $\{\mV_i\}_{i\in[K]}$ satisfying \Cref{assumption:min report}, there exists an $\O(\epsilon \frac{K^3}{c} T\log T)$-Approximate Bayesian Equilibrium $\pi^\ast$ such that
\begin{align*}
\mR_T(\pi^\ast,\mechname)&=\O\left (K^2+\epsilon \frac{K^3}{c} T\log T\right ),\\
\mB_T(\pi^\ast,\mechname)&=\O\left (\frac{K^3}{c}\log T\right ).
\end{align*}
\end{theorem}

Specifically, when $\epsilon\le 1/\sqrt T$, \Cref{thm:stoc noise model} gives an $\Otil_T(\sqrt{T})$-ABE with $\mR_T=\Otil_T(\sqrt{T})$ and $\mB_T=\Otil_T(1)$, giving a performance guarantee similar to that of \citet{gorokh2021monetary} but without distributional priors.
Technically, \Cref{thm:stoc noise model} fails to obtain a PBE because under stochastic noises, the auxiliary game no longer \emph{exactly corresponds} to the original game. We first overview the key technical challenge in obtaining such an exact correspondence, and how we succeeded under the perfect audit model (\Cref{assump:noiseless}):

Specifically, \Cref{item:only simplified history} of the auxiliary game (\Cref{def:auxiliary game}) requires agents \emph{only} to adapt to the simplified history $\tilde \mH_t=(t,\mA_t)$, i.e., the agents must be \emph{history-independent}; this property is pivotal since it drastically narrows the auxiliary-game strategy space. However, to make this restriction without loss of generality, we must also restrict the planner to be history-independent (see the end of \Cref{sec:sketch varying-p Part I 3}).
Under the perfect audit model, the history-independent auxiliary game exactly corresponds to the history-dependent \mechname, because in each round $t$ either: \textit{i)} the winner $i_t$ is truth-telling and remains alive deterministically; or \textit{ii)} the winner $i_t$ lies, is audited with probability 1, and is then eliminated with probability 1 (from \Cref{item:no mark up unless p>=1}, an over-reporting agent is audited with probability $\hat p_{t,i}\ge \min(p_{t,i})=1$).
Both cases result in a deterministic next-round state, and thus the history-independent auxiliary game gives an exact correspondence.

When the planner faces stochastic noises, however, this argument fails: Under \Cref{assump:stoc noise model}, a truth-telling agent $i$ can be wrongly eliminated in round $t$ with probability $\hat p_{t,i_t}\times \epsilon$, which is naturally history-dependent due to the adaptive $\hat p_{t,i_t}$ (defined in \Cref{line:check prob} of \Cref{alg:mechanism}). Due to this discrepancy, a history-independent auxiliary-game mechanism cannot \emph{exactly} correspond to the original-game \mechname. To accommodate this, in \Cref{sec:appendix varying-p}, we generalize our auxiliary-game analysis by allowing \emph{misspecified} auxiliary games: Instead of \Cref{thm:equivalence of PBE} which uses the exact correspondence to prove that every auxiliary-game PBE results in an original-game PBE, we now prove that an auxiliary-game PBE gives an original-game ABE, where the approximation parameter $\Delta$ is the maximum difference between an agent's expected utilities in both games under the same joint strategy (see \Cref{thm:equivalence of PBE formal} in \Cref{sec:equivalence of PBE formal} for the formal statement).

\section{Conclusion and Future Research}
This paper studies the dynamic resource allocation problem with strategic agents, under the dual constraints of prohibited monetary transfers and unknown distributional priors. By equipping the planner with costly ex-post audits---but without the ability of directly revoking allocations---we design a mechanism \mechname that attains $\O_T(1)$ social welfare regret using only $\O_T(\log T)$ audits in expectation. This result demonstrates that first-best efficiency---the maximum social welfare when all utilities are directly available to the planner---can be recovered in prior-free environments; the only instrument the planner needs is a logarithmic number of ex-post audits that cannot revoke allocations or impose any direct punishment.

The absence of distributional priors not only presents great algorithmic challenges but also limits the analytical tools.
Algorithmically, the planner faces the inherent difficulties of coupling online estimation with strategic incentives. That is, to implement the \emph{adaptive audits scheme} based on agents' ex-ante first-best utilities, the planner must estimate them online from endogenous reports that are subject to strategic manipulation. To bridge this gap, we design an \emph{incentive-aligned flagging component} which allows agents to assist the planner in correcting biased estimates, ensuring the estimation process is robust against strategic interference.
Analytically, the prior-free setting presents a more fundamental barrier: the revelation principle becomes inapplicable when utility distributions are non-unique and unknown. To overcome this, we develop a reduction-based \emph{auxiliary-game framework} to help identify a PBE consisting only of near-truthful strategies; this PBE, therefore, naturally yields low social welfare regret and a small number of audits in expectation. This analytical framework---by mapping the stochastic and strategic original game to a tractable deterministic counterpart---provides a structured way to bound strategic deviations in complex mechanisms and can be of independent interest to other robust mechanism design problems where standard analytical foundations fail.

We also complement our positive result with both hardness results and robustness guarantees. Capturing the fact that ex-post audits only introduce very limited punishment, we prove \textit{i)} an $\Omega(K)$ lower bound on the regret, and \textit{ii)} an $\Omega(1)$ lower bound on the expected number of audits when the regret is of order $o(\sqrt{T/\log T})$. Thus neither regret nor expected number of audits can be avoided. We further present robustness guarantees when audit outcomes are either adversarially manipulated or contaminated by stochastic noise, indicating the robustness of our \mechname mechanism and the generality of the auxiliary-game analysis framework.

There are several promising directions for future research. First, a gap remains between our $\O(K^2)$ regret upper bound and the $\Omega(K)$ lower bound. Even in our auxiliary game without estimation or audit---that is, any misreporting agent is eliminated immediately---this gap still persists. Closing this gap requires further refining the analysis of multi-agent unilateral deviations within the auxiliary game, particularly in how these deviations \emph{couple} across the interdependent histories induced by different agents' unilateral deviations.

Additionally, the exact Pareto frontier between the number of audits and regret is not yet fully characterized. Specifically, while Arrow's impossibility result suggests that $0$ audits imply $\Omega(T)$ regret, and our \Cref{thm:lower_bounds} dictates the existence of a constant $B$ such that no more than $B$ audits imply $\text{poly}(T)$ regret, it remains open whether $\O_T(1)$ regret can be sustained using a larger but still constant (in expectation) number of audits.

Moreover, the robustness of our mechanism under other imperfect audit models, such as heterogeneous adversarial manipulations or continuous Gaussian noise, requires further investigation. Such noise structures complicate the correspondence between the original and auxiliary games. Meanwhile, when such correspondence is not exact but rather approximate---as we faced in \Cref{sec:stoc noise model}---developing new techniques to establish Perfect Bayesian Equilibria rather than approximate equilibria remains a promising direction.

\bibliography{references}

\onecolumn
\newpage
\crefalias{section}{appendix}
\crefalias{subsection}{appendix}
\appendix

\startcontents[section]
\printcontents[section]{l}{1}{\setcounter{tocdepth}{2}}

\section{Simple Mechanism Auditing with Fixed Probability (\Cref{thm:simple mech main theorem})}\label{sec:appendix fixed-p}

We first analyze the simple mechanism $\bm M^0(p)$ which audits with a fixed probability $p\in(0,1]$ in every round; its pseudo-code is in \Cref{alg:simple_mechanism}. We now prove its guarantee \Cref{thm:simple mech main theorem}:

\begin{proof}[Proof of \Cref{thm:simple mech main theorem}]
Fix an agents' joint strategy PBE $\pi$ under the mechanism $\bm M^0(p)$ defined in \Cref{alg:simple_mechanism}. Since the mechanism always independently audits with probability $p$ in each round,
\begin{equation*}
\mB_T(\pi,\bm M^0(p)) = \sum_{t=1}^T \Pr\{o_t=1\} = pT.
\end{equation*}

Next, we fix an agent $i\in [K]$ and prove a lower bound on their V-function $V_i^{\pi}(\mH_1;\bm M^0(p))$---defined in \Cref{eq:V-func general}---obtained under joint strategy $\pi$ and mechanism \Cref{alg:simple_mechanism}. Denote by $\truth_i$ the truthful strategy of agent $i$ which always honestly reports $v_{t,i}=u_{t,i}$. Since $\pi$ is a PBE, we have
\begin{align*}
&\quad V_i^{\pi}(\mH_1;\bm M^0(p)) \geq V_i^{\truth_i\circ \pi_{-i}}(\mH_1;\bm M^0(p))= \E_{\truth_i\circ \pi_{-i}}\sqb{ \sum_{t=1}^T u_{t,i} \1[i_t=i] }\\
&\geq \E_{\truth_i\circ \pi_{-i}}\sqb{ \sum_{t=1}^T u_{t,i} \1[i\in \mA_t] \paren{\1[u_{t,i}>u_{t,j},\forall j\in \mA_t\setminus \{i\}] -  \1[i_t\neq i]\1[v_{t,i_t}>u_{t,i_t}]}}.
\end{align*}

Note that under $\truth_i$, agent $i$ is never eliminated. Hence, the indicator $\1[i\in \mA_t]=1$ holds for all $t\in [T]$.
We now define the per-round fair share of agent $i$ as
\begin{equation*}
    \mu_i := \E_{\bm u\sim \mV}\sqb{  u_i \1[u_i>u_j,\forall j\ne i]  }.
\end{equation*}
Since $\mA_t\subseteq [K]$, we know $(u_{t,i}>u_{t,j},\forall j\ne i)$ infers $(u_{t,i}>u_{t,j},\forall j\in \mA_t\setminus \{i\})$ for any $t\in [T]$ and $i\in \mA_t$. Therefore, we can further lower bound the V-function as
\begin{align*}
    V_i^{\pi}(\mH_1;\bm M^0(p)) &\geq \mu_i T - \sum_{j\neq i}\E \sqb{\sum_{t=1}^T \1[i_t=j] \1[v_j>u_j]}.
\end{align*}

Now we fix any opponent $j\neq i$. By construction of the mechanism, for every round $t\in [T]$ in which agent $j$ wins ($i_t=j$) by over-reporting ($v_{t,j}>u_{t,j}$), with probability $p$ independent of the past history, agent $j$ gets eliminated. Hence, $\sum_{t=1}^T \1[i_t=j] \1[v_j>u_j]$ is stochastically dominated by a geometric variable $\text{Geo}(p)$. Applying this to all $j\ne i$'s yields
\begin{equation*}
V_i^{\pi}(\mH_1;\bm M^0(p)) \geq \mu_i T - \frac{K-1}{p},\quad \forall i\in [K].
\end{equation*}

Given the lower bounds on V-functions, we are now ready to control the regret as
\begin{align*}
\mR_T(\pi,\bm M^0(p)) &= T\E_{\bm u\sim \mV}\sqb{\max_{i\in[K]}u_i} - \sum_{i\in[K]}V_i^{\pi}(\mH_1;\bm M^0(p)) \\ &\le T\sum_{i\in[K]}\mu_i -\sum_{i\in[K]}\left (\mu_i T - \frac{K-1}{p}\right ) = \frac{K(K-1)}{p},
\end{align*}
where we used $\E_{\bm u\sim \mV}[\max_i u_i]=\sum_{i\in [K]} \E_{\bm u\sim \mV}[u_i\1[u_i>u_j,\forall j\ne i]]=\sum_{i\in [K]} \mu_i$.
\end{proof}

\section{General Framework for the Auxiliary Game Approach (\Cref{thm:equivalence of PBE})}\label{sec:appendix varying-p}
In this section, we build up the general framework of our auxiliary game approach, which characterizes the performance of our \mechname mechanism under the perfect audit model (\Cref{assump:noiseless}), under the adversarial audit model (\Cref{assump:adv noise model}), and under the stochastic audit model (\Cref{assump:stoc noise model}).

The key result we establish in this section is (a generalized version of) \Cref{thm:equivalence of PBE}, which claims a PBE of the simple auxiliary game correspond to an (approximate) equilibrium in the complex original game induced by \mechname.
Specifically, we introduce abstract parameters $(R,\Delta)$ to capture the behavior of \mechname under various audit models, which are formally defined in \Cref{def:auxiliary game formal} and \Cref{thm:equivalence of PBE formal}. We then prove the generalized version of \Cref{thm:equivalence of PBE} (namely \Cref{thm:equivalence of PBE formal}). In \Cref{sec:appendix regret and audits}, we derive the regret and audits guarantees of \mechname based on \Cref{thm:equivalence of PBE formal}.
In \Cref{sec:appendix deterministic main theorem}, we prove that under all three audit models that we consider, there are appropriate $(R,\Delta)$'s leading to our \Cref{thm:deterministic main thm,thm:adv noise model,thm:stoc noise model}.

\subsection{Additional Notations}\label{sec:appendix additional notations varying-p}
For convenience, we define a few additional notations for agent strategies:
\begin{itemize}
\item For any agent $i\in [K]$, in analog to agent $i$'s strategy throughout the game $\pi_i=(\pi_{t,i})_{t\in [T]}$, the notation $\pi_{-i}=(\pi_j)_{j\ne i}=(\pi_{t,j})_{t\in [T],j\ne i}$ denotes the joint strategy of all agents other than $i$.
\item For any agent $i\in [K]$, given their strategy $\pi_i^1=(\pi_{t,i}^1)_{t\in [T]}$ and their opponents' joint strategy $\pi_{-i}^2=(\pi_{t,j}^2)_{t\in [T],j\ne i}$, $\pi_i^1 \circ \pi_{-i}^2$ denotes the corresponding joint strategy obtained by concatenation, i.e., for every round $t\in [T]$, agent $i$ follows $\pi_{t,i}^1$ while any other agent $j\ne i$ follows $\pi_{t,j}^2$.
\item Given a round $t\in [T]$ together with two joint strategies $\pi^1$ and $\pi^2$, we use $\pi^1\diamond_t \pi^2$ to denote the joint strategy of following $\pi^1$ up to round $t$ and following $\pi^2$ afterwards, i.e., $\pi^1\diamond_t \pi^2=(\pi_{\tau,i})_{\tau\in [T],i\in [K]}$ where for all $i\in [K]$, $\pi_{\tau,i}=\pi_{\tau,i}^1$ for $\tau\leq t$ and $\pi_{\tau,i}=\pi_{\tau,i}^2$ for $\tau>t$.
\end{itemize}

For any strategy $\pi=(\pi^r,\pi^f)$ under \mechname, we can also rewrite the V-function defined in \Cref{eq:V-func general} through the Bellman form: For any eliminated agent $i\not \in \mA_t$, $V_i^{\pi}(\mH_t)=0$, $\forall \mH_t\in H_t$; for any alive agent $i\in \mA_t$, their V-function under joint strategy $\pi$ starting from history $\mH_t\in H_t$ is
\begin{align}
V_i^{\pi}(\mH_t)=\E\left [u_{t,i}\1[i_t=i]+V_i^{\pi}(\mH_{t+1})\middle \vert \mH_t\right ],\label{eq:V-func recursive}
\end{align}
where the randomness lies in the generation of agents' utilities $\bm u_t\sim \mV$, reports $\bm v_t\sim \pi_t^r$, central planer's decision $(i_t,o_t)\sim M_t$, audit feedback $w_t$, agents' flags $\bm f_t\sim \pi_t^f$, and the next-round history $\mH_{t+1}$. Also recall agents' ex-ante first-best utilities $\mu_{t,i}$ and ex-ante winning probabilities $q_{t,i}$ from \Cref{eq:fair share}:
\begin{align}\label{eq:fair share formal}
\mu_{t,i}&:= \E_{\bm u_t\sim \mV}\bigg [u_{t,i}\cdot \1[i\in \mA_t] \1[u_{t,i}\geq c] \1[u_{t,i}>u_{t,j},\forall j\in \mA_t\setminus \{i\}]\bigg ] \\
q_{t,i}&:=\Pr_{\bm u_t\sim \mV}\bigg \{(i\in \mA_t)\wedge (u_{t,i}\geq c) \wedge (u_{t,i}>u_{t,j},\forall j\in \mA_t\setminus \{i\})\bigg \},\quad \forall t\in [T],i\in [K].
\end{align}
While for notational simplicity, we write them as $\mu_{t,i}$ and $q_{t,i}$, we shall remark that they---in addition to the round number $t$ alone---depend on the specific realization of $\mA_t$ during the execution of \mechname.

Finally, to accommodate the various audit models in \Cref{assump:noiseless,assump:adv noise model,assump:stoc noise model}, we need to generalize the auxiliary game. Under the adversarial audit model (\Cref{assump:adv noise model}), it is inevitable that agents may strategize more. In \Cref{def:auxiliary game formal}, we introduce an abstract parameter $R$, the \emph{over-report tolerance}, which allows agents to over-report by $R$ without being eliminated. The main-text version, \Cref{def:auxiliary game}, is a special case when $R=0$.
Moreover, under the stochastic audit model (\Cref{assump:stoc noise model}), truth-telling agents may nevertheless be eliminated, hence the V-functions in both games may differ. This is captured by the \emph{approximation error} $\Delta$ defined in \Cref{thm:equivalence of PBE formal}. The main-text version, as proved in \Cref{lem:equiv between actual and no-flagging}, is a special case when $\Delta=0$.

\subsection{Auxiliary Game and Joint Strategies Therein}\label{sec:appendix auxiliary game}

\begin{definition}[Auxiliary Game with Over-Report Tolerance $R$]\label{def:auxiliary game formal}
Initially, alive agent set $\mA_1=[K]$. In round $t\in [T]$, each alive $i\in \mA_t$ reports a $v_{t,i}\in [0,1]$ that can only depend on their own private utility $u_{t,i}$, round number $t$, and set of alive agents $\mA_t$---but not any other historical information.
Further, if $p_{t,i}:= \frac{1+K^2}{(T-t)q_{t,i} c} \leq 1$, $v_{t,i}\le u_{t,i}+{\color{magenta}R}$ must hold with probability 1, where the fair winning probability $q_{t,i}$ (and also the fair share $\mu_{t,i}$) in the auxiliary game is defined the same as \Cref{eq:fair share formal} using $\mA_t$.
After observing reports $\bm v_t$, the mechanism allocates to the highest-report alive agent $i_t=\argmax_{i\in \mA_t} v_{t,i}$. If $v_{t,i_t}>u_{t,i_t}+{\color{magenta}R}$, agent $i_t$ gets eliminated, i.e., $\mA_{t+1}=\mA_t\setminus \{i_t\}$; otherwise, $\mA_{t+1}=\mA_t$ is unchanged.
\end{definition}
The auxiliary game defined in the main text (namely \Cref{def:auxiliary game}) is therefore a special case where $R=0$.
We specialize our definition of agents' strategies (recall \Cref{def:agent_strat_mechanism}) to the auxiliary game as follows.
\begin{definition}[Agents' Strategies in Generalized Auxiliary Game]\label{def:agents strategy auxiliary}
For any $t\in [T]$, we call $\tilde \mH_t:= (t,\mA_t)$ the simplified history at the beginning of round $t$, and $\tilde H_t:=\{t\}\times 2^{[K]}$ the corresponding simplified history space.
A {report strategy} in the generalized auxiliary game with tolerance $R$ (\Cref{def:auxiliary game formal}) for agent $i\in[K]$ is a sequence of measurable functions $\pi_i^r:=(\pi_{t,i}^r)_{t\in[T]}$ where
\begin{equation*}
\pi_{t,i}^r\colon [0,1]\times \tilde H_t\times \Xi \to \text{Range}_{t,i}(u_{t,i}),~\text{where }\text{Range}_{t,i}(u_{t,i}):=\begin{cases}
[0,u_{t,i}+R],&p_{t,i}\leq 1\\
[0,1],&p_{t,i}> 1
\end{cases},
\end{equation*}
such that their report in round $t\in[T]$ is $v_{t,i}:=\pi_{t,i}^r(u_{t,i},\mH_t^i,\xi^r_{t,i})$, where $\xi^r_{t,i}\sim\mD_\xi$ is sampled independently from other random variables. There is no answers in this auxiliary game, i.e., $F=\{0\}$.
A {joint strategy} for agents is a collection of all agent's report strategies $\pi^r:=(\pi_{i}^r)_{i\in[K]}$.
\end{definition}

In the auxiliary game, the V-function only depends on the simplified history $\tilde \mH_t$ since agents' and the mechanism' actions only depend on $\tilde \mH_t$. Therefore, for a joint strategy ${\pi}^r$, the agents' V-functions under ${\pi}^r$ starting from simplified history $\tilde \mH_t$ can be written as
\begin{equation*}
\tilde V_i^{{\pi}^r}(\tilde \mH_t)=\E\left [\sum_{\tau=t}^T u_{t,i}\1[i_t=i]\middle \vert \tilde \mH_t\right ],
\end{equation*}
where agents utilities $\bm u_t,\bm u_{t+1},\ldots,\bm u_T\overset{\text{i.i.d.}}{\sim}\mV$, reports $\bm v_t\sim {\pi}^r_t,\bm v_{t+1}\sim {\pi}^r_{t+1},\ldots,\bm v_T\sim {\pi}^r_T$, and histories $\tilde \mH_{t+1}=(t+1,\mA_{t+1}),\ldots,\tilde \mH_T=(T,\mA_T)$ are computed according to \Cref{def:auxiliary game formal}.

To summarize, the auxiliary game \Cref{def:auxiliary game formal} differs from the original game in three aspects:
\begin{enumerate}
\item \textbf{Reports only depend on simplified $\tilde \mH_t$.} Agents must decide their reports $v_{t,i}$ only using their own private utilities $u_{t,i}$, the current round number $t$, and the set of alive agents $\mA_t$. \label{item:only simplified history formal}
\item \textbf{No Over-Reporting Unless $p_{t,i}>1$.} When $p_{t,i}=\frac{1+K^2}{(T-t)q_{t,i} c} \leq 1$ (cf. the $\frac{4(1+K^2)}{(T-t)\hat q_{t,i} c}$ term in $\hat p_{t,i}$), agents must have their reports (roughly) upper bounded by utilities, i.e., $v_{t,i}\le u_{t,i}+R$. \label{item:no over-report unless p>=1 formal}
\item \textbf{Eliminate Immediately on Mark Up.} In case an agent $i$ wins by over-reporting, i.e., $i_t=i$ and their report $v_{t,i}>u_{t,i}+R$, they are eliminated for all future rounds. In comparison, in the original game, agent $i$ would be eliminated w.p. $\hat p_{t,i}$ that may be less than $1$ (\Cref{line:randomly_check} in \Cref{alg:mechanism}).\label{item:eliminate immediately formal}
\end{enumerate}

Following the sketch in \Cref{sec:sketch varying-p}, we first prove \Cref{lem:V lower and upper bound} in presence of the tolerance parameter $R$, and then generalize and prove \Cref{thm:equivalence of PBE} to capture the approximation error of auxiliary games.

\subsection{Lower and Upper Bounds on V-Functions (Generalized \Cref{lem:V lower and upper bound})}
\begin{lemma}[Upper and Lower Bounds on V-Functions]\label{lem:V lower and upper bound formal}
Let $\pi^{\ast,r}$ be a PBE in the auxiliary game from \Cref{def:auxiliary game formal}. Fix any auxiliary game history $\tilde \mH_t=(t,\mA_t)$ yielded by following $\pi^{\ast,r}$ in the past. Then:
\begin{align}
(T-t+1) \mu_{t,i}-(K-1) &\leq \tilde V_i^{{\pi}^{\ast,r}}(\tilde \mH_t)\leq (T-t+1) \mu_{t,i}+K(K-1),&&\quad \forall \tilde \mH_t=(t,\mA_t), i\in \mA_t,\label{eq:for alive formal}\\
(T-t+1) \mu_{1,i}-K &\leq \tilde V_i^{{\pi}^{\ast,r}}(\tilde \mH_t)=0,&&\quad \forall \tilde \mH_t=(t,\mA_t), i\notin \mA_t. \label{eq:for dead formal}
\end{align}
\end{lemma}
\begin{proof}

\paragraph{Lower Bounds for Alive Agents.} Fix an alive agent $i\in\mA_t$. We start by proving a lower bound on the utility of agent $i$ in the auxiliary game (\Cref{def:auxiliary game formal}). Consider the reporting strategy $\pi_i^r$ for agent $i$ that always reports $v_{\tau,i}=u_{\tau,i}+R$ for any round $\tau\in[T]$ (i.e., the maximum over-report tolerance allowed by \Cref{def:auxiliary game formal}). Since ${\pi}^{\ast,r}$ is a PBE in the auxiliary game and $\pi_i^r\circ {\pi}^{\ast,r}_{-i}$ is a valid joint strategy therein,
\begin{equation}\label{eq:pbe_for_lower_bound}
\tilde V_i^{{\pi}^{\ast,r}}(\tilde \mH_t)\ge \tilde V_i^{\pi_i^r\circ {\pi}^{\ast,r}_{-i}}(\tilde \mH_{t}).
\end{equation}

We now lower bound the RHS, where agent $i$ always reports $v_{t,i}=u_{t,i}+R$. If in any round $\tau\ge t$ agent $i$ would have won had other agents over-reported by no more than $R$ (i.e., agent $i$ had the largest true utility $u_{\tau,i}=\max_{j\in \mA_{\tau}} u_{\tau,j}$ and $u_{\tau,i}+R\ge c$), either agent $i$ won so that $i_{\tau}=i$, or the winner $i_{\tau}$ over-reported and was eliminated. Since each agent can only be eliminated once and agent $i$ is never eliminated, the latter case occurs for at most $K-1$ rounds. Making it formal, we get
\begin{align*}
&\quad \tilde V_i^{\pi_i^r\circ {\pi}^{\ast,r}_{-i}}(\tilde \mH_{t}) = \E_{\pi_i^r\circ {\pi}^{\ast,r}_{-i}}\left [ \sum_{\tau=t}^T u_{\tau,i} \1[i_{\tau}=i] \middle \vert \tilde \mH_t\right ]\\
&\geq \E_{\pi_i^r\circ {\pi}^{\ast,r}_{-i}} \sqb{ \sum_{\tau=t}^T u_{\tau,i} \paren{\1\sqb{u_{\tau,i}=\max_{j\in \mA_{\tau}}u_{\tau,j}} \1[u_{\tau,i}+R\ge c]  - \1[i_t\neq i] \1\sqb{v_{t,i_t}>u_{t,i_t}+R}} \1\sqb{i\in \mA_{\tau}}}\\
&\overset{(a)}{\geq} \E_{\pi_i^r\circ {\pi}^{\ast,r}_{-i}} \sqb{ \sum_{\tau=t}^T u_{\tau,i} \1\sqb{u_{\tau,i}=\max_{j\in \mA_{\tau}}u_{\tau,j}} \1[u_{\tau,i}\ge c]} - \E_{\pi_i^r\circ \pi_{-i}^{\ast,r}} \left [ \sum_{\tau=t}^T u_{\tau,i} \1[i_t\neq i] \1\sqb{v_{t,i_t}>u_{t,i_t}+R}\right ]\\
&\overset{(b)}{=}(T-t+1)\mu_{t,i} - \E_{\pi_i^r\circ {\pi}^{\ast,r}_{-i}} \sqb{ \sum_{t=1}^T |\mA_t|-|\mA_{t+1}|} \geq (T-t+1)\mu_{t,i} - (K-1).
\end{align*}
where (a) is because agent $i$ is always alive under $\pi_i^r$ (hence $\1[i\in \mA_\tau]=1$ for all $\tau\ge t$) and $\1[u_{\tau,i}+R\ge c]\ge \1[u_{\tau,i}\ge c]$, while (b) uses the definition of $\mu_{\tau,i}$ in \Cref{eq:fair share formal}, the fact that $\mA_{\tau}\subseteq \mA_t$ and $i\in \mA_{\tau}$ imply $\mu_{\tau,i}\ge \mu_{t,i}$, and the fact that $v_{t,i_t}>u_{t,i_t}+R$ implies $\mA_{t+1}=\mA_t\setminus \{i_t\}$ according to \Cref{def:auxiliary game formal}.

\paragraph{Lower Bound for Eliminated Agents.}
We make the following claim: For any round $\tau\in[T]$, simplified history $\tilde\mH_\tau$, and alive agent $i\in\mA_\tau$ such that $(T-\tau)\mu_{\tau,i} > K$, agent $i$ remains alive a.s. under the PBE strategy ${\pi}^{\ast,r}$.
To prove it, consider an alternative strategy $\pi_{\tau,i}^{\prime,r}$ that caps the report $v_{\tau,i}$ suggested by $\pi_{\tau,i}^{\ast,r}$ as $v_{\tau,i}':=\min(v_{\tau,i},u_{\tau,i}+R)$.
As ${\pi}^{\ast,r}$ is a PBE, agent $i$'s unilateral deviation from $\pi_{\tau,i}^{\ast,r}$ to $\pi_{\tau,i}^{\prime,r}$ is unprofitable.

The only case where $\pi_{\tau,i}^{\ast,r}$ and $\pi_{\tau,i}^{\prime,r}$ yield different outcomes is when \emph{i)} agent $i$ over-reports under $\pi_{s,i}^{\ast,r}$, i.e., $v_{\tau,i}>u_{\tau,i}+R$, and \emph{ii)} agent $i$ wins in round $\tau$ by over-reporting $v_{s,i}$. Denote by $\mE_{\tau,i}$ this event. Under $\mE_{\tau,i}$, if agent $i$ followed $\pi_{\tau,i}^{\ast,r}$ by reporting $v_{\tau,i}$, their total utility is $u_{\tau,i}\leq 1$ since they are automatically eliminated (\Cref{item:eliminate immediately formal}). On the other hand, if agent $i$ followed $\pi_{\tau,i}^{\prime,r}$, they obtain
\begin{equation*}
u_{\tau,i}\1[i_\tau=i]+\tilde V_i^{{\pi}^{\ast,r}}(\tilde \mH_{\tau+1})\geq 0+(T-\tau)\mu_{\tau+1,i}-(K-1)\ge (T-\tau)\mu_{\tau,i}-(K-1),
\end{equation*}
where the first inequality uses the just-proved lower bound in \Cref{eq:for alive formal}, and the second inequality is because $\mA_{\tau+1}\subseteq \mA_\tau$ and hence $\mu_{\tau+1,i}\ge \mu_{\tau,i}$. Therefore, if $(T-\tau)\mu_{\tau,i}>K$, to ensure the unilateral deviation from $\pi_{\tau,i}^{\ast,r}$ to $\pi_{\tau,i}^{\prime,r}$ is unprofitable, we must have $\Pr\{\mE_{\tau,i}\}=0$. Thus $i\in \mA_{\tau+1}$ a.s.
This means that for an eliminated agent in round $t$, they must be eliminated in some previous round $\tau<t$, and consequently $(T-\tau)\mu_{\tau,i}\le K$. Since $\mu_{1,i}\le \mu_{\tau,i}$ when $i\in \mA_\tau$, we arrive at the conclusion that $(T-t+1)\mu_{1,i}\le (T-\tau)\mu_{\tau,i}\le K$.

\paragraph{Upper Bounds for Alive Agents.}
Starting from $\tilde \mH_t$, only those agents in $\mA_t$ may receive allocations during rounds $\tau=t,t+1,\ldots,T$. Hence, the sum of V-functions of all agents can never exceed the first-best mechanism which allocates the resource to the alive agent with the highest utility, namely
\begin{align*}
\sum_{i=1}^K \tilde V_i^{{\pi}^{\ast,r}}(\tilde \mH_{t}) &\le \sum_{i\in\mA_t} (T-t+1)\E_{\bm u_t\sim \mV} \bigg[u_{t,i}\1[i\in \mA_t] \1[u_{t,i}>u_{t,j},\forall j\in \mA_t\setminus \{i\}] \bigg]\\
&\le (T-t+1)\left (\sum_{i\in\mA_t} \mu_{t,i} + c\Pr_{\bm u\sim \mV}\{u_i<c,\forall i\in\mA_t\}\right ),
\end{align*}
where the second inequality is because the definition of $\mu_{t,i}$ in \Cref{eq:fair share formal} contains an indicator $\1[u_{t,i}\ge c]$.
We recall that by \Cref{assumption:min report}, there exists some (albeit unknown to the planner) agent $i_0\in [K]$ that ensures $\Pr_{u_{i_0}\sim\mV_{i_0}} \{u_{i_0}\geq c\}=1$. Thus, if $i_0\in \mA_t$ is still alive, we already proved the desired upper bound since $\Pr_{\bm u\sim \mV}\{u_i<c,\forall i\in\mA_t\}\le \Pr_{u_{i_0}\sim \mV_{i_0}}\{u_{i_0}<c\}=0$. We now assume $i_0\notin\mA_t$, in which case
\begin{equation*}
\sum_{i\notin\mA_t}\mu_{1,i}=\E_{\bm u\sim \mV}\sqb{\max_{i\notin \mA_t} u_i \1\sqb{\max_{i\in \mA_t} u_i<\max_{i\not\in \mA_t} u_i}}\ge\E_{\bm u\sim \mV} \left [u_{i_0} \1\left [\max_{i\in \mA_t}u_i<u_{i_0}\right ]\right ]\ge c\Pr_{\bm u\sim \mV}\{u_i<c,\forall i\in\mA_t\},
\end{equation*}
where the first inequality uses the assumption that $i_0\notin \mA_t$ and the second inequality uses $u_{i_0}\ge c$ a.s. from \Cref{assumption:min report}.
Combining it with the second inequality of \Cref{eq:for dead formal}, which says $(T-t+1)\mu_{1,i}-K\leq 0$ for all $i\notin \mA_t$, we have
\begin{equation}\label{eq:upper_bound_max_welfare}
\sum_{i=1}^K \tilde V_i^{{\pi}^{\ast,r}}(\tilde \mH_{t}) \leq (T-t+1)\left (\sum_{i\in\mA_t} \mu_{t,i}+\sum_{i\notin\mA_t} \mu_{1,i}\right ) \leq (T-t+1)\sum_{i\in\mA_t} \mu_{t,i} + K\sum_{i=1}^K \1[i\notin \mA_t].
\end{equation}
Therefore, utilizing the lower bound for those alive agents in \Cref{eq:for alive formal}, we know
\begin{equation*}
\sum_{j\in\mA_t\setminus\{i\}} \tilde V_{j}^{{\pi}^\ast}(\mH_{t})\ge \sum_{j\in\mA_t\setminus\{i\}} \left ((T-t+1)\mu_{t,j}-K\right ),\quad \forall i\in \mA_t.
\end{equation*}
Subtracting the RHS from \Cref{eq:upper_bound_max_welfare}, we arrive at
\begin{align*}
\tilde V_i^{{\pi}^{\ast,r}}(\tilde \mH_{t})&\le (T-t+1)\mu_{t,i}+K \sum_{j\ne i} \big (\1[j\notin \mA_t]+\1[j\in \mA_t\setminus \{i\}]\big )\\&\le (T-t+1)\mu_{t,i}+K(K-1),\quad \forall i\in \mA_t.
\end{align*}
This proves the upper bound part of \Cref{eq:for alive formal}.
\end{proof}

\subsection{Correspondence of Equilibrium (Generalized \Cref{thm:equivalence of PBE} and \Cref{lem:min report with u is good})}\label{sec:equivalence of PBE formal}
\begin{theorem}[Correspondence of Equilibrium]\label{thm:equivalence of PBE formal}
Suppose that \Cref{assumption:min report} holds.
Consider the generalized auxiliary game with tolerance $R$ (\Cref{def:auxiliary game formal}). Let $\pi^{\ast,r}$ be a PBE therein, and let $\pi^{\ast,f}$ be the well-behaved reporting strategy \Cref{def:well-behaved flagging}.
For the original-game strategy $\pi^\ast=(\pi^{\ast,r},\pi^{\ast,f})$, let
\begin{equation}\label{eq:definition of Delta}
\Delta(\pi^{\ast,r}):=\max_{i\in [K]} \left \vert V_i^{\pi^\ast} (\mH_1)-\tilde V_i^{\pi^{\ast,r}}(\tilde \mH_1)\right \rvert
\end{equation}
be the maximum approximation error between the auxiliary game and the original game. Then, in the original game induced by \mechname, $\pi^\ast$ is a $\Delta(\pi^{\ast,r})$-Approximate Bayesian Equilibrium (see \Cref{def:approx PBE}).
\end{theorem}
\begin{proof}
Fix an agent $i\in[K]$ and any alternative strategy ${\pi}_i=({\pi}_i^r,{\pi}_i^f)$. We show the unilateral deviation does not increase their original-game utility by more than $\Delta(\pi^{\ast,r})$, i.e., $V_i^{\pi_i\circ \pi_{-i}^\ast}(\mH_1)\le V_i^{\pi^\ast}(\mH_1)+\Delta(\pi^{\ast,r})$.

Knowing that all other agents' are using $\pi_{-i}^\ast$ and the planner is committing to \mechname, agent $i$ essentially face a $T$-round Markov Decision Process (which we call the ``agent-$i$ MDP''), where the round-$t$ state is the observable history $\mH_t^i\in H_t^i$. Given state $\mH_t^i$, agent $i$ decides their actions---mappings from newly observable information to reports and flags---namely $\pi_{t,i}^r(\mH_t^i)\colon u_{t,i}\mapsto v_{t,i}$ and $\pi_{t,i}^f(\mH_t^i)\colon (u_{t,i},\bm v_t,i_t,o_t,w_t)\mapsto f_{t,i}$. For notational simplicity, we denote the actions of agent $i$ on state $\mH_t^i$ by $\pi_{t,i}(\mH_t^i)=\big (\pi_{t,i}^r(\mH_t^i),\pi_{t,i}^f(\mH_t^i)\big )$.

Similar to the agent-$i$ V-function defined in \Cref{eq:V-func recursive}, we define the agent-$i$ Q-function (the abbreviation for state-action-value function) for any state $\mH_t$ and action $\pi_{t,i}(\mH_t^i)$ as:
\begin{equation}\label{eq:Q func recursive}
Q_i^{\pi^\ast}(\mH_t,\pi_{t,i}(\mH_t^i)):=\E_{\pi_{t,i}\circ \pi_{t,-i}^\ast} \Big [u_{t,i}\1[i_t=i]+V_i^{\pi^\ast}(\mH_{t+1})~\Big \vert~ \mH_t\Big ],\quad \forall \mH_t,~\forall \pi_{t,i}(\mH_t^i).
\end{equation}
By definition, we know $V_i^{\pi^\ast}(\mH_t)=Q_i^{\pi^\ast}(\mH_t,\pi_{t,i}^\ast(\mH_t^i))$.
Applying the performance difference lemma \citep{kakade2002approximately} to the ``agent-$i$ MDP'' w.r.t. two agent-$i$ policies $\pi_i$ and $\pi_i^\ast$, we have:
\begin{equation}
V_i^{\pi_i\circ \pi_{-i}^\ast}(\mH_1)-V_i^{\pi^\ast}(\mH_1)=\sum_{t=1}^T \E_{\mH_t\sim \pi_i\circ \pi_{-i}^\ast}\left [Q_i^{\pi^\ast}(\mH_t,\pi_{t,i}(\mH_t^i))-V_i^{\pi^\ast}(\mH_t)\right ].\label{eq:performance difference lemma}
\end{equation}

Fixing $\mH_t\sim \pi_i\circ \pi_{-i}^\ast$. Since all other agents were using $\pi_{<t,-i}^{\ast,f}$ in the past, we immediately have $\hat q_{t,j}\le 4q_{t,j}$ for all $j\in [K]$. 
We now construct another strategy $\pi_i'(\mH_t^i)$ that \emph{i)} well-approximates $\pi_i$ in the original game; and \emph{ii)} is a valid auxiliary-game strategy and thus inferior than $\pi_i^\ast$ therein:
\begin{enumerate}
\item \textbf{Eliminating the Deviation in Flagging Strategies.} Regardless what $\pi_i^f$ is, we always set $\pi_i^{\prime,f}$ as the well-behave flagging strategy $\pi_i^{\ast,f}$ in \Cref{def:well-behaved flagging}. We argue that $\mH_{t+1}$ always ensures $\hat q_{t+1,j}\le 4q_{t+1,j}$ for all $j\in [K]$ under either $\pi_{t,i}^f$ or ${\pi}_{t,i}^{\ast,f}$, which already suffices to impose enough threats: If $i_t$ gets eliminated, then $\hat q_{t+1,j}=0$ for all $j$; otherwise, for any $j\ne i_t$, we have $\hat q_{t+1,j}=\hat q_{t,j}\le 4q_{t,j}=4q_{t+1,j}$, and it only remains to consider the case where $\hat q_{t+1,i_t}>4q_{t+1,i_t}$---but due to ${\pi}_{t,-i}^{\ast,f}$, such $\hat q_{t+1,i_t}$ would be reset to $0$.
\item \textbf{Eliminating the Deviation to Over-Reporting.} For any private utility $u_{t,i}$, we define a coupling between $v_{t,i}\sim \pi_{t,i}^r(u_{t,i};\mH_t^i)$ and another random variable $v_{t,i}'$: If $p_{t,i}\le 1$, then $v_{t,i}'=\min(v_{t,i},u_{t,i}+R)$; otherwise $v_{t,i}'=v_{t,i}$. Let $\pi_i^{\prime,r}(u_{t,i};\mH_t^i)$ be the marginal of $v_{t,i}'$, then it complies with \Cref{item:no over-report unless p>=1 formal} of \Cref{def:auxiliary game formal}. As we informally argued in \Cref{lem:min report with u is good} in the main text, due to our carefully chosen audit probability $\hat p_{t,i}$, this strategy is always better off. This is formalized shortly when controlling the first term in \Cref{eq:Q diff decomposition}.\label{item:coupling of over-reports}
\item \textbf{Eliminating the History Dependency.} It only remains to eliminate the dependency of $\pi_i'$ on the full history $\mH_t^i$, which is not allowed according to \Cref{item:only simplified history formal} of \Cref{def:auxiliary game formal}. While $\pi_i'$ may not be directly reducible to a valid auxiliary-game strategy (the detailed argument is presented when controlling the ``Second Term in \Cref{eq:Q diff decomposition}''), we manage to bypass it using another performance difference lemma.
\end{enumerate}
After constructing the $\pi_i'$, we fix a history $\mH_t$ sampled from $\pi_i\circ \pi_{-i}^\ast$ and study the expectation in \Cref{eq:performance difference lemma}:
\begin{align}
&\quad Q_i^{\pi^\ast}(\mH_t,\pi_{t,i}(\mH_t^i))-V_i^{\pi^\ast}(\mH_t)\\
&=\left [Q_i^{\pi^\ast}(\mH_t,\pi_{t,i}(\mH_t^i))-Q_i^{\pi^\ast}(\mH_t,\pi_{t,i}'(\mH_t^i))\right ]+\left [Q_i^{\pi^\ast}(\mH_t,\pi_{t,i}'(\mH_t^i))-Q_i^{\pi^\ast}(\mH_t,\pi_{t,i}^\ast(\mH_t^i))\right ]. \label{eq:Q diff decomposition}
\end{align}

\paragraph{First Term in \Cref{eq:Q diff decomposition}.}
For notational simplicity, from now on, all expectations are taken conditional on the history $\mH_t$, fixing other agents' strategies as $\pi_{-i}^\ast$, and fixing the mechanism as \mechname.
For the first term in \Cref{eq:Q diff decomposition}, we generalize the arguments of \Cref{lem:min report with u is good}: By definition of Q-function in \Cref{eq:Q func recursive}, we compare immediate rewards and future gains separately, i.e.,
\begin{align*}
&\quad Q_i^{\pi^\ast}(\mH_t,\pi_{t,i}(\mH_t^i))-Q_i^{\pi^\ast}(\mH_t,\pi_{t,i}'(\mH_t^i))\\
&\overset{(a)}{=}\E_{\pi_{t,i}\circ \pi_{t,-i}^{\ast}} \left [u_{t,i}\1[i_i=i]+V_i^{\pi^\ast}(\mH_{t+1})\middle \vert \mH_t\right ]-\E_{\pi_{t,i}'\circ \pi_{t,-i}^{\ast}} \left [u_{t,i}\1[i_i=i]+V_i^{\pi^\ast}(\mH_{t+1})\middle \vert \mH_t\right ]\\
&\overset{(b)}{\le} 1+\E_{\pi_{t,i}\circ \pi_{t,-i}^{\ast}}\left [\tilde V_i^{\pi^{\ast,r}}(\tilde \mH_{t+1})\middle \vert \mH_t\right ]+\E_{\pi_{t,i}\circ \pi_{t,-i}^{\ast}}\left [V_i^{\pi^\ast}(\mH_{t+1})-\tilde V_i^{\pi^\ast}(\tilde \mH_{t+1})\middle \vert \mH_t\right ]\\
&\qquad -\E_{\pi_{t,i}'\circ \pi_{t,-i}^{\ast}}\left [\tilde V_i^{\pi^{\ast,r}}(\tilde \mH_{t+1})\middle \vert \mH_t\right ]-\E_{\pi_{t,i}'\circ \pi_{t,-i}^{\ast}}\left [V_i^{\pi^\ast}(\mH_{t+1})-\tilde V_i^{\pi^\ast}(\tilde \mH_{t+1})\middle \vert \mH_t\right ]\\
&\overset{(c)}{\le} \E_{\pi_{t,i}\circ \pi_{t,-i}^{\ast}}\left [V_i^{\pi^\ast}(\mH_{t+1})-\tilde V_i^{\pi^\ast}(\tilde \mH_{t+1})\middle \vert \mH_t\right ]-\E_{\pi_{t,i}'\circ \pi_{t,-i}^{\ast}}\left [V_i^{\pi^\ast}(\mH_{t+1})-\tilde V_i^{\pi^\ast}(\tilde \mH_{t+1})\middle \vert \mH_t\right ],
\end{align*}
where we remark that all $\tilde \mH_{t+1}$ here means the simplified history corresponding to $\mH_{t+1}$ (which is generated by \mechname from $\mH_t$, instead of the one generated by the auxiliary game from $\tilde \mH_t$). Here, (a) used \Cref{eq:Q func recursive}, (b) used the fact that $u_{t,i}\in [0,1]$, and we now prove (c), the following inequality:
\begin{equation}\label{lem:min report with u is good formal}
1+\E_{\pi_{t,i}\circ \pi_{t,-i}^\ast}\left [\tilde V_i^{\pi^{\ast,r}}(\tilde \mH_{t+1})\middle \vert \mH_t\right ]-\E_{\pi_{t,i}'\circ \pi_{t,-i}^\ast}\left [\tilde V_i^{\pi^{\ast,r}}(\tilde \mH_{t+1})\middle \vert \mH_t\right ]\le 0.
\end{equation}
It is essentially the generalized version of \Cref{lem:min report with u is good}, which claims that over-reporting is always unprofitable.

To prove \Cref{lem:min report with u is good formal}, we fix the realization of the true utility of agent $i$ as $u_{t,i}$.
Consider the coupling between $v_{t,i}\sim \pi_{t,i}^{r}(u_{t,i};\mH_t^i)$ and $v_{t,i}'\sim \pi_{t,i}^{\prime,r}(u_{t,i};\mH_t^i)$ defined in Item~\ref{item:coupling of over-reports} above. When $v_{t,i}=v_{t,i}'$, they trivially induce the same distribution of next-round simplified history $\tilde \mH_{t+1}$ (depending on the realization of other agents' reports). Otherwise, we have $v_{t,i}\ne v_{t,i}'$, which, by definition of the coupling, ensures $p_{t,i}=\frac{1+K^2}{(T-t)q_{t,i}c}\le 1$ and $v_{t,i}>u_{t,i}+R$. Furthermore, if neither $v_{t,i}$ nor $v_{t,i}'$ wins (given other agents' realizations of reports $\bm u_{t,-i}$), we also trivially induce the same $\tilde \mH_{t+1}$. Noticing that $v_{t,i}>u_{t,i}+R=v_{t,i}'$ in this case, we arrive at:
\begin{align*}
&\quad 1+\E_{\pi_{t,i}\circ \pi_{t,-i}^\ast}\left [\tilde V_i^{\pi^{\ast,r}}(\tilde \mH_{t+1})\middle \vert \mH_t\right ]-\E_{\pi_{t,i}'\circ \pi_{t,-i}^\ast}\left [\tilde V_i^{\pi^{\ast,r}}(\tilde \mH_{t+1})\middle \vert \mH_t\right ]\\
&=1+\E_{v_{t,i},v_{t,i}',\bm v_{t,-i}}\Bigg [\1[p_{t,i}\le 1] \1[v_{t,i}>u_{t,i}+R] \1[v_{t,i}>v_{t,j},\forall j\in \mA_t\setminus \{i\}] \\
&\quad \qquad \left (\E[\tilde V_i^{\pi^{\ast,r}}(\tilde \mH_{t+1})\mid \mH_t,v_{t,i},\bm v_{t,-i}]-\E[\tilde V_i^{\pi^{\ast,r}}(\tilde \mH_{t+1})\mid \mH_t,v_{t,i}',\bm v_{t,-i}]\right )\Bigg ],
\end{align*}
where we again remark that the $\tilde \mH_{t+1}$ here means the simplified history corresponding to $\mH_{t+1}$ (induced by \mechname), instead of the one induced by the auxiliary game starting from $\tilde \mH_t$.
We now compare the $\tilde \mH_{t+1}$'s induced by $v_{t,i}$ and $v_{t,i}'$, under the conditions \textit{i)} $p_{t,i}\le 1$, \textit{ii)} $v_{t,i}>u_{t,i}+R$, and \textit{iii)} $v_{t,i}>v_{t,j},\forall j\in \mA_t\setminus \{i\}$:
\begin{itemize}
\item When reporting $v_{t,i}$, agent $i$ is the winner because of \textit{iii)}. With probability $\hat p_{t,i}$, \mechname audit agent $i$, who is then eliminated because of \textit{ii)}. Thus in this case, we arrive at a simplified history $(t+1,\mA_{t}\setminus \{i\})$ such that $\tilde V_i^{\pi^{\ast,r}}(t+1,\mA_{t}\setminus \{i\})=0$. For the remaining probability of $1-\hat p_{t,i}$, no audit or elimination happen and we arrive at a simplified history $(t+1,\mA_{t})$. In this case, $\mu_{t+1,i}=\mu_{t,i}$; from \Cref{lem:V lower and upper bound formal} (generalized \Cref{lem:V lower and upper bound}), $\tilde V_i^{\pi^{\ast,r}}(t+1,\mA_t)\le (T-t)\mu_{t,i}+K(K-1)$. Hence we have
\begin{equation*}
\E\left [\tilde V_i^{\pi^{\ast,r}}(\tilde \mH_{t+1})\middle \vert \mH_t,v_{t,i},\bm v_{t,-i}\right ]\le (1-\hat p_{t,i}) \cdot \big ((T-t)\mu_{t,i}+K(K-1)\big ).
\end{equation*}
We further upper bound the coefficient $(1-\hat p_{t,i})$. Since $\mH_t$ ensures $\hat q_{t,i}\le 4q_{t,i}$, we have
\begin{equation*}
\hat p_{t,i}=\min\left (\frac{4(1+K^2)}{(T-t)\hat q_{t,i} c},1\right )\ge \min\left (\frac{1+K^2}{(T-t)q_{t,i} c},1\right )=\min(p_{t,i},1)=p_{t,i},
\end{equation*}
where the last step uses \textit{i)}. Further plugging in the definition that $p_{t,i}=\frac{1+K^2}{(T-t)q_{t,i} c}$ gives
\begin{equation*}
\E\left [\tilde V_i^{\pi^{\ast,r}}(\tilde \mH_{t+1})\middle \vert \mH_t,v_{t,i},\bm v_{t,-i}\right ]\le \left (1-\frac{1+K^2}{(T-t)q_{t,i} c}\right )\big ((T-t)\mu_{t,i}+K(K-1)\big ).
\end{equation*}
\item Otherwise, when near-truthfully reporting $v_{t,i}'=u_{t,i}+R$, by definition of the tolerance parameter $R$ in \Cref{def:auxiliary game formal}, agent $i$ is never eliminated by \mechname and thus the new simplified history $\tilde \mH_{t+1}=(t+1,\mA_{t+1})$ ensures $i\in \mA_{t+1}$. Consequently, $\mu_{t+1,i}\ge \mu_{t,i}$. Hence \Cref{lem:V lower and upper bound formal} gives
\begin{equation*}
\E\left [\tilde V_i^{\pi^{\ast,r}}(\tilde \mH_{t+1})\middle \vert \mH_t,v_{t,i}',\bm v_{t,-i}\right ]\ge (T-t)\mu_{t,i}-K.
\end{equation*}
\end{itemize}

Therefore, putting the two parts together, under \textit{i)}, \textit{ii)}, and \textit{iii)}, we have
\begin{align*}
&\quad \E\left [\tilde V_i^{\pi^{\ast,r}}(\tilde \mH_{t+1})\middle \vert \mH_t,v_{t,i},\bm v_{t,-i}\right ]-\E\left [\tilde V_i^{\pi^{\ast,r}}(\tilde \mH_{t+1})\middle \vert \mH_t,v_{t,i}',\bm v_{t,-i}\right ]\\
&\le \left (1-\frac{1+K^2}{(T-t)q_{t,i} c}\right )\big ((T-t)\mu_{t,i}+K(K-1)\big ) - \big ((T-t)\mu_{t,i}-K\big ) \\
&\le K^2 - (1+K^2)\frac{(T-t)\mu_{t,i}}{(T-t)q_{t,i} c} \le -1,
\end{align*}
where the last step uses \Cref{assumption:min report} which gives $\mu_{t,i}\ge q_{t,i} c$. This proves \Cref{lem:min report with u is good formal}, and hence
\begin{align}\label{eq:first term Q diff}
&\quad Q_i^{\pi^\ast}(\mH_t,\pi_{t,i}(\mH_t^i))-Q_i^{\pi^\ast}(\mH_t,\pi_{t,i}'(\mH_t^i))\\
&\le \E_{\pi_{t,i}\circ \pi_{t,-i}^{\ast}}\left [V_i^{\pi^\ast}(\mH_{t+1})-\tilde V_i^{\pi^\ast}(\tilde \mH_{t+1})\middle \vert \mH_t\right ]-\E_{\pi_{t,i}'\circ \pi_{t,-i}^{\ast}}\left [V_i^{\pi^\ast}(\mH_{t+1})-\tilde V_i^{\pi^\ast}(\tilde \mH_{t+1})\middle \vert \mH_t\right ].
\end{align}

\paragraph{Second Term in \Cref{eq:Q diff decomposition}.}
To compare $\pi_{t,i}'$ to $\pi_{t,i}^\ast$, we utilize the condition that $\pi_{t,i}^{\ast,r}$ is a PBE in the auxiliary game, hence any unilateral deviation therein is unprofitable. However, there is a technical subtlety: $\pi_{t,i}^{\prime,r}$ is \emph{not} a valid auxiliary-game strategy, because different $\mH_t$'s that correspond to the same $\tilde \mH_t$ may be assigned for different reports (because $\pi_{t,i}^r$ does; recall the construction of $\pi_{t,i}^{\prime,r}$ from Item~\ref{item:coupling of over-reports} above).

To bypass this, we view the auxiliary game as another MDP---which we call ``agent-$i$ auxiliary MDP''---by fixing all other agents' strategies as $\pi^{\ast,r}_{-i}$ and fixing the planner's mechanism as in the auxiliary game (\Cref{def:auxiliary game formal}). Thus, the state in this MDP is the simplified history $\tilde \mH_t\in H_t$ and the action is any $\pi_{t,i}^r\colon u_{t,i}\mapsto v_{t,i}$ that is valid on $\tilde \mH_t$ (recall \Cref{def:agents strategy auxiliary}). Similar to \Cref{eq:Q func recursive}, we define the Q-function for this MDP as $\tilde Q_i^{\pi^{\ast,r}}(\tilde H_t,\pi_{t,i}^r):=\E_{\pi_{t,i}^r\circ \pi_{t,-i}^{\ast,r}}[u_{t,i}\1[i_t=i]+\tilde V_i^{\pi^{\ast,r}}(\tilde \mH_{t+1})\mid \tilde \mH_t]$. Since $\pi^{\ast,r}$ is an auxiliary-game PBE, we have $\tilde Q_i^{\pi^{\ast,r}}(\tilde \mH_t,\pi_{t,i}^r)\le \tilde V_i^{\pi^{\ast,r}}(\tilde \mH_t)$ for any valid $\pi_{t,i}^r$ and any $\tilde \mH_t\in \tilde H_t$.
Therefore, we have
\begin{equation*}
\begin{aligned}
&\quad Q_i^{\pi^\ast}(\mH_t,\pi_{t,i}'(\mH_t^i))
=\tilde Q_i^{\pi^{\ast,r}}(\tilde \mH_t,\pi_{t,i}^{\prime,r}(\mH_t^i))+\left (Q_i^{\pi^\ast}(\mH_t,\pi_{t,i}'(\mH_t^i))-\tilde Q_i^{\pi^{\ast,r}}(\tilde \mH_t,\pi_{t,i}^{\prime,r}(\tilde \mH_t))\right )\\
&\le \tilde V_i^{\pi^{\ast,r}}(\tilde \mH_t)+\left (Q_i^{\pi^\ast}(\mH_t,\pi_{t,i}'(\mH_t^i))-\tilde Q_i^{\pi^{\ast,r}}(\tilde \mH_t,\pi_{t,i}'(\mH_t^i))\right )\\
&=V_i^{\pi^\ast}(\mH_t)+\left (\tilde V_i^{\pi^{\ast,r}}(\tilde \mH_t)-V_i^{\pi^\ast}(\mH_t)\right )+\left (Q_i^{\pi^\ast}(\mH_t,\pi_{t,i}'(\mH_t^i))-\tilde Q_i^{\pi^{\ast,r}}(\tilde \mH_t,\pi_{t,i}'(\mH_t^i))\right ),
\end{aligned}
\end{equation*}
where the inequality is because $\pi_{t,i}^{\prime,r}(\mH_t^i)\colon u_{t,i}\mapsto v_{t,i}$ is valid on state $\tilde \mH_t$ since it complies with \Cref{item:no over-report unless p>=1 formal} of \Cref{def:auxiliary game formal}. Furthermore, for the last term on the RHS, since the allocation rule is the same between \mechname (\Cref{alg:mechanism}) and the auxiliary-game mechanism (\Cref{def:auxiliary game formal}), we have
\begin{equation*}
Q_i^{\pi^\ast}(\mH_t,\pi_{t,i}'(\mH_t^i))-\tilde Q_i^{\pi^{\ast,r}}(\tilde \mH_t,\pi_{t,i}'(\mH_t^i))=\E_{\pi_{t,i}'\circ \pi_{t,-i}^\ast} \left [V_i^{\pi^\ast}(\mH_{t+1})-\tilde V_i^{\pi^\ast}(\tilde \mH_{t+1})\middle \vert \mH_t\right ].
\end{equation*}

Further recalling that $V_i^{\pi^\ast}(\mH_t)=Q_i^{\pi^\ast}(\mH_t,\pi_{t,i}^\ast(\mH_t^i))$ by definition, we arrive at
\begin{align}\label{eq:second term Q diff}
&\quad Q_i^{\pi^\ast}(\mH_t,\pi_{t,i}'(\mH_t^i))-Q_i^{\pi^\ast}(\mH_t,\pi_{t,i}^\ast(\mH_t^i))\\
&\le \left (\tilde V_i^{\pi^{\ast,r}}(\tilde \mH_t)-V_i^{\pi^\ast}(\mH_t)\right )+\E_{\pi_{t,i}'\circ \pi_{t,-i}^\ast} \left [V_i^{\pi^\ast}(\mH_{t+1})-\tilde V_i^{\pi^\ast}(\tilde \mH_{t+1})\middle \vert \mH_t\right ].
\end{align}

\paragraph{Finishing the Proof.} For any $\mH_t\sim \pi_i\circ \pi_{-i}^\ast$, plugging \Cref{eq:first term Q diff,eq:second term Q diff} into \Cref{eq:Q diff decomposition} gives
\begin{equation*}
\begin{aligned}
&\quad Q_i^{\pi^\ast}(\mH_t,\pi_{t,i}(\mH_t^i))-V_i^{\pi^\ast}(\mH_t)\\
&\le \E_{\pi_{t,i}\circ \pi_{t,-i}^\ast}\left [V_i^{\pi^\ast}(\mH_{t+1})-\tilde V_i^{\pi^\ast}(\tilde \mH_{t+1})\middle \vert \mH_t\right ]-\E_{\pi_{t,i}'\circ \pi_{t,-i}^\ast}\left [V_i^{\pi^\ast}(\mH_{t+1})-\tilde V_i^{\pi^\ast}(\tilde \mH_{t+1})\middle \vert \mH_t\right ]\\
&\quad +(\tilde V_i^{\pi^{\ast,r}}(\tilde \mH_t)-V_i^{\pi^\ast}(\mH_t))+\E_{\pi_{t,i}'\circ \pi_{t,-i}^\ast}\left [V_i^{\pi^\ast}(\mH_{t+1})-\tilde V_i^{\pi^\ast}(\tilde \mH_{t+1})\middle \vert \mH_t\right ]\\
&=(\tilde V_i^{\pi^{\ast,r}}(\tilde \mH_t)-V_i^{\pi^\ast}(\mH_t))-\E_{\pi_{t,i}\circ \pi_{t,-i}^\ast}\left [\tilde V_i^{\pi^\ast}(\tilde \mH_{t+1})-V_i^{\pi^\ast}(\mH_{t+1})\middle \vert \mH_t\right ].
\end{aligned}
\end{equation*}
Further plugging it into the performance difference lemma \Cref{eq:performance difference lemma} and using telescoping sums, we yield
\begin{align*}
&\quad V_i^{\pi_i\circ \pi_{-i}^\ast}(\mH_1)-V_i^{\pi^\ast}(\mH_1)\\
&=\sum_{t=1}^T \E_{\mH_t\sim \pi_i\circ \pi_{-i}^\ast}\left [(\tilde V_i^{\pi^{\ast,r}}(\tilde \mH_t)-V_i^{\pi^\ast}(\mH_t))-\E_{\pi_{t,i}\circ \pi_{t,-i}^\ast}\left [\tilde V_i^{\pi^\ast}(\tilde \mH_{t+1})-V_i^{\pi^\ast}(\mH_{t+1})\middle \vert \mH_t\right ]\right ]\\
&=\sum_{t=1}^T \E_{\mH_t\sim \pi_i\circ \pi_{-i}^\ast}\left [(\tilde V_i^{\pi^{\ast,r}}(\tilde \mH_t)-V_i^{\pi^\ast}(\mH_t))\right ]-\sum_{t=2}^T \E_{\mH_t\sim \pi_i\circ \pi_{-i}^\ast}\left [(\tilde V_i^{\pi^{\ast,r}}(\tilde \mH_t)-V_i^{\pi^\ast}(\mH_t))\right ]\\
&=\tilde V_i^{\pi^{\ast,r}}(\tilde \mH_1)-V_i^{\pi^\ast}(\mH_1)\le \Delta(\pi^{\ast,r}),
\end{align*}
where the last inequality is due to the definition of $\Delta(\pi^{\ast,r})$ in \Cref{eq:definition of Delta}. This finishes the proof.
\end{proof}
\section{Regret and Audits Guarantees (Generalized \Cref{thm:deterministic main thm})}\label{sec:appendix regret and audits}
The objective of this section is to prove \Cref{thm:regret and audits formal}, a generalized version of \Cref{thm:deterministic main thm} under $R$ and $\Delta(\pi^{\ast,r})$. While the $\mR_T$ part almost directly comes from \Cref{lem:V lower and upper bound formal}, the $\mB_T$ part is highly non-trivial.

\begin{theorem}[Regret and Audits Guarantees of \mechname]\label{thm:regret and audits formal}
Suppose that \Cref{assumption:min report} holds. Consider the generalized auxiliary game with tolerance $R$ (\Cref{def:auxiliary game formal}). Any PBE $\pi^{\ast,r}$ therein induces a $\Delta(\pi^{\ast,r})$-Approximate Bayesian Equilibrium $\pi^\ast$ in the original game (see \Cref{thm:equivalence of PBE formal}), such that
\begin{align*}
\mR_T(\pi^\ast,\mechname)&\le K^2+\sum_{i=1}^K\left (\tilde V_i^{\pi^{\ast,r}}(\tilde \mH_1)-V_i^{\pi^\ast}(\mH_1)\right ),\\
\mB_T(\pi^\ast,\mechname)&=\O\left (\frac{K^3}{c}\log T\right )+\frac{16 K^2}{c} \E_{\pi^\ast}\left [\sum_{t=1}^T \1[(i_t\text{ not eliminated})\wedge (v_{t,i_t}>u_{t,i_t}+R)\wedge o_t]\right ].
\end{align*}
\end{theorem}
The second terms on the both RHS's are solely for the stochastic audit model. Under the perfect audit model (\Cref{assump:noiseless}) and adversarial audit model (\Cref{assump:adv noise model}), they are both zero, which gives $\mR_T(\pi^\ast,\mechname)\le K^2$ and $\mB_T(\pi^\ast,\mechname)=\O(\frac{K^3}{c} \log T)$, as claimed in \Cref{thm:deterministic main thm,thm:adv noise model}.
Under the stochastic audit model (\Cref{assump:stoc noise model}) where each audit can be wrong with probability $\epsilon$, these two RHS's become those in \Cref{thm:stoc noise model}. Before proving this theorem, we first give a generalized version of \Cref{lem:number of under-reporting bound}, which controls the expected number of rounds where an agent \emph{under-reports}.

\subsection{Number of Under-Reports (Generalized \Cref{lem:number of under-reporting bound})}\label{sec:number of under-reporting bound formal}
\begin{lemma}[Number of Under-Reports]\label{lem:number of under-reporting bound formal}
Suppose that the generalized auxiliary game is a $(R,\Delta)$-approximation and that \Cref{assumption:min report} holds. Consider the $\pi^\ast$ suggested by \Cref{thm:equivalence of PBE formal}.
For any round $t\in [T]$ and agent $i\in [K]$, let $D_{t,i}$ be the indicator for the event that agent $i$ is alive in round $t$, does not win, but would have won had they reported the maximum report allowed by \Cref{def:auxiliary game formal}:
\begin{equation*}
D_{t,i}:=\1[i\in \mA_t]\1[i_t\ne i]\1[u_{t,i}\geq c]\1[u_{t,i}+R>v_{t,j},\forall j\in \mA_t\setminus\{i\}].
\end{equation*}
We claim that under $\pi^\ast$, $D_{t,i}$ happens rarely (in expectation) throughout the game, namely
\begin{equation*}
\E_{\pi^\ast}\left [\sum_{i=1}^K \sum_{t=1}^T D_{t,i}\right ]\le \frac{2K^3}{c} + \frac{2K^2}{c} \E_{\pi^\ast}\left [\sum_{t=1}^T \1[(i_t\text{ not eliminated})\wedge (v_{t,i_t}>u_{t,i_t}+R)\wedge o_t]\right ],
\end{equation*}
where we recall that $o_t$ is the indicator for whether an audit is requested in round $t$.
\end{lemma}
\begin{proof}
Let $\mH_t\sim \pi^\ast$ be the abbreviation for ``sampling round-$t$ history $\mH_t$ from $\pi^\ast$ under mechanism \mechname'', and let $\tilde \mH_t$ be the corresponding simplified history of such a $\mH_t$. Since the definition of $D_{t,i}$ only involves the set of alive agents $\mA_t$, the allocation $i_t$, the reports $\bm u_t$, and the true values $\bm v_t$, it remains a valid (measurable) indicator in the auxiliary game (\Cref{def:auxiliary game formal}). Therefore, we can rewrite the LHS as
\begin{equation}\label{eq:sum D original game}
\E_{\pi^\ast}\left [\sum_{i=1}^K \sum_{t=1}^T D_{t,i}\right ]=\sum_{t=1}^T \E_{\mH_t\sim \pi^\ast} \left [\E_{\pi_t^\ast}\left [\sum_{i=1}^K D_{t,i}\middle \vert \mH_t\right ]\right ]=\sum_{t=1}^T \E_{\mH_t\sim \pi^\ast} \left [\E_{\pi_t^{\ast,r}}\left [\sum_{i=1}^K D_{t,i}\middle \vert \tilde \mH_t\right ]\right ].
\end{equation}

Utilizing agents' incentives in the auxiliary game, we aim to argue that under-reporting is nonprofitable and thus $D_{t,i}$ happens rarely. To do so, analogous to the no-over-report strategy in the proof of \Cref{thm:equivalence of PBE formal} (see Item~\ref{item:coupling of over-reports} therein), we construct \emph{no-under-report} version of $\pi_{t,i}^{\ast,r}$. Specifically, for any round-$t$ simplified history $\tilde \mH_t\in \tilde H_t$ and alive agent $i\in \mA_t$, define a no-under-report strategy $\pi_{t,i}^{r}$ via a coupling with $\pi_{t,i}^{\ast,r}$:
\begin{equation*}
v_{t,i}:=\max(v_{t,i}^\ast,u_{t,i}+R),\quad \text{where }v_{t,i}^\ast\sim \pi_{t,i}^{\ast,r}(u_{t,i};\tilde \mH_{t}),\quad \forall t\in [T],\tilde \mH_t\in \tilde H_t,i\in \mA_t.
\end{equation*}
Since $\pi^{\ast,r}$ is a PBE in the auxiliary game, agent-$i$'s unilateral deviation to $\pi_{t,i}^r$ must be unprofitable, i.e.,
\begin{equation*}
0\leq \tilde V_i^{\pi^{\ast,r}}(\tilde \mH_{t}) - \tilde V_i^{(\pi_{t,i}^r\circ \pi_{t,-i}^{\ast,r})\diamond_t \pi^{\ast,r}}(\tilde \mH_{t}),
\end{equation*}
where recall that $\pi\diamond_t\pi'$ means following strategy $\pi$ up to round $t$ and following strategy $\pi'$ afterward.

In the auxiliary game, the only case where $\pi_{t,i}^{\ast,r}$ and $\pi_{t,i}^r$ make a difference is when $D_{t,i}=1$.
Suppose that indeed $D_{t,i}=1$. When unilaterally deviating to $\pi_{t,i}^r$, agent $i$ reports $v_{t,i}=u_{t,i}+R$, wins in round $t$ (thus achieving utility $u_{t,i}$), stays alive (according to \Cref{item:no over-report unless p>=1 formal} of \Cref{def:auxiliary game formal}), and faces a next-round history of $\tilde\mH_{t+1}^{(0)}:=(t+1,\mA_t)$. Therefore, the unilateral deviation yields an auxiliary-game utility of
\begin{equation*}
u_{t,i} + \tilde V_i^{{\pi}^{\ast,r}}(\tilde\mH_{t+1}^{(0)}).
\end{equation*}

On the other hand, sticking to the PBE strategy $\pi_{t,i}^{\ast,r}$ means that agent $i$ reports $v_{t,i}^\ast<u_{t,i}+R$ and loses round $t$. There are two possibilities for the next-round simplified history: either the winner $i_{t}\ne i$ is caught lying and consequently eliminated---resulting in $\tilde\mH_{t+1}^{(i_t)}:=(t+1,\mA_t\setminus\{i_{t}\})$---or the winner stays alive and gives the same $\tilde\mH_{t+1}^{(0)}=(t+1,\mA_t)$. Thus sticking to $\pi_{t,i}^{\ast,r}$ yields an auxiliary-game utility of
\begin{align*}
\tilde V_i^{{\pi}^{\ast,r}}(\tilde\mH_{t+1}^{(0)}) + \1[i_{t}\text{ eliminated}] \paren{ \tilde V_i^{{\pi}^{\ast,r}}(\tilde\mH_{t+1}^{(i_t)}) - \tilde V_i^{{\pi}^{\ast,r}}(\tilde\mH_{t+1}^{(0)}) }.
\end{align*}

As $\pi^{\ast,r}$ is an auxiliary-game PBE and these two strategies give identical outcomes when $D_{t,i}=0$,
\begin{align*}
0&\leq \tilde V_i^{\pi^{\ast,r}}(\tilde \mH_{t}) - \tilde V_i^{(\pi_{t,i}^r\circ \pi_{t,-i}^{\ast,r})\diamond_t \pi^{\ast,r}}(\tilde \mH_{t})\\
&= \E_{\pi_{t,i}^{\ast,r}} \left [u_{t,i}\1[i_t=i]+\tilde V^{\pi^{\ast,r}}(\tilde \mH_{t+1}) \middle \vert \tilde \mH_t,D_{t,i}=1\right ] \E_{\pi^{\ast,r}}\left [D_{t,i}=1\middle \vert \tilde \mH_{t}\right ]\\
&\quad - \E_{\pi_{t,i}^r\circ \pi_{t,-i}^{\ast,r}} \left [u_{t,i}\1[i_t=i]+\tilde V^{\pi^{\ast,r}}(\tilde \mH_{t+1}) \middle \vert \tilde \mH_t,D_{t,i}=1\right ] \E_{\pi^{\ast,r}}\left [D_{t,i}=1\middle \vert \tilde \mH_{t}\right ] \\
&=\E_{\pi^{\ast,r}}\left [D_{t,i}=1\middle \vert \tilde \mH_{t}\right ] \Bigg (\tilde V_i^{{\pi}^{\ast,r}}(\tilde\mH_{t+1}^{(0)}) + \E\left [\1[i_{t}\text{ eliminated}] \paren{ \tilde V_i^{{\pi}^{\ast,r}}(\tilde\mH_{t+1}^{(i_t)}) - \tilde V_i^{{\pi}^{\ast,r}}(\tilde\mH_{t+1}^{(0)}) }\right ]\\
&\quad \qquad \qquad \qquad \qquad - \E\left [u_{t,i}\middle \vert \tilde \mH_t,D_{t,i}=1\right ] - \tilde V_i^{{\pi}^{\ast,r}}(\tilde\mH_{t+1}^{(0)})\Bigg ).
\end{align*}
Canceling out the two $\tilde V_i^{{\pi}^{\ast,r}}(\tilde\mH_{t+1}^{(0)})$'s, using the fact that $D_{t,i}=1$ infers $u_{t,i}\ge c$, and summing up for all $i\in \mA_t$, we have the following inequality which formalizes the arguments in \Cref{lem:number of under-reporting bound}:
\begin{align}
\E\sqb{\1[i_t \text{ eliminated}]\sum_{i\in \mA_t,D_{t,i}=1} \paren{   \tilde V_i^{{\pi}^{\ast,r}}(\tilde\mH_{t+1}^{(i_t)}) - \tilde V_i^{{\pi}^{\ast,r}}(\tilde\mH_{t+1}^{(0)}) }\middle \vert \tilde \mH_t}\ge c\E\sqb{\sum_{i\in\mA_t}D_{t,i}\middle \vert \tilde \mH_t}.\label{eq:summed_PBE_equations}
\end{align}
This inequality roughly says that, every time $D_{t,i}=1$, the expected increase in agent $i$'s V-function---due to a potential elimination of the winner $i_t$---must compensate the immediate utility loss of $u_{t,i}\ge c$.

It only remains to upper bound the LHS of \Cref{eq:summed_PBE_equations}.
Fix a realization of utilities $\bm u_t$ and reports $\bm v_t$ such that the winner $i_t$ was eliminated. Let $S:=\{i\in\mA_t\mid D_{t,i}=1\}$. Then we \emph{i)} lower bound $\sum_{i\in S} \tilde V_i^{{\pi}^{\ast,r}}(\tilde\mH_{t+1}^{(0)})$, and \emph{ii)} upper bound $\sum_{i\in S} \tilde V_i^{{\pi}^{\ast,r}}(\tilde\mH_{t+1}^{(i_t)})$. From the lower bound for alive agents in \Cref{lem:V lower and upper bound formal},
\begin{equation}\label{eq:lower_bound_sum_values}
\sum_{i\in S} \tilde V_i^{{\pi}^{\ast,r}}(\tilde\mH_{t+1}^{(0)}) \geq (T-t)\sum_{i\in S}\mu_{t+1,i}^{(0)} - \lvert S\rvert \cdot K = (T-t)\sum_{i\in S}\mu_{t,i} - \lvert S\rvert \cdot K,
\end{equation}
where $\mu_{t+1,i}^{(0)}$ denotes the ex-ante first-best utility of agent $i$ in $\tilde\mH_{t+1}^{(0)}$, and the equality $\mu_{t+1,i}^{(0)}=\mu_{t,i}$ utilizes the fact that the set of alive agents in $\tilde\mH_{t+1}^{(0)}$ is still $\mA_t$. On the other hand, we have
\begin{align}
\sum_{i\in S} \tilde V_i^{{\pi}^{\ast,r}}(\tilde\mH_{t+1}^{(i_t)})&=\sum_{i\in \mA_{t+1}}\tilde V_i^{{\pi}^{\ast,r}}(\tilde\mH_{t+1}^{(i_t)})-\sum_{i\in \mA_{t+1}\setminus S} \tilde V_i^{{\pi}^{\ast,r}}(\tilde\mH_{t+1}^{(i_t)})\\
&\overset{(a)}{\le} \left ((T-t) \sum_{i\in\mA_{t+1}}\mu_{t+1,i}^{(i_t)} + K^2 \right ) - \left ((T-t)\sum_{i\in\mA_{t+1}\setminus S}\mu_{t+1,i}^{(i_t)} -|\mA_{t+1}\setminus S|\cdot K\right ) \\
&= (T-t)\sum_{i\in S}\mu_{t+1,i}^{(i_t)} + \Big (K+ (|\mA_t|-|S|-1)\Big )\cdot K,\label{eq:new_upper_bound_sum_S}
\end{align}
where (a) applied the intermediate technical result \Cref{eq:upper_bound_max_welfare} (proved in \Cref{lem:V lower and upper bound formal}) to the first sum, and the lower bound part of \Cref{lem:V lower and upper bound formal} to the second sum. Here, $\mu_{t+1,i}^{(i_t)}$ denotes the ex-ante first-best utility of agent $i$ in $\tilde\mH_{t+1}^{(i_t)}$, which can differ from $\mu_{t,i}$ only if agent $i_t$ is the winner; formally we have
\begin{align*}
\sum_{i\in S}\mu_{t+1,i}^{(i_t)} 
&= \E_{\bm u\sim \mV}\sqb{\max_{i\in S}u_i \cdot \1\sqb{\max_{i\in S}u_i> \max_{i\in \mA_t\setminus(S\cup\{i_t\})}u_i}}\\
&= \E_{\bm u\sim \mV}\sqb{\max_{i\in S}u_i \cdot \paren{ \1\sqb{\max_{i\in S}u_i> \max_{i\in \mA_t\setminus S}u_i} + \1\sqb{u_{i_t}>\max_{i\in S}u_i> \max_{i\in \mA_t\setminus S}u_i }}} \\
&\leq \sum_{i\in S}\mu_{t,i} + \E_{\bm u\sim \mV}\sqb{u_{i_t} \cdot \1\sqb{u_{i_t}>\max_{i\in S}u_i> \max_{i\in \mA_t\setminus (S\cup\{i_t\})}u_i }} \leq \mu_{t,i_t}+\sum_{i\in S}\mu_{t,i}.
\end{align*}
Plugging it back into \Cref{eq:new_upper_bound_sum_S} and combining it with \Cref{eq:lower_bound_sum_values} gives
\begin{align*}
\sum_{i\in S} \tilde V_i^{{\pi}^{\ast,r}}(\tilde\mH_{t+1}^{(i_t)})-\sum_{i\in S} \tilde V_i^{{\pi}^{\ast,r}}(\tilde\mH_{t+1}^{(0)})&\le (T-t)\left (\mu_{t,i_t}+\sum_{i\in S}\mu_{t,i}\right ) + \Big (K+ (|\mA_t|-|S|-1)\Big )\cdot K\\
&\quad - \left ((T-t)\sum_{i\in S}\mu_{t,i} - \lvert S\rvert \cdot K\right )\\
&=(T-t) \mu_{t,i_t} + \Big (K+\lvert \mA_t\rvert-1\Big )\cdot K.
\end{align*}

Finally, due to the lower bound for eliminated agents in \Cref{lem:V lower and upper bound formal}, the elimination of $i_t$ (i.e., $i_t\not \in \mA_{t+1}$) implies $(T-t)\mu_{t,i_t}\le K$. Henceforth the RHS is simply bounded by $2K^2$. Plugging back into \Cref{eq:summed_PBE_equations},
\begin{equation*}
\E_{\pi_t^{\ast,r}}\left [\sum_{i=1}^K D_{t,i}\middle \vert \tilde \mH_t\right ]\le \frac{2K^2}{c} \E_{\pi_t^{\ast,r}}\left [\1[i_t\text{ eliminated}]\middle \vert \tilde \mH_t\right ]\overset{(a)}{=}\frac{2K^2}{c} \E_{\pi_t^{\ast,r}}\left [\1[v_{t,i_t}>u_{t,i_t}+R]\middle \vert \tilde \mH_t\right ],
\end{equation*}
where (a) uses \Cref{item:eliminate immediately formal} of \Cref{def:auxiliary game formal}. Moving back to the original game using \Cref{eq:sum D original game}:
\begin{align*}
&\quad \E_{\pi^\ast}\left [\sum_{i=1}^K \sum_{t=1}^T D_{t,i}\right ]=\sum_{t=1}^T \E_{\mH_t\sim \pi^\ast} \left [\E_{\pi_t^\ast}\left [\sum_{i=1}^K D_{t,i}\middle \vert \mH_t\right ]\right ]=\sum_{t=1}^T \E_{\mH_t\sim \pi^\ast} \left [\E_{\pi_t^{\ast,r}}\left [\sum_{i=1}^K D_{t,i}\middle \vert \tilde \mH_t\right ]\right ]\\
&\le \frac{2K^2}{c} \sum_{t=1}^T \E_{\mH_t\sim \pi^\ast} \left [\E_{\pi^{\ast,r}}\left [\1[v_{t,i_t}>u_{t,i_t}+R]\middle \vert \tilde \mH_t\right ]\right ]\overset{(a)}{=} \frac{2K^2}{c} \sum_{t=1}^T \E_{\mH_t\sim \pi^\ast} \left [\E_{\pi^{\ast}} \left [\1[v_{t,i_t}>u_{t,i_t}+R]\middle \vert \mH_t\right ]\right ]\\
&\overset{(b)}{\le} \frac{2K^2}{c}\E_{\pi^\ast} \left [\sum_{t=1}^T \left (\1[i_t\text{ eliminated}]+\1[(i_t\text{ not eliminated})\wedge (v_{t,i_t}>u_{t,i_t}+R)\wedge o_t]\right )\right ]\\
&\overset{(c)}{\le} \frac{2K^3}{c} + \frac{2K^2}{c} \E_{\pi^\ast}\left [\sum_{t=1}^T \1[(i_t\text{ not eliminated})\wedge (v_{t,i_t}>u_{t,i_t}+R)\wedge o_t]\right ].
\end{align*}
In (a), we used the facts that $\pi^\ast=(\pi^{\ast,r},\pi^{\ast,f})$ and $\tilde \mH_t$ is the simplified history corresponding to $\mH_t$. In (b), we used the following facts: from \Cref{item:no over-report unless p>=1 formal} of \Cref{def:auxiliary game formal}, $v_{t,i_t}>u_{t, i_t}+R$ implies $p_{t,i_t}=\frac{1+K^2}{(T-t) cq_{t,i_t}}\ge 1$; due to the well-behaved flagging strategy $\pi^{\ast,f}$, any $\mH_t\sim \pi^\ast$ satisfies $\hat q_{t,j}\le 4 q_{t,j}$ for all $j\in [K]$; and these together imply $\hat p_{t,i_t}=\min(\frac{4(1+K^2)}{(T-t) c \hat q_{t,i_t}},1)=1$, hence $o_t=1$ (according to \Cref{line:check prob} in \Cref{alg:mechanism}), i.e., agent $i_t$ is always audited. In (c), we used the fact that $\1[i_t\text{ eliminated}]=\lvert \mA_t\rvert-\lvert \mA_{t+1}\rvert$. This finishes the proof.
\end{proof}

\subsection{Controlling Regret and Number of Audits (\Cref{thm:regret and audits formal})}

\begin{proof}[Proof of \Cref{thm:regret and audits formal}]
For the regret $\mR_T$, using \Cref{lem:V lower and upper bound formal}, we have for all $i\in [K]$:
\begin{equation*}
V_i^{\pi^\ast}(\mH_1)=\tilde V_i^{\pi^{\ast,r}}(\tilde \mH_1)-\left (\tilde V_i^{\pi^{\ast,r}}(\tilde \mH_1)-V_i^{\pi^\ast}(\mH_1)\right )\ge T \mu_{1,i}-K-\left (\tilde V_i^{\pi^{\ast,r}}(\tilde \mH_1)-V_i^{\pi^\ast}(\mH_1)\right ).
\end{equation*}
Recall the definition of $\mu_{1,i}$ from \Cref{eq:fair share formal} and the fact that $\mA_1=[K]$, we have
\begin{align*}
\sum_{i=1}^K \mu_{1,i}&=\E_{\bm u\sim \mV} \left [\sum_{i=1}^K u_i \cdot \1[i\in \mA_1] \1[u_i\ge c]\1[u_i>u_j,\forall j\in \mA_1\setminus \{i\}]\right ]\\
&=\E_{\bm u\sim \mV} \left [\max_{i\in [K]} u_i \cdot \1\left [\max_{i\in [K]} u_i\ge c\right ]\right ]\overset{(a)}{=}\E_{\bm u\sim \mV} \left [\max_{i\in [K]} u_i\right ],
\end{align*}
where (a) is due to the existence of some $i_0\in [K]$ such that $u_{i_0}\ge c$ a.s. (from \Cref{assumption:min report}). Therefore, $\sum_{i=1}^K \mu_{1,i}$ is exactly the optimal per-round social welfare, and consequently the regret is bounded as
\begin{equation*}
\mR_T(\pi^\ast,\mechname)= T \E_{\bm u\sim \mV} \left [\max_{i\in [K]} u_i\right ] - \sum_{i=1}^K V_i^{\pi^\ast}(\mH_1) \le K^2 + \sum_{i=1}^K \left (\tilde V_i^{\pi^{\ast,r}}(\tilde \mH_1)-V_i^{\pi^\ast}(\mH_1)\right ).
\end{equation*}

We now control $\mB_T$.
To formalize the {estimation phase} and {incentive-compatible phase} in \Cref{sec:sketch varying-p Part II}, we formally define the \emph{estimation epochs}: For each $\ell=1,2,\ldots,K$, let $s_\ell$ be the round where we eliminate the $(\ell-1)$-th agent and reset all empirical estimations $\hat q_{t,i}$'s to $0$ (\Cref{line:elimination} in \Cref{alg:mechanism}), namely
\begin{equation*}
s_\ell=\min \Big \{t\in[T]\mid \lvert \mA_t\rvert =T-i+1\Big \}\cup\{+\infty\}, \quad \ell\in[K],
\end{equation*}
The $\ell$-th estimation epoch $\mE_\ell$ is all the rounds $\mE_\ell:=\{s_\ell,s_\ell+1,\ldots,s_{\ell+1}-1\}$. For each epoch $\ell\in [L]$ and alive agent $i\in \mA_{s_\ell}$, we decompose the rounds in which agent $i$ wins---namely $\mE_{\ell,i}:=\{t\in \mE_\ell\mid i_t=i\}$---into two phases: The \emph{estimation} phase $\mE_{\ell,i}^E$ contains the rounds where $\hat q_{t,i}=0$, and the remaining ones forms the \emph{incentive-compatible} phase $\mE_{\ell,i}^I$.
To simplify notations, we denote by $\hat q_{\ell,i}$ the final estimation $\hat q_{t,i}$ in epoch $\ell$ (which no agent answer with $f_{t,i}=1$; recall \Cref{line:flagging} of \Cref{alg:mechanism}). This $\hat q_{\ell,i}$ is then fixed throughout agent $i$'s incentive-compatible phase $\mE_{\ell,i}^I$. In the case where there is no incentive-compatible phase during the $\ell$-th epoch (i.e., agent $i$'s estimation hasn't been finalized when someone gets eliminated), we write $\hat q_{\ell,i}=0$.

\paragraph{Estimation Phase.}
Same as \Cref{lem:number of under-reporting bound formal}, let $D_{t,i}=\1[i\in \mA_t]\1[i_t\ne i]\1[u_{t,i}\geq c]\1[u_{t,i}+R>v_{t,j},\forall j\in \mA_t\setminus\{i\}]$. When agent $i$ wins in round $t\in \mE_{\ell,i}^E$ during their estimation phase, at least one of these holds:
\begin{itemize}
\item Agent $i$ over-reports $v_{t,i}>u_{t,i}+R$. Thus $i$ is eliminated and $s_{\ell+1}=t+1$. This can only happen once.
\item It's the first-best for agent $i$ to win. Let $F_{t,i}$ be its indicator, i.e., $F_{t,i}:=\1[u_{t,i}>\max_{j\in \mA_{s_\ell}\setminus \{i\}} u_{t,j}]$.
\item It's the first-best for some other agent $j\in \mA_{s_\ell}\setminus \{i\}$ to win, but they under-reported (so $D_{t,j}=1$).
\end{itemize}
Therefore, we may upper bound the number of wins of agent $i$ from $s_\ell$ and until some $\tau\ge s_\ell$ as
\begin{align}
\Big \lvert \set{ t\in[s_\ell,\tau]\mid i_{t}=i }\Big \rvert
&\leq 1+\sum_{t=s_\ell}^{\tau} F_{t,i} + \sum_{t\in \mE_\ell} \1[i_{t}=i] \cdot \bigvee_{j\in\mA_{s_\ell}\setminus\{i\}} D_{t,j}.\label{eq:upper_bound_nb_wins}
\end{align}
In words, the number of rounds that agent $i$ actually wins (LHS) is at most its first-best winning rounds (those $t$'s with $F_{t,i}=1$) plus those other agents miss (those $t$'s where any other $j\ne i$ has $D_{t,j}=1$).

On the other hand, we lower bound the same quantity in another way. In a round $t\in \mE_\ell$ where agent $i$ would have won if \emph{all} agents reported truthfully (so $F_{t,i}=1$), one of these scenarios holds
\begin{itemize}
\item Agent $i$ indeed wins, i.e., $i_t=i$.
\item Agent $i$ does not win because themselves are under-reporting, i.e., we have $D_{t,i}=1$.
\item Some other agent $i_t$ won by over-reporting. This means $s_{\ell+1}=t+1$ and can happen at most once.
\end{itemize}
Therefore, for any $\tau\ge s_\ell$, we also have a lower bound on the number of wins of agent $i$:
\begin{equation}\label{eq:lower_bound_nb_wins}
\Big \lvert \set{ t\in[s_\ell,\tau]\mid i_{t}=i }\Big \rvert \geq \sum_{t=s_\ell}^{\tau} F_{t,i} - \sum_{t\in \mE_\ell}D_{t,i} - 1.
\end{equation}
In words, the number of rounds that agent $i$ actually wins (LHS) is at least their first-best winning rounds (those $t$'s with $F_{t,i}=1$) minus the rounds themselves miss (those $t$'s with $D_{t,i}=1$).

Comparing these two inequalities with the well-behaved flagging strategy $\pi^{\ast,f}$ (\Cref{def:well-behaved flagging}), we can control the end of $\mE_{\ell,i}^E$ (the estimation phase of agent $i$ in epoch $\ell$) as follows:
\begin{claim}\label{claim:end of estimation phase}
For $\ell\in [K]$ and $i\in \mA_{s_\ell}$, let ${\mathfrak t}_{\ell,i}$ be the first round $t\in \mE_\ell$ such that agent $i$ wins for at least $\max(4+4\sum_{t\in \mE_\ell} \1[i_t=i] \cdot \bigvee_{j\in \mA_{s_\ell}\setminus \{i\}} D_{t,j}, 3+3 \sum_{t\in \mE_\ell} D_{t,i})$ times during epoch $\mE_\ell$ and that
\begin{equation}\label{eq:constraint_F}
\frac{1}{t-s_\ell+1}\sum_{\tau=s_\ell}^t F_{t,i} \in\sqb{\frac{q_{s_\ell,i}}{3}, 3q_{s_\ell,i}}.
\end{equation}
If such $t\in \mE_\ell$ does not exist, we write ${\mathfrak t}_{\ell,i}:=s_{\ell+1}-1$. Then $\mE_{\ell,i}^E$ must end \emph{on or before} ${\mathfrak t}_{\ell,i}$.
\end{claim}
Intuitively, the definition of $\mathfrak t_{\ell,i}$ roughly means the estimation $\hat q_{\mathfrak t_{\ell,i},i}$ generated using outcomes of rounds $s_\ell,s_\ell+1,\ldots,\mathfrak t_{\ell,i}$ is close to the actual $q_{i}$. Therefore, no agent will flag this generated estimation (recall the flagging component from \Cref{line:flagging} in \Cref{alg:mechanism} and the well-behaved flagging strategy from \Cref{def:well-behaved flagging}), which means $\mE_{\ell,i}^E$ must end.
The formal proof of \Cref{claim:end of estimation phase} is presented immediately after this proof.

Carefully analyzing ${\mathfrak t}_{\ell,i}$ and applying concentration inequalities, we control the number of audits spent on agent $i$ during $\mE_{\ell,i}^E$ using the number of under-reports. The proof of \Cref{claim:audits in estimation phase} succeeds that of \Cref{claim:end of estimation phase}:
\begin{claim}\label{claim:audits in estimation phase}
The number of audits that we spend on agent $i$ during epoch $\ell$, namely $\sum_{t\in \mE_{\ell,i}^E} o_t$, satisfies
\begin{equation*}
\E\left [\sum_{i\in\mA_{s_\ell}} \sum_{t\in \mE_{\ell,i}^E} o_t\middle \vert \mH_{s_\ell}\right ]\le 589 K+7 \E\left [\sum_{i\in\mA_{s_\ell}} \sum_{t\in \mE_\ell} D_{t,i}\middle \vert \mH_{s_\ell}\right ],\quad \forall \ell\in [L].
\end{equation*}
\end{claim}

Summing up $\ell\in [K]$ and combining with the bound on $\E[\sum_{t=1}^T \sum_{i\in \mA_t} D_{t,i}]$ from \Cref{lem:number of under-reporting bound formal},
\begin{align*}
&\quad \E\sqb{\sum_{t=1}^T o_t \1[\hat q_{t,i_t}=0] }=\E\sqb{\sum_{\ell=1}^K \sum_{i\in\mA_{s_\ell}} \sum_{t\in \mE_{\ell,i}^E} o_t}=589 K^2+7 \E\left [\sum_{t=1}^T \sum_{i\in \mA_t} D_{t,i}\right ]\\
&=\O\left (K^2+\frac{K^3}{c}\right ) + \frac{14 K^2}{c} \E_{\pi^\ast}\left [\sum_{t=1}^T \1[(i_t\text{ not eliminated})\wedge (v_{t,i_t}>u_{t,i_t}+R)\wedge o_t]\right ].
\end{align*}

\paragraph{Incentive-Compatible Phase.} We now turn to the incentive-compatible phase, where an alive agent $i\in \mA_{s_\ell}$'s winning probability has been estimated as $\hat q_{\ell,i}\in [\frac 14 q_{\ell,i},4q_{\ell,i}]$ (thanks to the well-behaved flagging strategy $\pi^{\ast,f}$ in \Cref{def:well-behaved flagging}). Consider the number of audits spent on agent $i$'s incentive-compatible phase $\mE_{\ell,i}^I$:
\begin{equation*}
\E\Bigg [\sum_{t\in \mE_{\ell,i}^I} o_t\Bigg ]=\E\sqb{\sum_{t\in \mE_\ell}o_t \1[i_t=i] \1[\hat q_{t,i}>0] } = \E\sqb{\sum_{t\in \mE_\ell}\hat p_{i,t} \1[i_t=i] \1[\hat q_{t,i}>0]},
\end{equation*}
where we recall that $\hat p_{t,i}=\frac{4(1+K^2)}{(T-t) c \hat q_{t,i}}$ is the audit probability used by \mechname. Similar to \Cref{eq:upper_bound_nb_wins}, for each round $t\in \mE_\ell$, we decompose the term $\1[i_t=i]$ into three cases: \textit{i)} agent $i$ is the first-best agent (hence $F_{t,i}=1$), \textit{ii)} agent $i$ themselves over-reported some $v_{t,i}>u_{t,i}+R$, and \textit{iii)} another agent $j\in \mA_{s_\ell}\setminus \{i\}$ is the first-best but chooses to under-report (hence $D_{t,j}=1$). Therefore,
\begin{align*}
\E\Bigg [\sum_{t\in \mE_{\ell,i}^I} o_t\Bigg ]&\le \E\left [\sum_{t\in \mE_\ell} \frac{4(1+K^2)}{(T-t) c \hat q_{t,i}} \left (F_{t,i} + \1[i_t=i]\cdot \Big (\1[v_{t,i}>u_{t,i}+R]+\bigvee_{j\in \mA_{s_\ell}\setminus \{i\}} D_{t,j}\Big )\right ) \1[\hat q_{t,i}>0]\right ]\\
&\le \E\left [\sum_{t\in \mE_\ell} \frac{4(1+K^2)}{(T-t) c \hat q_{t,i}} q_{t,i} \1[\hat q_{t,i}>0]\right ]+\E\left [1+\sum_{t\in \mE_\ell} \1[i_t=i] \bigvee_{j\in \mA_{s_\ell}\setminus \{i\}} D_{t,j}\right ],
\end{align*}
where the second inequality uses the following facts: by definition, the first-best indicators $(F_{t,i})_{t\in \mE_\ell}$ form an i.i.d. sequence of $\text{Ber}(q_{\ell,i})$; $\hat p_{t,i}$ is upper bounded by 1; and from \Cref{item:no mark up unless p>=1} of \Cref{def:auxiliary game formal} (the auxiliary game), agent $i$ can only win by reporting $v_{t,i}>u_{t,i}+R$ once, after which they are eliminated a.s.

Due to the well-behaved flagging strategy ${\pi}^{\ast,f}$, $\1[\hat q_{t,i}>0]$ infers $\hat q_{t,i}\geq q_{t,i}/4$. Thus the first term on the RHS is further upper bounded as
\begin{equation*}
\E\left [\sum_{t\in \mE_\ell} \frac{4(1+K^2)}{(T-t) c \hat q_{t,i}} q_{t,i} \1[\hat q_{t,i}>0]\right ]\le \frac{16(1+K^2)}{c} \sum_{t\in \mE_\ell} \frac{1}{T-t}.
\end{equation*}
For the second term, using an intermediate result proved in \Cref{claim:audits in estimation phase}, namely \Cref{eq:useful_inequality}, we have
\begin{equation*}
\sum_{i\in \mA_{s_\ell}} \left (1+\sum_{t\in \mE_\ell}\1[i_t=i]\cdot \bigvee_{j\in \mA_{s_\ell} \setminus \{i\}} D_{t,j} \right )\le K+\sum_{t\in \mE_\ell} \sum_{i\in \mA_{s_\ell}} D_{t,i}.
\end{equation*}
Therefore, summing over all epochs $\ell\in[K]$ and agents $i\in\mA_{s_\ell}$ gives
\begin{align*}
&\quad \E\Bigg [\sum_{t=1}^T o_t \1[\hat q_{t,i_t}>0]\Bigg ]=\E\Bigg [\sum_{\ell=1}^K \sum_{t\in \mE_{\ell,i}^I} o_t\Bigg ]\\
&\le \frac{32K^3}{c} \E\sqb{\sum_{t=1}^T \frac{1}{T-t} } + \E\sqb{\sum_{\ell=1}^K \left (K+\sum_{t\in \mE_\ell} \sum_{i\in \mA_{s_\ell}} D_{t,i}\right )} =\O\left (\frac{K^3}{c}\log T\right )+K^2+\E\left [\sum_{i=1}^K \sum_{t=1}^T D_{t,i}\right ]\\
&\le \O\left (K^2 + \frac{K^3}{c}\log T\right )+\frac{2 K^2}{c} \E_{\pi^\ast}\left [\sum_{t=1}^T \1[(i_t\text{ not eliminated})\wedge (v_{t,i_t}>u_{t,i_t}+R)\wedge o_t]\right ],
\end{align*}
where the last inequality again uses \Cref{lem:number of under-reporting bound formal}.
Summing up the audit bounds for the estimation phase and the incentive-compatible phase gives the claimed bound on $\mB_T(\pi^\ast,\mechname)$.
\end{proof}

\begin{proof}[Proof of \Cref{claim:end of estimation phase}]
For notational simplicity, let
\begin{equation}\label{eq:M and N}
M_{\ell,i}:=1+\sum_{t\in \mE_\ell} \1[i_t=i]\cdot \bigvee_{j\in \mA_{s_\ell}\setminus \{i\}} D_{t,j},\quad N_{\ell,i}:=1+\sum_{t\in \mE_\ell} D_{t,i}.
\end{equation}
Then the definition of $\mathfrak t_{\ell,i}$ is such that agent $i$ has already won for at least $\max(4M_{\ell,i},3N_{\ell,i})$ times during $\mE_\ell$ and that \Cref{eq:constraint_F} holds.
From \Cref{line:estimate winning prob} of \Cref{alg:mechanism}, the estimation phase for agent $i$ ends in some round $t \in \mE_\ell$ once no agent flags the estimated $\hat q_{t,i}$. Recalling the well-behaved flagging policy ${\pi}^{\ast,f}$ from \Cref{def:well-behaved flagging}, this corresponds to the first $t\in \mE_\ell$ such that
\begin{equation}\label{eq:condition of ending estimation}
\hat q_{t,i} = \frac{ \abs{\set{ \tau\in[s_\ell,t]\mid i_{\tau}=i }} }{t-s_\ell+1} \in \left [ \frac{q_{t,i}}{4}, 4 q_{t,i}\right ]=\left [ \frac{q_{s_\ell,i}}{4}, 4 q_{s_\ell,i}\right ].
\end{equation}

We now show this happens before $\mathfrak t_{\ell,i}$.
The claim is immediate if ${\mathfrak t}_{\ell,i}=s_{\ell+1}-1$.
Otherwise, we show ${\mathfrak t}_{\ell,i}$ satisfies \Cref{eq:condition of ending estimation}: Since agent $i$ won at least $4M_{\ell,i}$ times by round $\mathfrak t_{\ell,i}$, we have
\begin{equation*}
\sum_{\tau=s_\ell}^{{\mathfrak t}_{\ell,i}} F_{t,i}\geq \abs{\set{ \tau\in[s_\ell,{\mathfrak t}_{\ell,i}]\mid i_{\tau}=i }}-M_{i,j} \geq \frac{3\abs{\set{ \tau\in[s_\ell,{\mathfrak t}_{\ell,i}]\mid i_{\tau}=i }}}{4},
\end{equation*}
where we used \Cref{eq:upper_bound_nb_wins} in the first inequality. Further using the property of $\mathfrak t_{\ell,i}$ in \Cref{eq:constraint_F},
\begin{equation*}
\frac{ \abs{\set{ \tau\in[s_\ell,{\mathfrak t}_{\ell,i}]\mid i_{\tau}=i }} }{{\mathfrak t}_{\ell,i}-s_\ell+1} \leq \frac{4}{3({\mathfrak t}_{\ell,i}-s_\ell+1)} \sum_{\tau=s_\ell}^{{\mathfrak t}_{\ell,i}} F_{t,i} \leq 4q_{\ell,i}.
\end{equation*}

For the other direction, because agent $i$ won at least $3N_{\ell,i}$ times by round ${\mathfrak t}_{\ell,i}$, we have
\begin{equation*}
    \sum_{\tau=s_\ell}^{{\mathfrak t}_{\ell,i}} F_{t,i} \leq \abs{\set{ \tau\in[s_\ell,{\mathfrak t}_{\ell,i}]\mid i_{\tau}=i }} + N_{\ell,i}\leq \frac{4}{3}\abs{\set{ \tau\in[s_\ell,{\mathfrak t}_{\ell,i}]\mid i_{\tau}=i }},
\end{equation*}
where we used \Cref{eq:lower_bound_nb_wins} in the first inequality.
Again using the property of $\mathfrak t_{\ell,i}$ in \Cref{eq:constraint_F},
\begin{equation*}
    \frac{ \abs{\set{ \tau\in[s_\ell,{\mathfrak t}_{\ell,i}]\mid i_{\tau}=i }} }{{\mathfrak t}_{\ell,i}-s_\ell+1} \geq \frac{3}{4({\mathfrak t}_{\ell,i}-s_\ell+1)} \sum_{\tau=s_\ell}^{{\mathfrak t}_{\ell,i}} F_{t,i} \geq \frac{q_{\ell,i}}{4}.
\end{equation*}
Hence $\mathfrak t_{\ell,i}$ satisfies \Cref{eq:condition of ending estimation}. Thus the estimation phase $\mE_{\ell,i}^E$ ends at or before $\mathfrak t_{\ell,i}$.
\end{proof}

\begin{proof}[Proof of \Cref{claim:audits in estimation phase}]
We now define $\tilde {\mathfrak t}_{\ell,i}$ identical to $\mathfrak t_{\ell,i}$ except for removing \Cref{eq:constraint_F}, i.e., $\tilde {\mathfrak t}_{\ell,i}$ is the first round $t\in \mE_\ell$ when agent $i$ wins for the $\max(4M_{\ell,i}, 3N_{\ell,i})$-th time, where recall from \Cref{eq:M and N}
\begin{equation*}
M_{\ell,i}=1+\sum_{t\in \mE_\ell} \1[i_t=i]\cdot \bigvee_{j\in \mA_{s_\ell}\setminus \{i\}} D_{t,j},\quad N_{\ell,i}=1+\sum_{t\in \mE_\ell} D_{t,i}.
\end{equation*}
If no such $t\in \mE_\ell$ exists, we set $\tilde {\mathfrak t}_{\ell,i}:=s_{\ell+1}-1$ as usual.

If $\tilde {\mathfrak t}_{\ell,i}$ already satisfies \Cref{eq:constraint_F}, then $\mathfrak t_{\ell,i}=\tilde{\mathfrak t}_{\ell,i}$ and we obtain directly from \Cref{claim:end of estimation phase} that
\begin{equation}\label{eq:case_1}
\sum_{t\in \mE_{\ell,i}^E} o_t\leq \Big \lvert \Big \{ \tau\in[s_\ell,{\mathfrak t}_{\ell,i}]\mid i_{\tau}=i \Big \} \Big \rvert=\Big \lvert \Big \{ \tau\in[s_\ell,\tilde {\mathfrak t}_{\ell,i}]\mid i_{\tau}=i \Big \} \Big \rvert = \max(4M_{\ell,i}, 3N_{\ell,i}).
\end{equation}

Otherwise, we have $\tilde{\mathfrak t}_{\ell,i}<{\mathfrak t}_{\ell,i}$. Using the upper bound on the number of wins derived in \Cref{eq:upper_bound_nb_wins},
\begin{align}\label{eq:case_2a}
\sum_{t\in \mE_{\ell,i}^E} o_t\leq \Big \lvert \Big \{ \tau\in[s_\ell,{\mathfrak t}_{\ell,i}]\mid i_{\tau}=i \Big \} \Big \rvert\overset{\text{\Cref{eq:upper_bound_nb_wins}}}\le M_{\ell,i} +\sum_{\tau=s_\ell}^{{\mathfrak t}_{\ell,i}} F_{t,i}  \overset{\text{\Cref{eq:constraint_F}}}\leq   M_{\ell,i} + 3q_{\ell,i}(\mathfrak t_{\ell,i}-s_\ell+1).
\end{align}

We analyze the second term on the RHS of \Cref{eq:case_2a}. Let $k_0$ be the integer such that $2^{k_0}\leq \mathfrak t_{\ell,i}-s_\ell+1 <2^{k_0+1}$ (which is unique). Since $\tilde {\mathfrak t}_{\ell,i}$ does not satisfy \Cref{eq:constraint_F}, it must violate either the lower or the upper bound. If $\tilde{\mathfrak t}_{\ell,i}$ violates the upper bound of \Cref{eq:constraint_F}, then (supposing that $s_\ell+2^{k_0+1}-1\in \mE_l$)
\begin{equation}\label{eq:upper_bound_power_2}
    \frac{1}{ 2^{k_0+1}}\sum_{\tau=s_\ell}^{s_\ell + 2^{k_0+1}-1} F_{t,i} \geq \frac{1}{2(\tilde{\mathfrak t}_{\ell,i} -s_\ell + 1)}\sum_{\tau=s_\ell}^{\tilde{\mathfrak t}_{\ell,i}} F_{t,i} > \frac{3}{2}q_{\ell,i}.
\end{equation}
Otherwise, if $\tilde{\mathfrak t}_{\ell,i}$ violates the lower bound of \Cref{eq:constraint_F}, then
\begin{equation}\label{eq:lower_bound_power_2}
    \frac{1}{ 2^{k_0}}\sum_{\tau=s_\ell}^{s_\ell + 2^{k_0}-1} F_{t,i} < \frac{2}{\tilde{\mathfrak t}_{\ell,i} -s_\ell + 1}\sum_{\tau=s_\ell}^{\tilde{\mathfrak t}_{\ell,i}} F_{t,i} < \frac{2}{3}q_{\ell,i}.
\end{equation}

Now define $k_{\ell,i}$ as the smallest integer such that for all $k\geq k_{\ell,i}$, if $s_\ell+2^k-1\in \mE_\ell$ then
\begin{equation}\label{eq:def_k_il}
    \frac{1}{2^k}\sum_{\tau=s_\ell}^{s_\ell+2^k-1} F_{t,i} \in \sqb{\frac{2}{3}q_{\ell,i}, \frac{3}{2}q_{\ell,i} }.
\end{equation}
\Cref{eq:upper_bound_power_2,eq:lower_bound_power_2} show that $k_0$ violates \Cref{eq:def_k_il} and thus $k_{\ell,i}>k_0$. Hence, furthering \Cref{eq:case_2a} gives $\sum_{t\in \mE_{\ell,i}^E} o_t \leq M_{\ell,i} + 6 q_{\ell,i} 2^{k_{\ell,i}}$.
Together with \Cref{eq:case_1} this shows that in all cases one has
\begin{equation}\label{eq:full_case}
\sum_{t\in \mE_{\ell,i}^E} o_t \leq 1 + 4M_{\ell,i} + 3N_{\ell,i} + 6 q_{\ell,i} 2^{k_{\ell,i}}.
\end{equation}

It only remains to bound $2^{k_{\ell,i}}$, or its expectation $\E[2^{k_{\ell,i}}\mid\mH_{s_\ell}]$.
Note that conditionally on the history $\mH_{s_\ell}$ at the beginning of epoch $\mE_\ell$, the sequence of $F_{t,i}$ is an independent sequence of Bernoulli $\text{Ber}(q_{\ell,i})$. Hence, for any $k\geq 3+\log_2(1/q_{\ell,i})$, Chernoff's bound gives
\begin{equation*}
    \Pr\{k_{\ell,i} \geq k\} \leq \sum_{k'\geq k}\Pr\left \{\frac{1}{2^{k'}}\sum_{t=s_\ell}^{s_\ell+2^{k'}-1} F_{t,i} \notin \sqb{\frac{2}{3}q_{\ell,i},\frac{3}{2}q_{\ell,i}}\right \} \leq 2\sum_{k'\geq k}e^{-2^{k'}q_{\ell,i}/18} \leq 6e^{-2^{k}q_{\ell,i}/18}.
\end{equation*}

As a result,
\begin{equation*}
    \E[2^{k_{\ell,i}}\mid\mH_{s_\ell}] \leq \frac{8}{q_{\ell,i}} + \sum_{k\geq 3+\log_2(1/q_{\ell,i})}2^k \Pr \{k_{\ell,i} = k\} \leq \frac{98}{q_{\ell,i}}.
\end{equation*}

Plugging this, as well as the definition of and $N_{\ell,i}$ in \Cref{eq:M and N}, into \Cref{eq:full_case}, we get
\begin{equation}\label{eq:full_case 2}
\E\Bigg [\sum_{t\in \mE_{\ell,i}^E} o_t\mathrel{\Bigg \vert} \mH_{s_{\ell}}\Bigg ]\le \E\left [1+4 M_{\ell,i}+3\left (1+\sum_{t\in \mE_\ell} D_{t,i}\right )\middle \vert \mH_{s_{\ell}}\right ]+588.
\end{equation}
Further observing that, by the definition of $M_{\ell,i}$ in \Cref{eq:M and N}, we have
\begin{align}\label{eq:useful_inequality}
\sum_{i\in \mA_{s_\ell}} M_{\ell,i} &= \sum_{i\in \mA_{s_\ell}} \left (1+\sum_{t\in \mE_\ell} \1[i_t=i]\cdot \bigvee_{j\in \mA_{s_\ell}\setminus \{i\}} D_{t,j}\right )\\
&\leq K + \sum_{t\in \mE_\ell} \max_{j\in\mA_{s_\ell}} D_{t,j}\paren{ \sum_{i\in \mA_{s_\ell}} \1[i_{\tau}=i] } \\
&\leq K+\sum_{t\in \mE_\ell} \sum_{i\in \mA_{s_\ell}} D_{t,i}.
\end{align}
Therefore, summing up \Cref{eq:full_case 2} for all $i\in \mA_{s_\ell}$ gives
\begin{equation*}
\E\sqb{\sum_{i\in\mA_{s_\ell}} \sum_{t\in \mE_{\ell,i}^E} o_t \middle \vert \mH_{s_\ell}} \leq 589 K + 7 \E\sqb{\sum_{i\in \mA_{s_\ell}}\sum_{t\in \mE_\ell} D_{t,i} \middle \vert \mH_{s_\ell}}.
\end{equation*}
This completes the proof.
\end{proof}

\section{Main Theorems under Different Audit Models}\label{sec:appendix deterministic main theorem}
After establishing the general reduction-based analysis framework in \Cref{sec:appendix varying-p,sec:appendix regret and audits}, we plug in the various audit models in \Cref{assump:noiseless,assump:adv noise model,assump:stoc noise model}. This finishes the proof of \Cref{thm:deterministic main thm,thm:adv noise model,thm:stoc noise model}.

\subsection{Perfect Audit Model (Proof of \Cref{thm:deterministic main thm})}\label{sec:equiv between actual and no-flagging formal}
\begin{lemma}[Correspondence between Auxiliary and Original Games]\label{lem:equiv between actual and no-flagging formal}
Under \Cref{assump:noiseless}, consider the auxiliary game with tolerance $0$ (which reduces to \Cref{def:auxiliary game}). For any PBE $\pi^{\ast,r}$ therein, let $\pi^\ast=(\pi^{\ast,r},\pi^{\ast,f})$ is the corresponding original-game strategy.
Then $\Delta(\pi^{\ast,r})=0$, i.e., $V_i^{\pi^\ast}(\mH_1)=\tilde V_i^{\pi^{\ast,r}}(\tilde \mH_1)$.
\end{lemma}
\begin{proof}
We prove the following claim by backward induction: for any round $t\in [T]$ and any history $\mH_t\in H_t$ such that $\hat q_{t,j}\le 4q_{t,j}$ for all $j\in \mA_t$, we have
\begin{equation}\label{eq:equiv between actual and no-flagging}
V_i^{\pi^\ast}(\mH_t)=\tilde V_i^{\pi^{\ast,r}}(\tilde \mH_t),\quad \forall i\in \mA_t.
\end{equation}
We first justify that backward induction is applicable, namely any possible next-round history $\mH_{t+1}$ resulted from $\mH_t$ also ensures $\hat q_{t+1,j}\le 4q_{t+1,j}$ for all $j\in \mA_{t+1}$:
\begin{itemize}
\item If someone is eliminated in round $t$, then $\hat q_{t+1,j}=0$ for all $j\in \mA_{t+1}$, which is no more than $4q_{t+1,j}$.
\item If no new $\hat q_{t+1,i_t}$ is generated or the new $\hat q_{t+1,i_t}\le 4q_{t+1,i_t}$, then the property holds naturally.
\item If a new $\hat q_{t+1,i_t}$ is generated and $\hat q_{t+1,i_t}>4q_{t+1,i_t}$, then ${\pi}_{t,-i_t}^{\ast,f}$ would flag them. In this case $\hat q_{t+1,i_t}$ is replaced with 0 and the property also holds.
\end{itemize}

Let \Cref{eq:equiv between actual and no-flagging} hold for all $\mH_{t+1}$'s such that $\hat q_{t+1,j}\le 4q_{t+1,j}$, $\forall j$. We now prove it for a fixed $\mH_t\in H_t$ such that $\hat q_{t,j}\le 4q_{t,j}$, $\forall j$. By defintion of $V_i$ and $\tilde V_i$, we write
\begin{align*}
V_i^{\pi^\ast}(\mH_t)&=\E_{\pi^\ast}\Big [u_{t,i}\1[i_t=i]\mid \mH_t\Big ]+\E_{\pi^\ast}\Big [V_i^{\pi^\ast}(\mH_{t+1})\mid \mH_t\Big ],\\\tilde V_i^{\pi^{\ast,r}}(\tilde \mH_t)&=\E_{\pi^{\ast,r}}\Big [u_{t,i}\1[i_t=i]\mid \tilde \mH_t\Big ]+\E_{\pi^{\ast,r}}\Big [V_i^{\pi^\ast}(\tilde \mH_{t+1})\mid \tilde \mH_t\Big ].
\end{align*}
Since the allocation rule is the same between \mechname and the auxiliary game, and $\pi^\ast=(\pi^{\ast,r},\pi^{\ast,f})$, the first terms on the RHS coincide. It only remains to construct a coupling between $\mH_{t+1}$ (induced by $\pi^\ast$ from $\mH_t$) and $\tilde\mH_{t+1}$ (induced by $\pi^{\ast,r}$ from $\tilde \mH_t$).
As the report policies and allocation mechanism are identical for both games, we natually couple the reports $\bm v_t$ and winning agent $i_t$ to be the same for both games.

Recall the $\hat p_{t,j}$ defined in \Cref{line:check prob} of \Cref{alg:mechanism} and the $p_{t,j}$ defined in \Cref{def:auxiliary game formal}, we have
\begin{equation*}
\hat p_{t,j}=\min\left (\frac{4(1+K^2)}{(T-t) \hat q_{t,j} c},1\right )\overset{(a)}{\ge} \min\left (\frac{1+K^2}{(T-t) q_{t,j} c},1\right ) = \min(p_{t,j},1),\quad \forall j\in [K],
\end{equation*}
where (a) used the property that $\hat q_{t,j}\le 4q_{t,j}$ for all $j\in [K]$.
We discuss whether the winner $i_t$ (the same in both games) over-reported, i.e., whether $v_{t,i_t}>u_{t,i_t}$.
\begin{itemize}
\item If so, they are eliminated immediately in the auxiliary game according to \Cref{item:eliminate immediately}. Thus, the next-round simplified history generated by ${\pi}^{r,\ast}$ in the auxiliary game is $\tilde \mH_{t+1}=(t+1,\mA_t\setminus \{i_t\})$.
For the original game, because $i_t$ marked up using $\pi^{\ast,r}_{t,i_t}$, \Cref{item:no over-report unless p>=1 formal} implies that $\hat p_{t,i_t}\ge \min(p_{t,i_t},1)=1$, which means $i_t$ is audited w.p. 1 in the original game. This implies $i_t$ is eliminated a.s., and thus $\mA_{t+1}=\mA_t\setminus \{i_t\}$. The original game history $\mH_{t+1}$ hence corresponds to the same simplified history $\tilde \mH_{t+1}$ a.s.
\item Otherwise, $i_t$ does not over-report. Thus, the next-round simplified history in the auxiliary game is $\tilde\mH_{t+1}=(t+1,\mA_t)$. Meanwhile, regardless of whether they are audited in the original game, $i_t$ stays alive because $v_{t,i_t}\le u_{t,i_t}+R_{t}$. Thus the next-round histories are again identical.
\end{itemize}

In summary, there is a coupling between games, in which the $\mH_{t+1}$ generated by the original game always matches the $\tilde \mH_{t+1}$ generated by the auxiliary game.
Hence, our claim---namely \Cref{eq:equiv between actual and no-flagging} holds for all $\mH_t$ resulted from $\pi$---follows from the backward induction on $t=T,T-1,\ldots,1$.
\end{proof}

In fact, \Cref{eq:equiv between actual and no-flagging} holds not only to the auxiliary-game PBE $\bm \pi^{\ast,r}$ but also \emph{any} auxiliary-game strategy $\pi^r$ (in the proof, we didn't use any PBE property) and therefore recovers \Cref{lem:equiv between actual and no-flagging} in the main text. Nevertheless, according to \Cref{thm:equivalence of PBE formal}, due to the performance difference lemma, we only require \Cref{eq:equiv between actual and no-flagging} for the PBE strategy $\pi^{\ast,r}$. This is critical for the stochastic audit model that we will analyze in \Cref{sec:stoc noise model appendix}.

\begin{proof}[Proof of \Cref{thm:deterministic main thm}]
Under the perfect audit model (\Cref{assump:noiseless}), let $\pi^{\ast,r}$ be a PBE in the auxiliary game with $R=0$ and let $\pi^\ast=(\pi^{\ast,r},\pi^{\ast,f})$ (where $\pi^{\ast,f}$ is the well-behaved flagging strategy from \Cref{def:well-behaved flagging}).
From \Cref{lem:equiv between actual and no-flagging formal}, $\Delta(\pi^{\ast,r})=0$. Therefore, \Cref{thm:equivalence of PBE formal} reveals $\pi^\ast$ is a 0-ABE---hence a PBE---in the original game. \Cref{thm:regret and audits formal} further characterizes the regret and number of audits of $\pi^\ast$ as
\begin{align*}
\mR_T(\pi^\ast,\mechname)&\le K^2+\sum_{i=1}^K\left (\tilde V_i^{\pi^{\ast,r}}(\tilde \mH_1)-V_i^{\pi^\ast}(\mH_1)\right )=K^2,\\
\mB_T(\pi^\ast,\mechname)&=\O\left (\frac{K^3}{c}\log T\right )+\frac{16 K^2}{c} \E_{\pi^\ast}\left [\sum_{t=1}^T \1[(i_t\text{ not eliminated})\wedge (v_{t,i_t}>u_{t,i_t}+R)\wedge o_t]\right ]\\
&=\O\left (\frac{K^3}{c} \log T\right ).
\end{align*}
The first equality uses \Cref{lem:equiv between actual and no-flagging formal}, i.e., $V_i^{\pi^\ast}(\mH_1)=\tilde V_i^{\pi^{\ast,r}}(\tilde \mH_1)$ for all $i\in [K]$, and the second is because under the perfect audit model, audit outcome $w_t=u_{t,i_t}$. Hence $o_t=1$ (audit requested) and $v_{t,i_t}>u_{t,i_t}$ (winner over-reported) imply the elimination of $i_t$ (\Cref{line:elimination} of \Cref{alg:mechanism}). This completes the proof.
\end{proof}

\subsection{Adversarial Audit Model (Proof of \Cref{thm:adv noise model})}\label{sec:adv noise model appendix}

\begin{proof}[Proof of \Cref{thm:adv noise model}]
Under the adversarial audit model (\Cref{assump:adv noise model}), let $\pi^{\ast,r}$ be a PBE in the auxiliary game with $R=\sigma$ and let $\pi^\ast=(\pi^{\ast,r},\pi^{\ast,f})$ (where $\pi^{\ast,f}$ is the well-behaved flagging strategy from \Cref{def:well-behaved flagging}). We claim that $\Delta(\pi^{\ast,r})=0$: the proof is almost identical to that of \Cref{lem:equiv between actual and no-flagging formal}, except that we now consider whether the winner $i_t$ over-reported \emph{by more than $\sigma$}, i.e., whether $v_{t,i_t}>u_{t,i_t}+\sigma$.
\begin{itemize}
\item If so, they are eliminated in the auxiliary game with tolerance $R:=\sigma$. In the original game, \Cref{item:no over-report unless p>=1 formal} implies that $\hat p_{t,i_t}\ge \min(p_{t,i_t},1)=1$, and thus they are also eliminated in the original game.
\item Otherwise, the over-report is within the tolerance of the auxiliary game, hence $i_t$ remains alive therein. In the original game, agents always prefer staying alive to being eliminated. Therefore, due to the power of manipulating audit outcomes within the radius of $\sigma$ (recall \Cref{assump:adv noise model}), $i_t$ can ensure $w_t=u_{t,i_t}$ (since $\lvert u_{t,i_t}-v_{t,i_t}\rvert\le \sigma$) and thus remains alive. Therefore, in both games agent $i_t$ remains alive.
\end{itemize}

This proves that both games still exactly correspond to the other, i.e., \Cref{eq:equiv between actual and no-flagging} holds (but now with tolerance parameter $R=\sigma$) and therefore $\Delta(\pi^{\ast,r})=0$. Same as \Cref{sec:equiv between actual and no-flagging formal}, \Cref{thm:equivalence of PBE formal,thm:regret and audits formal} then dictate a PBE $\pi^\ast$ such that $\mR(\pi^\ast,\mechname)\le K^2$ and $\mB(\pi^\ast,\mechname)=\O(\frac{K^3}{c} \log T)$.
\end{proof}

\subsection{Stochastic Audit Model (Proof of \Cref{thm:stoc noise model})}\label{sec:stoc noise model appendix}
\begin{proof}[Proof of \Cref{thm:stoc noise model}]
If $\epsilon\ge \frac{c}{32K^2}$, which means $\epsilon \frac{K^3}{c} T \log T\ge \frac{K}{32} T\log T$, the conclusion is trivial as any strategy is a $T$-ABE and ensures $\mR_T\le T$. Therefore, we assume that $\epsilon<\frac{c}{32K^2}$, which gives $1-\epsilon \frac{16K^2}{c}\ge \frac 12$.

Under the stochastic audit model (\Cref{assump:stoc noise model}), let $\pi^{\ast,r}$ be a PBE in the auxiliary game with $R=0$ and let $\pi^\ast=(\pi^{\ast,r},\pi^{\ast,f})$ (where $\pi^{\ast,f}$ is the well-behaved flagging strategy from \Cref{def:well-behaved flagging}). Instead of working out $\Delta(\pi^{\ast,r})$ directly as in \Cref{sec:equiv between actual and no-flagging formal,sec:adv noise model appendix}, we first invoke \Cref{thm:equivalence of PBE formal,thm:regret and audits formal}. They together suggest that $\pi^\ast$ is a $\Delta(\pi^{\ast,r})$-ABE such that
\begin{align*}
\mR_T(\pi^\ast,\mechname)&\le K^2+\sum_{i=1}^K\left (\tilde V_i^{\pi^{\ast,r}}(\tilde \mH_1)-V_i^{\pi^\ast}(\mH_1)\right ),\\
\mB_T(\pi^\ast,\mechname)&=\O\left (\frac{K^3}{c}\log T\right )+\frac{16 K^2}{c} \E_{\pi^\ast}\left [\sum_{t=1}^T \1[(i_t\text{ not eliminated})\wedge (v_{t,i_t}>u_{t,i_t}+R)\wedge o_t]\right ]
\end{align*}

For the second term on the RHS of $\mB_T$, we have
\begin{align*}
&\quad \E_{\pi^\ast}\left [\sum_{t=1}^T \1[(i_t\text{ not eliminated})\wedge (v_{t,i_t}>u_{t,i_t}+R)\wedge o_t]\right ]\\
&=\sum_{t=1}^T \E_{\mH_t\sim \pi^\ast}\left [\1[v_{t,i_t}>u_{t,i_t}+R]o_t \Pr\{i_t\text{ not eliminated}\mid (v_{t,i_t}>u_{t,i_t}+R)\wedge o_t\}\right ]\\
&\overset{(a)}{\le} \epsilon \sum_{t=1}^T \E_{\mH\sim \pi^\ast} \Big [\1[v_{t,i_t}>u_{t,i_t}+R]o_t\Big ]\overset{(b)}{\le} \epsilon \mB_T(\pi^\ast,\mechname),
\end{align*}
where (a) uses \Cref{assump:stoc noise model} and (b) uses the definition of $\mB_T$. This gives a self-bounding form of $\mB_T$:
\begin{align*}
\mB_T(\pi^\ast,\mechname)&\le \O\left (\frac{K^3}{c}\log T\right )+\epsilon \frac{16K^2}{c}\mB_T(\pi^\ast,\mechname)\\ \Longrightarrow \mB_T(\pi^\ast,\mechname)&=\O\left (\frac{K^3}{c}\log T \middle / \left (1-\epsilon \frac{16K^2}{c}\right )\right )=\O\left (\frac{K^3}{c}\log T\right ),
\end{align*}
where we used $1-\epsilon \frac{16K^2}{c}\ge \frac 12$.
We now utilize this bound to find out $\Delta(\pi^{\ast,r})$ and $\sum_{i=1}^K (\tilde V_i^{\pi^{\ast,r}}(\tilde \mH_1)-V_i^{\pi^\ast}(\mH_t))$.
Specifically, for any $(r_\tau\colon \mH_\tau\to [0,1])_{\tau\in [T]}$ that is well-defined in the auxiliary game, and any history $\mH_t$ such that $\hat q_{t,j}\le 4q_{t,j}$ for all $j\in [K]$, we compare the two trajectories starting from $\mH_t$ in the original game (under \mechname) and that from $\tilde \mH_t$ in the auxiliary game (under \Cref{def:auxiliary game formal}). Viewing the transitions of $\mH_t$ and that of $\tilde \mH_t$ as two Markov Decision Processes (MDP), we slightly abuse the notations and define
\begin{equation*}
V^{\pi^\ast}_r(\mH_t):=\E_{\pi^\ast,\mechname}\left [\sum_{\tau\ge t} r_\tau(\mH_\tau)\middle \vert \mH_t\right ],\quad \tilde V^{\pi^{\ast,r}}_r(\tilde \mH_t):=\E_{\pi^{\ast,r},\text{\Cref{def:auxiliary game formal}}}\left [\sum_{\tau\ge t} r_\tau(\tilde \mH_\tau)\middle \vert \tilde \mH_t\right ].
\end{equation*}
Invoking the simulation lemma \citep[see][Exercise 2.1 for the exact form we use below]{Sun},
\begin{align}\label{eq:simulation lemma}
&\quad \left \lvert \E_{\pi^\ast,\mechname}\left [\sum_{\tau\ge t} r_\tau(\mH_\tau)\middle \vert \mH_t\right ]-\E_{\pi^{\ast,r},\text{\Cref{def:auxiliary game formal}}}\left [\sum_{\tau\ge t} r_\tau(\tilde \mH_\tau)\middle \vert \tilde \mH_t\right ] \right \rvert=\lvert V_r^{\pi^\ast}(\mH_t)-\tilde V_r^{\pi^{\ast,r}}(\tilde \mH_t)\rvert\\
&\overset{(a)}{=}\sum_{\tau\ge t} \E_{\mH_\tau\sim \mechname} \left [\E_{\mH_{\tau+1}\sim \mechname} [\tilde V_r^{\pi^{\ast,r}}(\tilde \mH_{\tau+1})]-\E_{\tilde \mH_{\tau+1}\sim \text{\Cref{def:auxiliary game formal}}} [\tilde V_r^{\pi^{\ast,r}}(\tilde \mH_{\tau+1})]\right ]\\
&\le \sum_{\tau\ge t} \E_{\mH_\tau\sim \mechname} \left [\Pr_{\mH_{\tau+1}\sim \mechname,\tilde \mH_{\tau+1}\sim \text{\Cref{def:auxiliary game formal}}} \left \{\mH_{\tau+1}\text{ and }\tilde \mH_{\tau+1}\text{ mismatch}\right \}\middle \vert \mH_\tau\right ]\cdot (T-\tau)\\
&\le \mB_T(\pi^\ast,\mechname)\cdot \epsilon \cdot (T-\tau)=\O\left (\epsilon \frac{K^3}{c} T \log T\right ),\quad \forall t\in [T],
\end{align}
where we remark that in (a), the $\tilde \mH_{\tau+1}$ in the first expectation means the simplified history corresponding to $\mH_{\tau+1}$ (induced by \mechname from full history $\mH_\tau$), while the $\tilde \mH_{t+1}$ in the second expectation stands for the simplified history induced by \Cref{def:auxiliary game formal} from $\tilde \mH_\tau$ (the simplified history corresponding to $\mH_\tau$).

To control $\Delta(\pi^{\ast,r})$, for any fixed $i\in [K]$, we pick $r_t(\mH_t):=u_{t,i}\1[i_t=i]$ (well-defined in the auxiliary game due to the identical allocation rules). To control $\sum_{i=1}^K (\tilde V_i^{\pi^{\ast,r}}(\tilde \mH_1)-V_i^{\pi^\ast}(\mH_1))$, we pick $r_t(\mH_t):=u_{t,i_t}$ (again well-defined). Therefore, both are bounded by $\O(\epsilon \frac{K^3}{c} T\log T)$ according to \Cref{eq:simulation lemma}. To summarize,
\begin{align*}
&\pi^\ast\text{ is an }\O\left (\epsilon \frac{K^3}{c} T \log T\right )\text{-ABE in the original game,}\\
&\mR_T(\pi^\ast,\mechname)=\O\left (K^2+\epsilon \frac{K^3}{c} T \log T\right ),\\
&\mB_T(\pi^\ast,\mechname)=\O\left (\frac{K^3}{c} \log T\right ),
\end{align*}
as claimed in \Cref{thm:stoc noise model}.
\end{proof}

\section{Lower Bounds on Regret and Number of Audits (\Cref{thm:lower_bounds})}\label{sec:appendix lower bounds}

In this section, we prove the hardness results in \Cref{thm:lower_bounds}. We start by proving the lower bound $\Omega(K)$ on the regret, then prove the lower bound $\Omega(1)$ on the number of audits for having low regret.

\begin{theorem}[Lower Bound on Regret]\label{thm:lower_bound_regret}
Let $K\geq 2$ and $T\geq K$. Then, there exists $K$ utility distributions $\mV_1,\mV_2,\ldots,\mV_K$ over $[0,1]$ satisfying \Cref{assumption:min report} with $c=\frac 13$, such that for any central planner mechanism $\bm M$ and corresponding agents' joint strategy PBE $\pi=(\pi_{t,i})_{t\in [T],i\in[K]}$,
\begin{equation*}
\mR_T(\pi,\bm M) \geq \frac{K-1}{64}=\Omega(K).
\end{equation*}
\end{theorem}
\begin{proof}
    We consider the following distributions: $\mV_1 = \delta_{2/3}$, that is, $u_1=\frac{2}{3}$ almost surely, and for all $i\geq 2$, we have $\mV_i \sim \frac{1}{3} + \frac{2}{3}\text{Ber}(\frac{\alpha}{T})$ where $\alpha=\frac{1}{8}$, that is, $u_i=1$ with probability $\alpha/T$ and $u_i=\frac{1}{3}$ otherwise.
    Now fix any PBE $\pi$ corresponding to these utility distributions. We fix an agent $i\geq 2$ and aim to show that
    \begin{equation}\label{eq:bound_error_i}
        R_i:=\E_{\pi}\sqb{\sum_{t=1}^T \1[u_{t,i}=1]\1[i_t\neq i] + \1[u_{t,i}=1/3]\1[i_t=i]} \geq \frac{\alpha}{2}.
    \end{equation}
    
    Suppose by contradiction that \Cref{eq:bound_error_i} does not hold. We then have
    \begin{equation}\label{eq:upper_bound_gain_i}
        V_i^{\pi}(\mH_1) \leq \E\sqb{\sum_{t=1}^T \frac{1}{3}\1[u_{t,i}=1/3]\1[i_t=i] + \1[u_{t,i}=1]} \leq \alpha+\frac{R_i}{3}\leq \frac{7\alpha}{6}.
    \end{equation}
    On the other hand, we can also bound the expected gain of agent $i$ conditional on the event $\mE_i\colon \exists t\in[T]\text{ s.t. } u_{t,i}=1$ which excludes the case when agent $i$ always has utility $1/3$.
    \begin{align*}
        \E\sqb{\paren{\sum_{t=1}^T u_{t,i}\1[i_t=i]}\1[\mE_i]} &\geq \E\sqb{\paren{\sum_{t=1}^T \1[u_{t,i}=1]}\1[\mE_i]}- R_i\\
        &= \E\sqb{\sum_{t=1}^T \1[u_{t,i}=1]}- R_i =\alpha-R_i \geq \frac{\alpha}{2}.
    \end{align*}
    As a result,
    \begin{equation*}
        \E\sqb{\sum_{t=1}^T u_{t,i}\1[i_t=i] \middle \vert \mE_i} \geq \frac{\alpha}{2\Pr\{\mE_i\}} = \frac{\alpha}{2(1-(1-\alpha/T)^T)}\geq \frac{1}{2}.
    \end{equation*}
    We now define a new strategy $\tilde{\pi}_i$ for agent $i$ which proceeds as follows: sample a "fake" sequence $\tilde{\bm u}_i=(\tilde u_{t,i})_{t\in[T]}$ from the distribution $\mV_i^{\otimes T}\mid\mE_i$, then run strategy $\pi_i$ throughout the whole allocation process by bidding as if the true utility of agent $i$ as $\tilde{\bm u}$. Then,
    \begin{align*}
        V_i^{(\pi_{-i},\tilde\pi_i)}(\mH_1) &\geq \frac{1}{3}\Pr_{(\pi_{-i},\tilde\pi_i)}\{\exists t\in[T]\text{ s.t. }i_t=i\} \geq \frac{1}{3} \E_{\pi}\sqb{\sum_{t=1}^T u_{t,i}\1[i_t=i] \middle \vert \mE_i} \geq \frac{1}{6} > V_i^{\pi}(\mH_1),
    \end{align*}
    where in the last inequality we used \Cref{eq:upper_bound_gain_i} and the fact that $\alpha<1/7$.
    This contradicts the fact that $\pi$ is a PBE. In summary, we proved that \Cref{eq:bound_error_i} holds for all $i\geq 2$.

    We can then bound the regret of the mechanism at the PBE $\pi$. First, note that any time $t$ for which $u_{t,i_t}=1/3$ induces a regret $1/3$ since agent $1$ always had utility $1$. Second, if at time $t$, there is a unique agent $i\geq 2$ which has utility $u_{t,i}=1$, the mechanism also incurs a regret at least $1/3$ if it does not allocate to agent $i$, that is if $i_t\neq i$. Formally,
    \begin{align*}
        \mR_T(\pi,\bm M)
        &\geq \E_{\bm u,\pi}\sqb{\sum_{t=1}^T \sum_{i=2}^K\frac{1}{3}\paren{\1\left [u_{t,i}=\frac 13\right ]\1[i_t=i]  + \1[u_{t,i}=1]\1[i_t\neq i]\1\left [u_{j,t}=\frac 13,\forall j\notin\{1,i\}\right ]} }\\
        &\geq \frac{1}{3}\sum_{i=2}^K R_i - \frac{1}{3}\E\sqb{\sum_{t=1}^T \paren{\sum_{i=2}^K\1[u_{t,i}=1]}\1\sqb{\sum_{i=2}^K\1[u_{t,i}=1] >1}}\\
        &\geq \frac{(K-1)\alpha}{6} - \frac{T}{3}\paren{\frac{K\alpha}{T}-K\frac{\alpha}{T}\paren{1-\frac{\alpha}{T}}^{K-1}}\geq \frac{(K-1)\alpha}{6}\paren{1-2\alpha},
    \end{align*}
    where in the last inequality, we used the fact that $1-(1-\alpha/T)^{K-1}\leq 1- (1-\alpha/T)^{T-1} \leq 1-e^{-\alpha}\leq \alpha$. Plugging in the utility $\alpha=1/8$ gives the desired result.
\end{proof}

\begin{theorem}[Lower Bound on Tradeoff Between Regret and Number of Audits]\label{thm:lower_bound_nb_checks}
Let $K\geq 2$ and $T\geq K$. Then, there exists universal constants $c_1,c_2>0$ and $K$ utility distributions $\mV_1,\mV_2,\ldots,\mV_K$ over $[0,1]$ satisfying \Cref{assumption:min report} with $c=\frac 13$, such that for any central planner mechanism $\bm M$ and corresponding agents' joint strategy PBE $\pi=(\pi_{t,i})_{t\in [T],i\in[K]}$,
\begin{equation*}
    \mB_T(\pi,\bm M)\leq c_1\Longrightarrow \mR_T(\pi,\bm M)\ge c_2\sqrt {\frac{T}{\log T}}.
\end{equation*} 
\end{theorem}
\begin{proof}
    We consider similar distributions as in the proof of \Cref{thm:lower_bound_regret}, but using only the first two agents. Precisely, we set $\mV_1=\delta_{2/3}$,  $\mV_2\sim\delta_{1/3} + \frac{2}{3}\text{Ber}(p)$, where $p=\frac{1}{3}$, and $\mV_i=\delta_{1/4}$ for $i\geq 3$. That is, $u_1=\frac{2}{3}$ always, $u_2$ equals $\frac{1}{3}$ with probability $1-p$ or $1$ with probability $p$, and $u_i=\frac{1}{4}$ for all $i\geq 3$ (so their ex-ante first-best utilities are all 0 as $u_i<u_1$ and $u_i<u_2$ a.s.). By the revelation principle, we assume without loss of generality that under this specific utility distributions setup $\{\mV_i\}_{i\in [K]}$ and mechanism $\bm M$ truthful reporting $\truth$ is the considered PBE. 

    By design, the central planner incurs a constant regret $\frac 13$ whenever they do not allocate the resource to agent $2$ if they had utility $1$. Also, because agent $1$ always has utility $\frac 23$, we have
    \begin{equation}\label{eq:lower_bound_regret}
        \mR_T(\truth,\bm M) = \frac{1}{3}\E\sqb{\sum_{t=1}^T \1[u_{t,2}=1]\1[i_t\neq 2] + \1\sqb{u_{t,2}=\frac{1}{3} }\1[i_t\neq 1]}.
    \end{equation}

    From now, we focus on the agent $i=2$. For simplicity, for any function $f\colon [0,1]^T\to [0,1]^T$, we abbreviate
    \begin{equation*}
    \E_{f(\bm u_i),\pi}[X]:=\E_{\text{i.i.d. }\bm u_i=(u_{t,i})_{t\in [T]}\text{ from }\mV_i} \left [\E_{\text{agents follow joint strategy }\pi}\left [X\middle \vert \text{agent $i$ has true utilities }f(\bm u_i)\right ]\right ]
    \end{equation*}
    We similarly define $\Pr_{f(\bm u_i),\pi}\{X\}$. For example, if agent $i$'s true utilities are i.i.d. samples from $\mV_i$ and all agents follow $\truth$, then the corresponding expectation operator is denoted by $\E_{\bm u_i,\truth}$.
    We start by upper bounding the gain of agent $i$ under truthful reporting as follows.
    \begin{align}
        V_i^{\truth}(\mH_1;\bm M) &\leq \E_{\bm u_i,\truth} \sqb{\sum_{t=1}^T \1[u_{t,i}=1] + \frac{1}{3} \1\sqb{u_{t,i}=\frac{1}{3}}\1[i_t=i] } \notag\\
        &= Tp + \frac{1}{3} \E_{\bm u_i,\truth} \sqb{\sum_{t=1}^T  \1\sqb{u_{t,i}=\frac{1}{3}}\1[i_t=i] }. \label{eq:upper_bound_truth}
    \end{align}
    
    We next construct an alternative strategy $\pi_i$ for agent $i$ as follows. Define a parameter $\epsilon:=(\sqrt{T\log T})^{-1}$. For $T$ sufficiently large, we have $p+\epsilon\leq \frac{1}{2}$. Let $\epsilon'$ be such that $p+(1-p)\epsilon'=p+\epsilon$.
    Then in round $t\in[T]$, after observing their utility $u_{t,i}$, $\pi_i$ suggests
    \begin{itemize}
    \item if $u_{t,i}=1$, truthfully reporting $v_{t,i}=1$ as in \truth,
    \item otherwise, independently w.p. $\epsilon'$ reporting $v_{t,i}=1$ (i.e., pretend that they had true utility $u_{t,i}=1$), and otherwise reporting $v_{t,i}=\frac 13$ as in \truth.
    \end{itemize}
    Thus overall, when adopting strategy $\pi_i$, agent $i$ reports $v_{t,i}=1$ w.p. $p+(1-p)\epsilon'=p+\epsilon$.

    As a side note, the resource allocation setup defined in \Cref{sec:setup} allows the central planner to ask further information from the agents than only their utility report (see Step~\ref{item:additional_questions}). To these questions at round $t\in[T]$, $\pi$ answers identically as what the PBE strategy $\truth$ would have replied had their true utility sequence been the reported $(v_{\tau, i})_{\tau\leq t}$, and for the same public history. 
    In summary, the strategy $\pi_i$ corresponds to using the PBE strategy $\truth$ for the pretended utilities $(v_{t,i})_{t\in[T]}$. Note that by construction, they form an i.i.d. sequence of distribution $\tilde {\mathcal U}_i\sim\frac{1}{3} + \frac{2}{3}\text{Ber}(p+\epsilon)$.
    
    Importantly, the allocation process is identical under the following two scenarios:
    \begin{enumerate}
        \item[S1.] agent $i$ had true utilities $\bm u_i=(u_{t,i})_{t\in [T]}$ but used strategy $\pi_i$ to report $\bm v_i=(v_{t,i})_{t\in [T]}$, while all other agents follow $\truth_{-i}$ (this corresponds to expectation operator $\E_{\bm u_i,\pi_i\circ \truth_{-i}}$), and
        \item[S2.] agent $i$ had true utilities as the aforementioned $\bm v_i=(v_{t,i})_{t\in [T]}$ and used the strategy $\truth_i$, while all other agents still follow $\truth_{-i}$ (we denote this case by $\E_{\bm v_i,\truth}$); we remark that \truth{} is not necessarily a PBE under the new configuration of distributions $(\mV_1,\tilde{\mV}_2,\mV_3,\ldots,\mV_K)$, since the revelation principle was only applied to the original $\{\mV_i\}_{i\in [K]}$,
    \end{enumerate}
    unless the following three conditions are simultaneously true for some round $t\in [T]$:
\begin{itemize}
    \item agent $i$ flips their report from $u_{t,i}=\frac 13$ to $v_{t,i}=1$ under strategy $\pi_{t,i}$ in round $t$, whose indicator we denote by $F_{t,i}:=\1[u_{t,i}=\frac 13,v_{t,i}=1]$;
    \item agent $i$ wins in round $t$, i.e., $i_t=i$; and
    \item the central planner audits $i$ in round $t$, namely $o_t=1$.
\end{itemize}

Only when all three conditions hold, i.e., $1=F_{t,i} \1[i_t=i] o_t$, the planner would observe $v_{t,i}\ne u_{t,i}$ in S1 but $v_{t,i}=u_{t,i}$ in S2. Hence S1 and S2 collide if $0=\bigvee_{t\in [T]} F_{t,i} \1[i_t=i]o_t$, where $\vee$ is logical or.
    
    Putting everything together, we obtain
    \begin{align}
        V_i^{\pi_i\circ \truth_{-i}}(\mH_1;\bm M) &= \E_{\bm u_i,\pi_i\circ \truth_{-i}} \sqb{\sum_{t=1}^T u_{t,i} \1[i_t=i] } \notag \\
        &\geq \E_{\bm v_i,\truth} \sqb{ \1\left [0=\bigvee_{t\in [T]} F_{t,i}\1[i_t=i]o_t\right ]\sum_{t=1}^T u_{t,i} \1[i_t=i] } \notag  \\
        &\geq \E_{\bm v_i,\truth} \sqb{\sum_{t=1}^T u_{t,i} \1[i_t=i] } - T\E_{\bm v_i,\truth}\sqb{\sum_{t=1}^T F_{t,i}\1[i_t=i]o_t}. \label{eq:bound_V_i_deviation_1}
    \end{align}

    For the first term in the RHS of \Cref{eq:bound_V_i_deviation_1}, we lower bound it as follows:
    \begin{align}
        \E_{\bm v_i,\truth} \sqb{\sum_{t=1}^T u_{t,i} \1[i_t=i] } &\geq \E_{\bm v_i,\truth} \sqb{\sum_{t=1}^T u_{t,i} \1[v_{t,i}=1] - \1[v_{t,i}=1]\1[i_t\neq i] } \notag \\
        &=Tp +\frac{\epsilon T}{3}  - \E_{\bm v_i,\truth} \sqb{\sum_{t=1}^T  \1[v_{t,i}=1]\1[i_t\neq i] }. \label{eq:bound_V_i_deviation_2}
    \end{align}
    To bound the second term in the RHS of \Cref{eq:bound_V_i_deviation_1}, note that in a round $t\in[T]$ in which the central planner allocated the resource to $i_t=i$, when the central planner makes the decision to audit agent $i$ or not, they only have access to the past history $\mH_t$, the current-round reports $\bm v_t$, and the winner $i_t$. Conditioning on these, the value of $u_{t,i}$ is in fact only dependent on the report $v_{t,i}$ as
    \begin{equation*}
        \Pr_{\bm v_i,\truth}\set{ u_{t,i}=\frac{1}{3} \middle \vert \mH_t,\bm v_t, i_t, o_t} = \Pr\set{ u_{t,i}=\frac{1}{3} \middle \vert  v_{t,i}} = \begin{cases}
            1 & v_{t,i}=\frac{1}{3},\\
            \frac{\epsilon}{p}  & v_{t,i}=1.
        \end{cases}
    \end{equation*}
    In the last inequality, we used the construction of the variable $v_{t,i}$. 
    As a result, 
    \begin{align}
         \E_{\bm v_i,\truth}\sqb{\sum_{t=1}^T F_{t,i}\1[i_t=i]o_t} &= \E_{\bm v_i,\truth}\sqb{\sum_{t=1}^T \1[v_{t,i}=1,i_t=i]o_t \1\sqb{u_{t,i}=\frac 13}} \notag \\
         &= \frac{\epsilon}{p} \E_{\bm v_i,\truth}\sqb{\sum_{t=1}^T \1[v_{t,i}=1,i_t=i]o_t}. \label{eq:bounding_nb_checks_for_catch}
    \end{align}
    
    In the remaining, we compare the expectation operators $\E_{\bm u_i,\truth}$ and $\E_{\bm v_i,\truth}$ which only differ in the fact that the true utility of agent $i$ was sampled i.i.d. following $\mathcal U_i$ and $\tilde{\mathcal U}_i$ respectively, that is (notice the difference between $\bm u\sim {{\mathcal U}_i^{\otimes T}}$ and $\bm u\sim {\tilde{\mathcal U}_i^{\otimes T}}$),
    \begin{align}
    \E_{\bm u_i,\truth}[X] &= \sum_{\bm z\in\set{\frac{1}{3},1}^T} \Pr_{\bm u\sim {{\mathcal U}_i^{\otimes T}}} \set{\bm u = \bm z} \E_{\truth}[X\mid \text{agent $i$ has true utility }\bm z],  \\
    \E_{\bm v_i,\truth}[X] &= \sum_{\bm z\in\set{\frac{1}{3},1}^T} \Pr_{\bm u\sim {\tilde{\mathcal U}_i^{\otimes T}}} \set{\bm u = \bm z} \E_{\truth}[X\mid \text{agent $i$ has true utility }\bm z].\label{eq:comparing_expectations}
    \end{align}
    For any fixed $\bm z \in \set{\frac 13,1}^T$, denoting by $k_{\bm z}$ the number of components equal to $1$ within $\bm z$, we have
    
     \begin{align*}
         \log \frac{\Pr_{\bm u\sim \tilde{\mathcal U}_i^{\otimes T}} \set{\bm u = \bm z} }{\Pr_{\bm u\sim {\mathcal U}_i^{\otimes T}} \set{\bm u = \bm z} } &= k_{\bm z} \log\paren{\frac{p+\epsilon}{p}} +(T-k_{\bm z})\log \paren{\frac{1-p}{1-p-\epsilon}}\\
         &= T\cdot \text{kl}(p+\epsilon\parallel p) + (k_{\bm z}-T(p+\epsilon)) \log\paren{\frac{(p+\epsilon)(1-p)}{p(1-p-\epsilon)}}\\
         &\overset{(a)}{\leq} 1 + (k_{\bm z}-T(p+\epsilon)) \frac{\epsilon}{p(1-p-\epsilon)} \overset{(b)}{\leq} 1 + \frac{2\epsilon}{p}   (k_{\bm z}-T(p+\epsilon)),
     \end{align*}
     where $\text{kl}(a\parallel b):=a\ln\frac{a}{b} + (1-a)\ln\frac{1-a}{1-b}$ is the binary KL divergence. In (a) we used the fact that because $p+\epsilon\leq \frac{1}{2}$, one has $\text{kl}(p+\epsilon \parallel p)\leq \frac{\epsilon^2}{p}$ \citep[Lemma 16]{blanchard2024tight} as well as $\log(1+x)\leq x$ for $x>-1$. In (b) we used $p+\epsilon\leq \frac{1}{2}$. Next, by the Hoeffding bound, we have,
     \begin{align}
        \Pr_{\bm z\sim \tilde{\mathcal U}_i^{\otimes T}}\set{ \frac{\Pr_{\bm u\sim \tilde{\mathcal U}_i^{\otimes T}} \set{\bm u = \bm z} }{\Pr_{\bm u\sim {\mathcal U}_i^{\otimes T}} \set{\bm u = \bm z} } \leq e^{5}} &\geq \Pr_{\bm z\sim \tilde{\mathcal U}_i^{\otimes T}}\set{ k_{\bm z} \leq T(p+\epsilon) + \frac{2p}{\epsilon} } \notag \\
        &\geq 1-\exp \paren{-\frac{8 p^2 T}{\epsilon^2}} = 1-\frac{1}{T^2}.\label{eq:large_proba_similar_likelihood}
    \end{align}
    
    We denote by $\mathcal E_i$ the corresponding event. Plugging our estimates in \Cref{eq:comparing_expectations} then gives (which is the relationship between expectation operators)
    \begin{equation}\label{eq:complete_comparison_expectations}
        \E_{\bm v_i,\truth} \leq e^{5} \E_{\bm u_i,\truth} + \E_{\bm v_i,\truth}\1[\neg \mathcal E_i].
    \end{equation}
    Using this identity twice within \Cref{eq:bound_V_i_deviation_2,eq:bounding_nb_checks_for_catch} then plugging the results within \Cref{eq:bound_V_i_deviation_1} yields
    \begin{align*}
        &V_i^{\pi_i\circ \truth_{-i}}(\mH_1;\bm M) \geq Tp + \frac{\epsilon T}{3} -\paren{\frac{\epsilon T^2}{p}+T}\Pr\set{\neg \mathcal E_i} \\
        &\quad \qquad - e^{5} \paren{ \frac{\epsilon T}{p} \E_{\bm u_i,\truth}\sqb{\sum_{t=1}^T \1[v_{t,i}=1,i_t=i]o_t} +  \E_{\bm u_i,\truth} \sqb{\sum_{t=1}^T  \1[v_{t,i}=1]\1[i_t\neq i] }}.
    \end{align*}
    
    By assumption \truth{} is a PBE and thus $V_i^{\pi_i\circ \truth_{-i}}(\mH_1;\bm M) \leq V_i^{\truth}(\mH_1;\bm M)$. Together with the lower bound on $V_i^{\pi_i\circ \truth_{-i}}(\mH_1;\bm M)$ and the upper bound on $V_i^{\truth}(\mH_1;\bm M)$ from \Cref{eq:upper_bound_truth},
    \begin{align*}
        0&\leq  V_i^{\truth}(\mH_1;\bm M) - V_i^{\pi_i\circ \truth_{-i}}(\mH_1;\bm M) \\
        &\overset{(a)}{\leq} e^{5}  \frac{\epsilon T}{p} \E_{\bm u_i,\truth}\sqb{\sum_{t=1}^T \1[v_{t,i}=1,i_t=i]o_t} + \frac{\epsilon}{p}+\frac{1}{T}  - \frac{\epsilon T}{3} \\
        &\quad + e^5 \E_{\bm u_i,\truth} \sqb{\sum_{t=1}^T \1\sqb{u_{t,i}=\frac{1}{3}}\1[i_t=i] + \1[v_{t,i}=1]\1[i_t\neq i]  }\\
        &\overset{(b)}{\leq} e^5\paren{\frac{\epsilon T}{p}\mB_T(\truth,\bm M) + 3\mR_T(\truth,\bm M)} + 2 - \frac{\epsilon T}{3}.
    \end{align*}
    In (a) we also used \Cref{eq:large_proba_similar_likelihood}, and in (b) we used \Cref{eq:lower_bound_regret} and the definition of $\mB_T(\truth,\bm M)$.

    As a result, if we have $B_T(\truth,\bm M) \leq \frac{p}{6e^5}$, the previous inequality implies
    \begin{equation*}
        \mR_T(\truth,\bm M) \geq \frac{\epsilon T}{3}-2\geq \frac{\epsilon T}{4} = \frac{1}{4}\sqrt{\frac{T}{\log T}},
    \end{equation*}
    for $T$ sufficiently large. This ends the proof.
\end{proof}

\end{document}